\documentclass[10pt]{amsart}
\usepackage{amssymb,amsfonts,latexsym,amscd}
\usepackage{amsmath}

\newcommand{\gF}{{\mathfrak F}}  
\newcommand{\gS}{{\mathfrak{S}}}

\newcommand{\gD}{{\mathfrak D}}
\newcommand{\cD}{{\mathcal{D}}} 

\newcommand{\cV}{{\mathcal{V}}} 

\newcommand{\Snn}{\Sigma_{1\, n}^{\ \, n+1}}
\newcommand{\Fldnn}{\gF_{\ell,2,n}^{\ \ \ \, n+1}}

\newcommand{\Fnn}{\gF_n^{n+1}}

\newcommand{{\adat}}{\mathrm{ad}_{A_t}}
\def\C{{\mathbb{C}}} 
\def\N{{\mathbb{N}}} 
\def\R{{\mathbb{R}}} 

\def\vp{\varphi}

\newcommand{\Aa}{{a_\epsilon}} 
\newcommand{\Ba}{b_{\ell, \epsilon}} 
\newcommand{\Ca}{c_{\ell, \epsilon}} 
\newcommand{\Ac}{{a^*_\epsilon}} 
\newcommand{\Bc}{b^*_{\ell, \epsilon}} 
\newcommand{\Cc}{c^*_{\ell, \epsilon}} 

\def\1{{\mathbf{1}}}
\def\la{\langle}
\def\ra{\rangle}

\renewcommand\d{\mathrm{d}}

\newcommand{\ind}{{(q_\ell, \bar q_\ell, r_\ell,\bar r_\ell)}}

\newcommand{\sig}{{\ \sigma}} 

\renewcommand\Im{\mathrm{Im}} 

\newtheorem{theorem}{Theorem}[section]

\newtheorem{lemma}[theorem]{Lemma}

\newtheorem{proposition}[theorem]{Proposition}
\newtheorem{remark}[theorem]{Remark}

\newtheorem{hypothesis}[theorem]{Hypothesis}

\begin{document}
\markboth{J.-M. Barbaroux, J.-C. Guillot} {Spectral theory for a
mathematical model of the weak interaction: The decay of the
intermediate vector bosons $\mathbf{W}^\pm$}

\title[Mathematical model of the weak interaction]{Spectral theory for a
mathematical model of the weak interaction: The decay of the
intermediate vector bosons $\textbf{\textit{W}}^\pm$. I.}

\author{J.-M. Barbaroux}

\address{
 Centre de Physique Th\'eorique, Luminy Case 907, 13288
 Marseille Cedex~9, France and D\'epartement de Math\'ematiques,
 Universit\'e du Sud Toulon-Var, 83957 La
 Garde Cedex, France\\
}

\email{barbarou@univ-tln.fr}

\author{J.-C. Guillot}

\address{
Centre de Math\'ematiques Appliqu\'ees, UMR 7641, \'Ecole
Polytechnique - C.N.R.S, 91128 Palaiseau Cedex, France
\\
}

\email{Jean-Claude.Guillot@polytechnique.edu}

\maketitle

\begin{abstract}
We consider a Hamiltonian with cutoffs describing the weak decay
of spin $1$ massive bosons into the full family of leptons. The
Hamiltonian is a self-adjoint operator in an appropriate Fock
space with a unique ground state. We prove a Mourre estimate and a
limiting absorption principle above the ground state energy and
below the first threshold for a sufficiently small coupling
constant. As a corollary, we prove absence of eigenvalues and
absolute continuity of the energy spectrum in the same spectral
interval.
\end{abstract}

\section{Introduction}
\setcounter{equation}{0}

In this article, we consider a mathematical model of the weak
interaction as patterned according to the Standard Model in
Quantum Field Theory (see \cite{GreinerMuller1989,
WeinbergII2005}). We choose the example of the weak decay of the
intermediate vector bosons $W^\pm$ into the full family of
leptons.

The mathematical framework involves fermionic Fock spaces for the
leptons and bosonic Fock spaces for the vector bosons. The
interaction is described in terms of annihilation and creation
operators together with kernels which are square integrable with
respect to momenta. The total Hamiltonian, which is the sum of the
free energy of the particles and antiparticles and of the
interaction, is a self-adjoint operator in the Fock space for the
leptons and the vector bosons and it has an unique ground state in
the Fock space for a sufficiently small coupling constant.

The weak interaction is one of the four fundamental interactions
known up to now. But the weak interaction is the only one which
does not generate bound states. As it is well known it is not the
case for the strong, electromagnetic and gravitational
interactions. Thus we are expecting that the spectrum of the
Hamiltonian associated with every model of weak decays is
absolutely continuous above the energy of the ground state and
this article is a first step towards a proof of such a statement.
Moreover a scattering theory has to be established for every such
Hamiltonian.

In this paper we establish a Mourre estimate and a limiting
absorption principle for any spectral interval above the energy of
the ground state and below the mass of the electron for a small
coupling constant.

Our study of the spectral analysis of the total Hamiltonian is
based on the conjugate operator method with a self-adjoint
conjugate operator. The methods used in this article are taken
largely from \cite{Bachetal2006} and \cite{Frohlichetal2008} and
are based on \cite{Amreinetal1996} and \cite{Sahbani1997}. Some of
the results of this article has been announced in
\cite{Barbarouxetal2009}.

For other applications of the conjugate operator method see
\cite{Ammari2004, Bachetal1999S, Bachetal1999P,
DerezinskiGerard1999, DerezinskiJacsik2001, Faupin2009,
Frohlichetal2004, Georgescuetal2004, Georgescuetal2004C,
Golenia2008, HubnerSpohn1995S, Skibsted1998}.


For related results about models in Quantum Field Theory see
\cite{Barbarouxetal2004} and \cite{Takaesu2008} in the case of the
Quantum Electrodynamics and \cite{Amouretal2007} in the case of
the weak interaction.

The paper is organized as follows. In section~\ref{the-model}, we
give a precise definition of the model we consider. In
section~\ref{section3}, we state our main results and in the
following sections, together with the appendix, detailed proofs of
the results are given.

\noindent\textbf{Acknowledgments.} One of us (J.-C. G) wishes to
thank Laurent Amour and Beno\^\i t Gr\'ebert for helpful
discussions. The authors also thank Walter Aschbacher for valuable
remarks. The work was done partially while J.M.-B. was visiting
the Institute for Mathematical Sciences, National University of
Singapore in 2008. The visit was supported by the Institute.

\section{The model}\label{the-model}
\setcounter{equation}{0}

The weak decay of the intermediate bosons $W^+$ and $W^-$ involves
the full family of leptons together with the bosons themselves,
according to the Standard Model (see \cite[Formula
(4.139)]{GreinerMuller1989} and \cite{WeinbergII2005}).

The full family of leptons involves the electron $e^-$ and the
positron $e^+$, together with the associated neutrino $\nu_e$ and
antineutrino $\bar\nu_e$, the muons $\mu^-$ and $\mu^+$ together
with the associated neutrino $\nu_\mu$ and antineutrino
$\bar\nu_\mu$ and the tau leptons $\tau^-$ and $\tau^+$ together
with the associated neutrino $\nu_\tau$ and antineutrino
$\bar\nu_\tau$.

It follows from the Standard Model that neutrinos and
antineutrinos are massless particles. Neutrinos are left-handed,
i.e., neutrinos have helicity $-1/2$ and antineutrinos are right
handed, i.e., antineutrinos have helicity $+1/2$.

In what follows, the mathematical model for the weak decay of the
vector bosons $W^+$ and $W^-$ that we propose is based on the
Standard Model, but we adopt a slightly more general point of view
because we suppose that neutrinos and antineutrinos are both
massless particles with helicity $\pm 1/2$. We recover the
physical situation as a particular case. We could also consider a
model with massive neutrinos and antineutrinos built upon the
Standard Model with neutrino mixing \cite{Srednicki2007}.

Let us sketch how we define a mathematical model for the weak
decay of the vector bosons $W^\pm$ into the full family of
leptons.

The energy of the free leptons and bosons is a self-adjoint
operator in the corresponding Fock space (see below) and the main
problem is associated with the interaction between the bosons and
the leptons. Let us consider only the interaction between the
bosons and the electrons, the positrons and the corresponding
neutrinos and antineutrinos. Other cases are strictly similar. In
the Schr\"odinger representation the interaction is given by (see
\cite[p159, (4.139)]{GreinerMuller1989} and \cite[p308,
(21.3.20)]{WeinbergII2005})
\begin{equation}\label{eq:2.1}
  I = \int \d^3 \!x\, \overline{\Psi_e}(x) \gamma^\alpha
  (1-\gamma_5)\Psi_{\nu_e}(x) W_\alpha(x) + \int \d^3\!x\,
  \overline{\Psi_{\nu_e}}(x) \gamma^\alpha (1-\gamma_5)
  \Psi_e(x) W_\alpha(x)^*\ ,
\end{equation}
where $\gamma^\alpha$, $\alpha=0,1,2,3$ and $\gamma_5$ are the
Dirac matrices and $\Psi_.(x)$ and $\overline{\Psi_.}(x)$ are the
Dirac fields for $e_-$, $e_+$, $\nu_e$ and $\bar\nu_e$.

We have
\begin{equation}\nonumber
\begin{split}
 & \Psi_e(x) = \big(\frac{1}{2\pi}\big)^\frac32 \sum_{s=\pm\frac12}
 \int \d^3\!p\, (b_{e,+}(p,s) \frac{u(p,s)}{\sqrt{p_0}}
 \mathrm{e}^{ip.x}
 + b_{e,-}^* (p,s) \frac{v(p,s)}{\sqrt{p_0}} \mathrm{e}^{- i p .
 x})\ ,\\
 & \overline{\Psi_e}(x) = \Psi_e(x)^\dagger \gamma^0\ .
\end{split}
\end{equation}
Here $p_0 = (|p|^2 + m_e^2)^\frac12$ where $m_e>0$ is the mass of
the electron and $u(p,s)$ and $v(p,s)$ are the normalized
solutions to the Dirac equation (see
\cite[Appendix]{GreinerMuller1989}).

The operators $b_{e,+}(p,s)$ and $b_{e,+}^*(p,s)$ (respectively
$b_{e,-}(p,s)$ and $b_{e,-}^*(p,s)$) are the annihilation and
creation operators for the electrons (respectively the positrons)
satisfying the anticommutation relations (see below).

Similarly we define $\Psi_{\nu_e}(x)$ and
$\overline{\Psi_{\nu_e}}(x)$ by substituting the operators
$c_{\nu_e,\pm}(p,s)$ and $c_{\nu_e,\pm }^*(p,s)$ for
$b_{e,\pm}(p,s)$ and $b_{e,\pm}^*(p,s)$ with $p_0= |p|$. The
operators $c_{\nu_e,+}(p,s)$ and $c_{\nu_e,+}^*(p,s)$
(respectively $c_{\nu_e,-}(p,s)$ and $c_{\nu_e,-}^*(p,s)$) are the
annihilation and creation operators for the neutrinos associated
with the electrons (respectively the antineutrinos).

For the $W_\alpha$ fields we have (see \cite[\S
5.3]{WeinbergI2005}).
\begin{equation}\nonumber
 W_\alpha(x) = \big(\frac{1}{2\pi}\big)^\frac32 \sum_{\lambda=-1,0,1}
 \int \frac{\d^3\!k}{\sqrt{2 k_0}} (\epsilon_\alpha(k,\lambda)
 a_+(k,\lambda)
 \mathrm{e}^{i k.x}
 + \epsilon_\alpha^*(k,\lambda) a_-^*(k,\lambda)
 \mathrm{e}^{- i k.x})\ .
\end{equation}
Here $k_0 = (|k|^2 + m_W^2)^\frac12$ where $m_W>0$ is the mass of
the bosons $W^\pm$. $W^+$ is the antiparticule of $W^-$. The
operators $a_+(k,\lambda)$ and $a_+^*(k,\lambda)$ (respectively
$a_-(k,\lambda)$ and $a_-^*(k,\lambda)$) are the annihilation and
creation operators for the bosons $W^-$ (respectively $W^+$)
satisfying the canonical commutation relations. The vectors
$\epsilon_\alpha(k,\lambda)$ are the polarizations of the massive
spin~1 bosons $W^\pm$ (see \cite[Section~5.2]{WeinbergI2005}).

The interaction \eqref{eq:2.1} is a formal operator and, in order
to get a well defined operator in the Fock space, one way is to
adapt what Glimm and Jaffe have done in the case of the Yukawa
Hamiltonian (see \cite{GlimmJaffe1977}). For that sake, we have to
introduce a spatial cutoff $g(x)$ such that $g\in L^1(\R^3)$,
together with momentum cutoffs $\chi(p)$ and $\rho(k)$ for the
Dirac fields and the $W_\mu$ fields respectively.

Thus when one develops the interaction $I$ with respect to
products of creation and annihilation operators, one gets a finite
sum of terms associated with kernels of the form
\begin{equation}\nonumber
 \chi(p_1) \, \chi(p_2)\, \rho(k)\, \hat g(p_1+p_2-k)\ ,
\end{equation}
where $\hat g$ is the Fourier transform of $g$. These kernels are
square integrable.

In what follows, we consider a model involving terms of the above
form but with more general square integrable kernels.

We follow the convention described in
\cite[section~4.1]{WeinbergI2005} that we quote: ``The
state-vector will be taken to be symmetric under interchange of
any bosons with each other, or any bosons with any fermions, and
antisymmetric with respect to interchange of any two fermions with
each other, in all cases, wether the particles are of the same
species or not''. Thus, as it follows from section~4.2 of
\cite{WeinbergI2005}, fermionic creation and annihilation
operators of different species of leptons will always anticommute.

Concerning our notations, from now on, $\ell\in\{1,2,3\}$ denotes
each species of leptons. $\ell=1$ denotes the electron $e^-$ the
positron $e^+$ and the neutrinos $\nu_e$, $\bar\nu_e$. $\ell=2$
denotes the muons $\mu^-$, $\mu^+$ and the neutrinos $\nu_\mu$ and
$\bar\nu_\mu$, and $\ell=3$ denotes the tau-leptons and the
neutrinos $\nu_\tau$ and $\bar\nu_\tau$.

Let $\xi_1=(p_1,\ s_1)$ be the quantum variables of a massive
lepton, where $p_1\in\R^3$ and $s_1\in\{-1/2,\ 1/2\}$ is the spin
polarization of particles and antiparticles. Let $\xi_2=(p_2,\
s_2)$ be the quantum variables of a massless lepton where
$p_2\in\R^3$ and $s_2\in\{-1/2,\ 1/2\}$ is the helicity of
particles and antiparticles and, finally, let $\xi_3=(k,\
\lambda)$ be the quantum variables of the spin $1$ bosons $W^+$
and $W^-$ where $k\in\R^3$ and $\lambda\in\{-1,\ 0,\ 1\}$ is the
polarization of the vector bosons (see
\cite[section~5]{WeinbergI2005}). We set $\Sigma_1 =
\R^3\times\{-1/2,\ 1/2\}$ for the leptons and $\Sigma_2 =
\R^3\times\{-1,\ 0,\ 1\}$ for the bosons. Thus $L^2(\Sigma_1)$ is
the Hilbert space of each lepton and $L^2(\Sigma_2)$ is the
Hilbert space of each boson. The scalar product in
$L^2(\Sigma_j)$, $j=1,2$ is defined by
\begin{equation}\label{2.1}
  ( f,\ g ) = \int_{\Sigma_j} \overline{f(\xi)} g(\xi) \d
  \xi,\quad j=1,2\ .
\end{equation}
Here
\begin{equation}\nonumber
  \int_{\Sigma_1} \d \xi = \sum_{s=+\frac12, -\frac12} \int
  \d  p\quad \mbox{and}\quad
  \int_{\Sigma_2} \d \xi = \sum_{\lambda=0,1,-1} \int \d
  k,\quad (p,k\in\R^3)\ .
\end{equation}

The Hilbert space for the weak decay of the vector bosons $W^+$
and $W^-$ is the Fock space for leptons and bosons that we now
describe.

Let $\gS$ be any separable Hilbert space. Let $\otimes_a^n\gS$
(resp. $\otimes _s^n\gS$) denote the antisymmetric (resp.
symmetric) $n$-th tensor power of $\gS$. The fermionic (resp.
bosonic) Fock space over $\gS$, denoted by $\gF_a(\gS)$ (resp.
$\gF_s(\gS))$, is the direct sum
\begin{equation}\label{2.3}
 \gF_a(\gS) = \bigoplus_{n=0}^\infty \bigotimes_a^n \gS\quad
 (\mbox{resp. }
  \gF_s(\gS) = \bigoplus_{n=0}^\infty \bigotimes_s^n \gS)\ ,
\end{equation}
where $\otimes_a^0 \gS = \otimes_s^0\gS \equiv\C$. The state
$\Omega = (1,0,0,\ldots,0,\ldots)$ denotes the vacuum state in
$\gF_a(\gS)$ and in $\gF_s(\gS)$.

For every $\ell$, $\gF_\ell$ is the fermionic Fock space for the
corresponding species of leptons including the massive particle
and antiparticle together with the associated neutrino and
antineutrino, i.e.,
\begin{equation}\label{eq:2.4}
\gF_\ell = \bigotimes^4 \gF_a(L^2(\Sigma_1))\, \quad \ell=1,2,3 \
.
\end{equation}
We have
\begin{equation}\label{eq:2.5}
 \gF_\ell = \bigoplus_{q_\ell \geq 0, \bar q_\ell\geq
 0, r_\ell\geq 0, \bar r_\ell\geq0} \gF_\ell^{(q_\ell,\bar q_\ell,
 r_\ell,\bar r_\ell)}\ ,
\end{equation}
with
\begin{equation}\label{eq:2.6}
 \gF_\ell^{(q_\ell,\bar q_\ell, r_\ell, \bar r_\ell)}
 = (\otimes _a^{q_\ell} L^2(\Sigma_1))\otimes
 (\otimes _a^{\bar q_\ell} L^2(\Sigma_1))\otimes
 (\otimes _a^{r_\ell} L^2(\Sigma_1))\otimes
 (\otimes _a^{\bar r_\ell} L^2(\Sigma_1))\ .
\end{equation}
Here $q_\ell$ (resp. $\bar q_\ell$) is the number of massive
particle (resp. antiparticles) and $r_\ell$ (resp. $\bar r_\ell$)
is the number of neutrinos (resp. antineutrinos). The vector
$\Omega_\ell$ is the associated vacuum state. The fermionic Fock
space denoted by $\gF_L$ for the leptons is then
\begin{equation}\label{2.7}
 \gF_L = \otimes_{\ell=1}^3 \gF_\ell\ ,
\end{equation}
and $\Omega_L = \otimes_{\ell=1}^3 \Omega_\ell$ is the vacuum
state.

The bosonic Fock space for the vector bosons $W^+$ and $W^-$,
denoted by $\gF_W$, is then
\begin{equation}\label{2.8}
 \gF_W = \gF_s(L^2(\Sigma_2))\otimes \gF_s(L^2(\Sigma_2)) \simeq
 \gF_s(L^2(\Sigma_2) \oplus L^2(\Sigma_2))\ .
\end{equation}
We have
\begin{equation}\nonumber
 \gF_W = \bigoplus_{t\geq 0, \bar t\geq 0} \gF_W^{(t,\bar t)}\ ,
\end{equation}
where $\gF_W^{(t,\bar t)} = (\otimes_s^t
L^2(\Sigma_2))\otimes(\otimes_s^{\bar t} L^2(\Sigma_2))$. Here $t$
(resp. $\bar t$) is the number of bosons $W^-$ (resp. $W^+$). The
vector $\Omega_W$ is the corresponding vacuum.

The Fock space for the weak decay of the vector bosons $W^+$ and
$W^-$, denoted by $\gF$, is thus
\begin{equation}\nonumber
 \gF = \gF_L \otimes\gF_W
\end{equation}
and $\Omega = \Omega_L \otimes \Omega_W$ is the vacuum state.

For every $\ell\in\{1,2,3\}$ let $\gD_\ell$ denote the set of
smooth vectors $\psi_\ell\in\gF_\ell$ for which $\psi_\ell^{\ind}$
has a compact support and $\psi_\ell^\ind = 0$ for all but
finitely many $\ind$. Let
\begin{equation}\nonumber
  \gD_L = \widehat\bigotimes_{\ell=1}^3 \gD_\ell\ .
\end{equation}
Here $\hat\otimes$ is the algebraic tensor product.

Let $\gD_W$ denote the set of smooth vectors $\phi\in\gF_W$ for
which $\phi^{(t,\bar t)}$ has a compact support and
$\phi^{(t,\bar t)}=0$ for all but finitely many $(t,\bar t)$.

Let
\begin{equation}\nonumber
\gD = \gD_L\hat\otimes\, \gD_W\ .
\end{equation}
The set $\gD$ is dense in $\gF$.

Let $A_\ell$ be a self-adjoint operator in $\gF_\ell$ such that
$\gD_\ell$ is a core for $A_\ell$. Its extension to $\gF_L$ is, by
definition, the closure in $\gF_L$ of the operator
$A_1\otimes\1_2\otimes\1_3$ with domain $\gD_L$ when $\ell=1$, of
the operator $\1_1\otimes A_2 \otimes\1_3$ with domain $\gD_L$
when $\ell=2$, and of the operator $\1_1 \otimes\1_2\otimes A_3$
with domain $\gD_L$ when $\ell=3$. Here $\1_\ell$ is the operator
identity on $\gF_\ell$.

The extension of $A_\ell$ to $\gF_L$ is a self-adjoint operator
for which $\gD_L$ is a core and it can be extended to $\gF$. The
extension of $A_\ell$ to $\gF$ is, by definition, the closure in
$\gF$ of the operator $\tilde A_\ell\otimes \1_W$ with domain
$\gD$, where $\tilde A_\ell$ is the extension of $A_\ell$ to
$\gF_L$. The extension of $A_\ell$ to $\gF$ is a self-adjoint
operator for which $\gD$ is a core.

Let $B$ be a self-adjoint operator in $\gF_W$ for which $\gD_W$ is
a core. The extension of the self-adjoint operator $A_\ell\otimes
B$ is, by definition, the closure in $\gF$ of the operator
$A_1\otimes\1_2\otimes\1_3\otimes B$ with domain $\gD$ when
$\ell=1$, of the operator $\1_1\otimes A_2\otimes\1_3\otimes B$
with domain $\gD$ when $\ell=2$, and of the operator
$\1_1\otimes\1_2\otimes A_3\otimes B$ with domain $\gD$ when
$\ell=3$. The extension of $A_\ell\otimes B$ to $\gF$ is a
self-adjoint operator for which $\gD$ is a core.

We now define the creation and annihilation operators. For each
$\ell=1,2,3$, $\Ba(\xi_1)$ (resp. $\Bc(\xi_1)$) is the
annihilation (resp. creation) operator for the corresponding
species of massive particle when $\epsilon=+$ and for the
corresponding species of massive antiparticle when $\epsilon=-$.
Similarly, for each $\ell=1,2,3$, $\Ca(\xi_2)$ (resp.
$\Cc(\xi_2)$) is the annihilation (resp. creation) operator for
the corresponding species of neutrino when $\epsilon=+$ and for
the corresponding species of antineutrino when $\epsilon=-$. The
operator $\Aa(\xi_3)$ (resp. $\Ac(\xi_3)$) is the annihilation
(resp. creation) operator for the boson $W^-$ when $\epsilon=+$
and for the boson $W^+$ when $\epsilon=-$.

Let $\Psi\in\gD$ be such that
\begin{equation}\nonumber
 \Psi = \left(\Psi^{(Q)}\right)_Q\ ,
\end{equation}
with $Q= \Big(  \ind_{\ell=1,2,3},\,(t,\bar t)\Big)$, and
\begin{equation}\nonumber
 \Psi^{(Q)} = \left(\otimes_{\ell=1}^3 \Psi^\ind\right)\otimes
 \varphi^{(t,\bar t)}\ ,
\end{equation}
where $(q_\ell, \bar q_\ell, r_\ell, \bar r_\ell, t, \bar t)\in
\N^6$. Here, $(\Psi^\ind)_{q_\ell\geq 0, \bar q_\ell\geq 0,
r_\ell\geq 0, \bar r_\ell\geq 0}\in\gD_\ell$, and
$(\varphi^{(t,\bar t)})_{t\geq 0, \bar t\geq 0}\in\gD_W$.

Let
\begin{equation}\nonumber
\begin{split}
  Q_{\ell,+} & = \Big( (q_{\ell'}, \bar q_{\ell'},
  r_{\ell'}, \bar r_{\ell'} )_{\ell'<\ell},\, (q_{\ell}+1, \bar q_\ell,
  r_\ell, \bar r_{\ell}),\,
  (q_{\ell'}, \bar q_{\ell'},
  r_{\ell'}, \bar r_{\ell'})_{\ell'>\ell},\, (t,\bar t)\Big)\ , \\
  Q_{\ell,-} & = \Big((q_{\ell'}, \bar q_{\ell'},
  r_{\ell'}, \bar r_{\ell'} )_{\ell'<\ell} ,\, (q_{\ell}, \bar q_\ell +1,
  r_\ell, \bar r_{\ell}),\,
  (q_{\ell'}, \bar q_{\ell'},
  r_{\ell'}, \bar r_{\ell'})_{\ell'>\ell},\,  (t,\bar t)\Big)\ , \\
  \tilde Q_{\ell,+} & = \Big( (q_{\ell'}, \bar q_{\ell'},
  r_{\ell'}, \bar r_{\ell'} )_{\ell'<\ell},\,  (q_{\ell}, \bar q_\ell,
  r_\ell +1, \bar r_{\ell}),\,
  (q_{\ell'}, \bar q_{\ell'},
  r_{\ell'}, \bar r_{\ell'})_{\ell'>\ell},\, (t,\bar t)\Big)\ , \\
  \tilde Q_{\ell,-} & = \Big( (q_{\ell'}, \bar q_{\ell'},
  r_{\ell'}, \bar r_{\ell'} )_{\ell'<\ell},\,  (q_{\ell}, \bar q_\ell,
  r_\ell, \bar r_{\ell}+1 ),\,
 (q_{\ell'}, \bar q_{\ell'},
  r_{\ell'}, \bar r_{\ell'})_{\ell'>\ell},\, (t,\bar t)\Big)\ ,
\end{split}
\end{equation}
and
\begin{equation}\nonumber
\begin{split}
  Q_{b,+} & = \Big( (q_{\ell}, \bar q_{\ell},
  r_{\ell}, \bar r_{\ell} )_{\ell=1,2,3},\,  (t+1,\bar t)\Big)\ , \\
  Q_{b,-} & = \Big( (q_{\ell}, \bar q_{\ell},
  r_{\ell}, \bar r_{\ell} )_{\ell=1,2,3},\, (t,\bar t +1)\Big)\ .
\end{split}
\end{equation}
We define
\begin{equation}\nonumber
\begin{split}
  & (b_{\ell,+}(\xi_1) \Psi)^{(Q)} (\,.\,;\, \xi_1^{(1)},
  \xi_1^{(2)},\ldots,\xi_1^{(q_\ell)}; \, .\, ) \\
  & = \sqrt{q_\ell+1}\, \Pi_{\ell'<\ell}\ (-1)^{q_{\ell'} + \bar
  q_{\ell'}} \Psi^{(Q_{\ell,+})} (\,.\, ; \xi_1,\xi_1^{(1)},
  \xi_1^{(2)},\ldots,\xi_1^{(q_\ell)}; \, .\, )\\
  & (b_{\ell,-}(\xi_1) \Psi)^{(Q)} (\,.\,;\, \xi_1^{(1)},
  \xi_1^{(2)},\ldots,\xi_1^{(\bar q_\ell)}; \, .\, ) \\
  & = \sqrt{\bar q_\ell+1}\,(-1)^{q_\ell} \Pi_{\ell'<\ell}\ (-1)^{q_{\ell'}
  + \bar q_{\ell'}} \Psi^{(Q_{\ell,-})} (\,.\, ; \xi_1,\xi_1^{(1)},
  \xi_1^{(2)},\ldots,\xi_1^{(\bar q_\ell)}; \, .\, )\ ,
\end{split}
\end{equation}
\begin{equation}\nonumber
\begin{split}
  & (c_{\ell,+}(\xi_2) \Psi)^{(Q)} (\,.\,;\, \xi_2^{(1)},
  \xi_2^{(2)},\ldots,\xi_2^{(r_\ell)}; \, .\, ) \\
  & = \sqrt{r_\ell+1}\, (-1)^{q_\ell +\bar q_\ell}\Pi_{\ell'<\ell}\ (-1)^{q_{\ell'} + \bar
  q_{\ell'}+ r_{\ell'} + \bar r_{\ell'}}
  \Psi^{(\tilde Q_{\ell,+})} (\,.\, ; \xi_2,\xi_2^{(1)},
  \xi_2^{(2)},\ldots,\xi_2^{(r_\ell)}; \, .\, )\\
  & (c_{\ell,-}(\xi_2) \Psi)^{(Q)} (\,.\,;\, \xi_2^{(1)},
  \xi_2^{(2)},\ldots,\xi_2^{(\bar r_\ell)}; \, .\, ) \\
  & = \sqrt{\bar r_\ell+1}\,(-1)^{q_\ell +\bar q_\ell+r_\ell} \Pi_{\ell'<\ell}\
  (-1)^{q_{\ell'} + \bar q_{\ell'} + r_{\ell'} + \bar r_{\ell'}}
  \Psi^{(\tilde Q_{\ell,-})} (\,.\, ; \xi_2,\xi_2^{(1)},
  \xi_2^{(2)},\ldots,\xi_2^{(\bar r_\ell)}; \, .\, )\ ,
\end{split}
\end{equation}
and
\begin{equation}\nonumber
\begin{split}
 & (a_+(\xi_3) \Psi)^{(Q)} (\,.\, ; \, \xi_3^{(1)}, \xi_3^{(2)},
 \ldots, \xi_3^{(t)};\, .\,) \\
 & = \sqrt{t+1} \Psi^{(Q_{b,+})}
 (\,.\, ; \, \xi_3, \xi_3^{(1)}, \xi_3^{(2)},
 \ldots, \xi_3^{(t)};\, .\,)\ ,\\
 & (a_-(\xi_3) \Psi)^{(Q)} (\,.\, ; \, \xi_3^{(1)}, \xi_3^{(2)},
 \ldots, \xi_3^{(\bar t)};\, .\,) \\
 & = \sqrt{\bar t+1} \Psi^{(Q_{b,-})}
 (\,.\, ; \, \xi_3, \xi_3^{(1)}, \xi_3^{(2)},
 \ldots, \xi_3^{(\bar t)};\, .\,)\ .
\end{split}
\end{equation}

As usual, $b^*_{\ell,\epsilon}(\xi_1)$ (resp. $c^*_{\ell,
\epsilon}(\xi_2)$) is the formal adjoint of
$b_{\ell,\epsilon}(\xi_1)$ (resp. $c_{\ell, \epsilon}(\xi_2)$).
For example, we have
\begin{equation}\nonumber
\begin{split}
  & (b^*_{\ell,\epsilon}(\xi_1) \Psi)^{(Q_{\ell,+})} (\,.\, ;
  \xi_1^{(1)}, \xi_1^{(2)}, \ldots, \xi_1^{(q_\ell)},
  \xi_1^{(q_\ell+1)};\, . \, ) \\
  & = \frac{1}{\sqrt{q_\ell+1}} \prod_{\ell' < \ell}
  (-1)^{q_{\ell'} + \bar q_{\ell'}}\\
  & \sum_{i=1}^{q_\ell+1} (-1)^{i+1} \delta (\xi_1 - \xi_1^{(i)})
  \Psi^{(Q)} (\, .\, ; \xi_1^{(1)}, \xi_1^{(2)}, \ldots,
  \widehat{\xi_1^{(i)}},\ldots, \xi_1^{(q_\ell+1)};\, .\, )\ ,
\end{split}
\end{equation}
where $\widehat . $ denotes that the $i$-th variable has to be
omitted, and $\delta(\xi_1-\xi_1^{(i)}) = \delta_{s_1 s_1^{(i)}}
\delta(p_1 - p_1^{(i)})$. The operator $a^*_\epsilon(\xi_3)$ is
the formal adjoint of $a_\epsilon(\xi_3)$ and we have
\begin{equation}\nonumber
\begin{split}
  & (a^*_+(\xi_3) \Psi)^{(Q_{b,+})} (\, .\, ; \xi_3^{(1)},
  \xi_3^{(2)}, \ldots, \xi_3^{(t+1)} ; .) \\
  & = \frac{1}{\sqrt{t+1}} \sum_{i=1}^{t+1} \delta(\xi_3
  -\xi_3^{(i)}) \Psi^{(Q)} (\, .\, ; \xi_3^{(1)}, \ldots,
  \widehat{\xi_3^{(i)}}, \ldots, \xi_3^{(t+1)};\, .\, )
\end{split}
\end{equation}
where $\delta(\xi_3 - \xi_3^{(i)}) = \delta_{\lambda
\lambda^{(i)}} \delta(k - k^{(i)})$.

The following canonical anticommutation and commutation relations
hold.
\begin{equation}\nonumber
\begin{split}
 &\{ b_{\ell, \epsilon}(\xi_1), b^*_{\ell', \epsilon'}(\xi_1')\} =
 \delta_{\ell \ell'}\delta_{\epsilon \epsilon'} \delta(\xi_1 - \xi_1') \ ,\\
 &\{ c_{\ell, \epsilon}(\xi_2), c^*_{\ell', \epsilon'}(\xi_2')\} =
 \delta_{\ell \ell'}\delta_{\epsilon \epsilon'} \delta(\xi_2 - \xi_2')\ , \\
 &[ a_{\epsilon}(\xi_3), a^*_{\epsilon'}(\xi_3')] =
 \delta_{\epsilon \epsilon'} \delta(\xi_3 - \xi_3') \ ,\\
 &\{ b_{\ell, \epsilon}(\xi_1), b_{\ell', \epsilon'}(\xi_1')\}
 = \{ c_{\ell, \epsilon}(\xi_2), c_{\ell', \epsilon'}
 (\xi_2')\} =0\ ,\\
 & [ a_{\epsilon}(\xi_3), a_{\epsilon'}(\xi_3') ] = 0\ ,\\
 &\{ b_{\ell, \epsilon}(\xi_1), c_{\ell', \epsilon'}(\xi_2)\}
 = \{ b_{\ell, \epsilon}(\xi_1), c^*_{\ell',
 \epsilon'}(\xi_2)\}=0 \ , \\
 &[ b_{\ell, \epsilon}(\xi_1), a_{\epsilon'}(\xi_3)]
 = [ b_{\ell, \epsilon}(\xi_1), a^*_{\epsilon'}(\xi_3)]
 = [ c_{\ell, \epsilon}(\xi_2), a_{\epsilon'}(\xi_3)]
 = [ c_{\ell, \epsilon}(\xi_2), a^*_{\epsilon'}(\xi_3)] = 0
 \ .
\end{split}
\end{equation}
Here, $\{b, b'\} = bb' + b'b$, $[a,a'] = aa' - a'a$.

We recall that the following operators, with $\vp\in
L^2(\Sigma_1)$,
\begin{equation}\nonumber
\begin{split}
  & b_{\ell, \epsilon}(\vp) = \int_{\Sigma_1} b_{\ell,
  \epsilon}(\xi) \overline{\vp(\xi)} \d \xi,\quad
  c_{\ell, \epsilon}(\vp) = \int_{\Sigma_1} c_{\ell,
  \epsilon}(\xi)\overline{\vp(\xi)} \d \xi\ , \\
  & b^*_{\ell, \epsilon}(\vp) = \int_{\Sigma_1} b^*_{\ell,
  \epsilon}(\xi) {\vp(\xi)} \d \xi,\quad
  c^*_{\ell, \epsilon}(\vp) = \int_{\Sigma_1} c^*_{\ell,
  \epsilon}(\xi){\vp(\xi)} \d \xi
\end{split}
\end{equation}
are bounded operators in $\gF$ such that
\begin{equation}\label{eq:2.32}
  \| b^\sharp_{\ell,\epsilon}(\vp)\| = \|
  c^\sharp_{\ell,\epsilon}(\vp)\| = \|\vp\|_{L^2}\ ,
\end{equation}
where $b^\sharp$ (resp. $c^\sharp$) is b (resp. $c$) or $b^*$
(resp. $c^*$).

The operators $b^\sharp_{\ell,\epsilon}(\vp)$ and
$c^\sharp_{\ell,\epsilon}(\vp)$ satisfy similar anticommutaion
relations (see e.g. \cite{Thaller1992}).

The free Hamiltonian $H_0$ is given by
\begin{equation}\nonumber
\begin{split}
 H_0 & = H_0^{(1)} + H_0^{(2)} + H_0^{(3)}  \\
 & = \sum_{\ell=1}^3 \sum_{\epsilon=\pm} \int
 w_\ell^{(1)}(\xi_1) b^*_{\ell,\epsilon}(\xi_1)
 b_{\ell,\epsilon}(\xi_1) \d \xi_1
 + \sum_{\ell=1}^3 \sum_{\epsilon=\pm}\int
 w_\ell^{(2)}(\xi_2) c^*_{\ell,\epsilon}(\xi_2)
 c_{\ell,\epsilon}(\xi_2) \d \xi_2 \\
 & + \sum_{\epsilon=\pm} \int w^{(3)}(\xi_3) a^*_\epsilon(\xi_3)
 a_\epsilon(\xi_3) \d \xi_3\ ,
\end{split}
\end{equation}
where
\begin{equation}\nonumber
\begin{split}
 & w_\ell^{(1)}(\xi_1) = (|p_1|^2 + m_\ell^2)^\frac12,\quad \mbox{with}\
 0<m_1<m_2<m_3\ , \\
 & w_\ell^{(2)}(\xi_2) = |p_2| \ , \\
 & w^{(3)}(\xi_3) = (|k|^2 + m^2_W)^\frac12\ ,
\end{split}
\end{equation}
where $m_W$ is the mass of the bosons $W^+$ and $W^-$ such that
$m_W>m_3$.

The spectrum of $H_0$ is $[0,\,\infty)$ and $0$ is a simple
eigenvalue with $\Omega$ as eigenvector. The set of thresholds of
$H_0$, denoted by $T$, is given by
\begin{equation}\nonumber
 T= \{ p\,m_1 + q\,m_2 + r\, m_3 + s\, m_W ; (p,\,q,\,r,\,s)
 \in\N^4 \mbox{ and } p+q+r+s\geq 1  \}\ ,
\end{equation}
and each set $[t,\infty)$, $t\in T$, is a branch of absolutely
continuous spectrum for $H_0$.

The interaction, denoted by $H_I$, is given by
\begin{equation}\label{eq:2.35}
  H_I = \sum_{\alpha=1}^2 H_I^{(\alpha)}\ ,
\end{equation}
where
\begin{equation}\label{eq:2.36}
\begin{split}
  H_I^{(1)} = & \sum_{\ell=1}^3 \sum_{\epsilon\neq\epsilon'}
  \int G^{(1)}_{\ell,\epsilon,\epsilon'} (\xi_1, \xi_2,\xi_3)
  b^*_{\ell,\epsilon}(\xi_1) c^*_{\ell, \epsilon'}(\xi_2)
  a_\epsilon(\xi_3) \d \xi_1 \d \xi_2 \d \xi_3 \\
  & + \sum_{\ell=1}^3 \sum_{\epsilon\neq \epsilon'}
  \int\overline{G^{(1)}_{\ell, \epsilon,\epsilon'} (\xi_1,
  \xi_2,\xi_3)} a^*_\epsilon(\xi_3) c_{\ell,\epsilon'} (\xi_2)
  b_{\ell,\epsilon}(\xi_1)  \d \xi_1 \d
  \xi_2 \d \xi_3 \ ,
\end{split}
\end{equation}
\begin{equation}\label{eq:2.37}
\begin{split}
  H_I^{(2)} = & \sum_{\ell=1}^3 \sum_{\epsilon\neq\epsilon'}
  \int G^{(2)}_{\ell,\epsilon,\epsilon'} (\xi_1, \xi_2,\xi_3)
  b^*_{\ell,\epsilon}(\xi_1) c^*_{\ell, \epsilon'}(\xi_2)
  a^*_\epsilon(\xi_3) \d \xi_1 \d \xi_2 \d \xi_3 \\
  & + \sum_{\ell=1}^3 \sum_{\epsilon\neq \epsilon'}
  \int\overline{G^{(2)}_{\ell, \epsilon,\epsilon'} (\xi_1,
  \xi_2,\xi_3)}
  a_\epsilon(\xi_3)
  c_{\ell,\epsilon'}(\xi_2)
  b_{\ell,\epsilon}(\xi_1)
  \d \xi_1 \d\xi_2 \d \xi_3 \ .
\end{split}
\end{equation}
The kernels $G^{(2)}_{\ell, \epsilon, \epsilon'} (.,.,.)$,
$\alpha=1,2$, are supposed to be functions.

The total Hamiltonian is then
\begin{equation}\label{eq:2.38}
  H = H_0 + g H_I,\quad g>0\ ,
\end{equation}
where $g$ is a coupling constant.

The operator $H_I^{(1)}$ describes the decay of the bosons $W^+$
and $W^-$ into leptons. Because of $H_I^{(2)}$ the bare vacuum
will not be an eigenvector of the total Hamiltonian for every
$g>0$ as we expect from the physics.

Every kernel $G_{\ell, \epsilon, \epsilon'}(\xi_1, \xi_2, \xi_3)$,
computed in theoretical physics, contains a $\delta$-distribution
because of the conservation of the momentum (see
\cite{GreinerMuller1989} \cite[section~4.4]{WeinbergI2005}). In
what follows, we approximate the singular kernels by square
integrable functions.

Thus, from now on, the kernels $G^{(\alpha)}_{\ell, \epsilon,
\epsilon'}$ are supposed to satisfy the following hypothesis .

\begin{hypothesis}\label{hypothesis:2.1}
For $\alpha=1,2$, $\ell=1,2,3$, $\epsilon, \epsilon' = \pm$, we
assume
\begin{equation}
 G^{(\alpha)}_{\ell, \epsilon, \epsilon'} (\xi_1, \xi_2, \xi_3)\in
 L^2(\Sigma_1\times \Sigma_1\times\Sigma_2)\ .
\end{equation}
\end{hypothesis}
\begin{remark}
 A similar model can be written down for the weak decay of pions
 $\pi^-$ and $\pi^+$ (see \cite[section~6.2]{GreinerMuller1989}).
\end{remark}

\begin{remark}
The total Hamiltonian is more general than the one involved in the
theory of weak interaction because, in the Standard Model,
neutrinos have helicity $-1/2$ and antineutrinos have helicity
$1/2$.

In the physical case, the Fock space, denoted by $\gF'$, is
isomorphic to $\gF_L'\otimes\gF_W$, with
\begin{equation}\nonumber
  \gF_L' = \bigotimes_{\ell=1}^3 \gF_{\ell}'\ ,
\end{equation}
and
\begin{equation}\nonumber
 \gF_\ell' = (\otimes_a^2 L^2(\Sigma_1)) \otimes (\otimes_a^2
 L^2(\R^3)) \ .
\end{equation}
The free Hamiltonian, now denoted by $H_0'$, is then given by
\begin{equation}\nonumber
\begin{split}
 H_0' = & \sum_{\ell=1}^3 \sum_{\epsilon=\pm} \int
 w_\ell^{(1)}(\xi_1) b^*_{\ell,\epsilon}(\xi_1)
 b_{\ell,\epsilon}(\xi_1) \d \xi_1
 + \sum_{\ell=1}^3 \sum_{\epsilon=\pm} \int_{\R^3}
 |p_2|  c^*_{\ell,\epsilon}(p_2)
 c_{\ell,\epsilon}(p_2) \d p_2 \\
 & + \sum_{\epsilon=\pm} \int w^{(3)}(\xi_3) a^*_\epsilon(\xi_3)
 a_\epsilon(\xi_3) \d \xi_3\ ,
\end{split}
\end{equation}
and the interaction, now denoted by $H_I'$, is the one obtained
from $H_I$ by supposing that $G^{(\alpha)}(\xi_1, (p_2, s_2),
\xi_3)=0$ if $s_2 = \epsilon\frac12$. The total Hamiltonian,
denoted by $H'$, is then given by $H' = H_0' + g\, H_I'$. The
results obtained in this paper for $H$ hold true for $H'$ with
obvious modifications.
\end{remark}

Under Hypothesis~\ref{hypothesis:2.1} a well defined operator on
$\gD$ corresponds to the formal interaction $H_I$ as it follows.

The formal operator
\begin{equation}\nonumber
  \int G^{(1)}_{\ell, \epsilon, \epsilon'}(\xi_1, \xi_2, \xi_3)
  b^*_{\ell, \epsilon}(\xi_1) c^*_{\ell, \epsilon'}(\xi_2)
  a_\epsilon(\xi_3) \d \xi_1 \d \xi_2 \d \xi_3
\end{equation}
is defined as a quadratic form on
$(\gD_\ell\otimes\gD_W)\times(\gD_\ell\otimes\gD_W)$ as
\begin{equation}\nonumber
 \int ( c_{\ell, \epsilon'}(\xi_2) b_{\ell, \epsilon}(\xi_1)
 \psi,\ G^{(1)}_{\ell, \epsilon, \epsilon'} a_\epsilon(\xi_3)
 \phi) \d \xi_1 \d \xi_2 \d \xi_3\ ,
\end{equation}
where $\psi$, $\phi\in\gD_\ell\otimes\gD_W$.

By mimicking the proof of \cite[Theorem~X.44]{ReedSimon1975}, we
get a closed operator, denoted by
$H^{(1)}_{I,\ell,\epsilon,\epsilon'}$, associated with the
quadratic form such that it is the unique operator in
$\gF_\ell\otimes\gF_W$ such that $\gD_\ell\otimes\gD_W \subset\
\cD(H^{(1)}_{I,\ell,\epsilon,\epsilon'})$ is a core for
$H^{(1)}_{I,\ell,\epsilon,\epsilon'}$ and
\begin{equation}\nonumber
 H^{(1)}_{I,\ell,\epsilon,\epsilon'} = \int
 G^{(1)}_{\ell,\epsilon,\epsilon'}(\xi_1,\xi_2,\xi_3)
 b^*_{\ell,\epsilon}(\xi_1) c^*_{\ell,\epsilon'}(\xi_2)
 a_\epsilon(\xi_3) \d \xi_1 \d \xi_2 \d \xi_3
\end{equation}
as quadratic forms on
$(\gD_\ell\otimes\gD_W)\times(\gD_\ell\otimes\gD_W)$.

The formal operator
\begin{equation}\nonumber
  \int \overline{ G^{(1)}_{\ell,\epsilon, \epsilon'}(\xi_1,
  \xi_2, \xi_3)} c_{\ell, \epsilon'}(\xi_2)
  b_{\ell, \epsilon}(\xi_1) a^*_\epsilon(\xi_3)
  \d \xi_1 \d \xi_2 \d \xi_3
\end{equation}
is similarly associated with
$(H^{(1)}_{I,\ell,\epsilon,\epsilon'})^*$ and
\begin{equation}\nonumber
 (H^{(1)}_{I,\ell,\epsilon,\epsilon'})^* =
 \int \overline{ G^{(1)}_{\ell,\epsilon, \epsilon'}(\xi_1,
  \xi_2, \xi_3)}  c_{\ell, \epsilon'}(\xi_2)
  b_{\ell, \epsilon}(\xi_1) a^*_\epsilon(\xi_3)
  \d \xi_1 \d \xi_2 \d \xi_3
\end{equation}
as quadratic forms on
$(\gD_\ell\otimes\gD_W)\times(\gD_\ell\otimes\gD_W)$. Moreover,
$\gD_\ell\otimes\gD_W
\subset\cD((H^{(1)}_{I,\ell,\epsilon,\epsilon'})^*)$ is a core for
$(H^{(1)}_{I,\ell,\epsilon,\epsilon'})^*$.

Again, there exists two closed operators
$H^{(2)}_{I,\ell,\epsilon,\epsilon'}$ and
$(H^{(2)}_{I,\ell,\epsilon,\epsilon'})^*$ such that
$\gD_\ell\otimes\gD_W
\subset\cD(H^{(2)}_{I,\ell,\epsilon,\epsilon'})$,
$\gD_\ell\otimes\gD_W
\subset\cD((H^{(2)}_{I,\ell,\epsilon,\epsilon'})^*)$ and
$\gD_\ell\otimes\gD_W$ is a core for
$H^{(2)}_{I,\ell,\epsilon,\epsilon'}$ and
$(H^{(2)}_{I,\ell,\epsilon,\epsilon'})^*$ and such that
\begin{equation}\nonumber
H^{(2)}_{I,\ell,\epsilon,\epsilon'} = \int
G^{(2)}_{\ell,\epsilon,\epsilon'} (\xi_1, \xi_2, \xi_3)
b^*_{\ell,\epsilon}(\xi_1) c^*_{\ell,\epsilon'}(\xi_2)
a^*_\epsilon(\xi_3) \d \xi_1 \d \xi_2 \d \xi_3\ ,
\end{equation}
\begin{equation}\nonumber
(H^{(2)}_{I,\ell,\epsilon,\epsilon'})^* = \int
G^{(2)}_{\ell,\epsilon,\epsilon'} (\xi_1, \xi_2, \xi_3)
a_\epsilon(\xi_3) c_{\ell,\epsilon'}(\xi_2)
b_{\ell,\epsilon}(\xi_1)
 \d \xi_1 \d \xi_2 \d \xi_3
\end{equation}
as quadratic forms on $(\gD_\ell\otimes\gD_W)\times
(\gD_\ell\otimes\gD_W)$.

We shall still denote $H^{(\alpha)}_{I,\ell,\epsilon,\epsilon'}$
and $(H^{(\alpha)}_{I,\ell,\epsilon,\epsilon'})^*$ ($\alpha=1,2$)
their extensions to $\gF$. The set $\gD$ is then a core for
$H^{(\alpha)}_{I,\ell,\epsilon,\epsilon'}$ and
$(H^{(\alpha)}_{I,\ell,\epsilon,\epsilon'})^*$

Thus
\begin{equation}\nonumber
 H = H_0 + g\sum_{\alpha=1,2} \sum_{\ell=1}^3
 \sum_{\epsilon\neq\epsilon'}
 (H^{(\alpha)}_{I,\ell,\epsilon,\epsilon'}
 + (H^{(2)}_{I,\ell,\epsilon,\epsilon'})^*)
\end{equation}
is a symmetric operator defined on $\gD$.

We now want to prove that $H$ is essentially self-adjoint on $\gD$
by showing that $H^{(\alpha)}_{I,\ell,\epsilon,\epsilon'}$ and
$(H^{(\alpha)}_{I,\ell,\epsilon,\epsilon'})^*$ are relatively
$H_0$-bounded.

Once again, as above, for almost every $\xi_3\in\Sigma_2$, there
exists closed operators in $\gF_L$, denoted by
$B^{(\alpha)}_{\ell,\epsilon,\epsilon'}(\xi_3)$ and
$(B^{(\alpha)}_{\ell,\epsilon,\epsilon'}(\xi_3))^*$ such that
\begin{equation}\nonumber
 B^{(1)}_{\ell,\epsilon,\epsilon'}(\xi_3) =
 - \int \overline{
 G^{(1)}_{\ell,\epsilon,\epsilon'}(\xi_1,\xi_2,\xi_3)}
 b_{\ell,\epsilon}(\xi_1) c_{\ell,\epsilon'}(\xi_2) \d \xi_1 \d
 \xi_2\ ,
\end{equation}
\begin{equation}\nonumber
 (B^{(1)}_{\ell,\epsilon,\epsilon'}(\xi_3))^* =
 \int
 G^{(1)}_{\ell,\epsilon,\epsilon'}(\xi_1,\xi_2,\xi_3)
 b^*_{\ell,\epsilon}(\xi_1) c^*_{\ell,\epsilon'}(\xi_2) \d \xi_1 \d
 \xi_2\ ,
\end{equation}
\begin{equation}\nonumber
 B^{(2)}_{\ell,\epsilon,\epsilon'}(\xi_3) =
 \int
 G^{(2)}_{\ell,\epsilon,\epsilon'}(\xi_1,\xi_2,\xi_3)
 b^*_{\ell,\epsilon}(\xi_1) c^*_{\ell,\epsilon'}(\xi_2) \d \xi_1 \d
 \xi_2\ ,
\end{equation}
\begin{equation}\nonumber
 (B^{(2)}_{\ell,\epsilon,\epsilon'}(\xi_3))^* =
 - \int \overline{
 G^{(2)}_{\ell,\epsilon,\epsilon'}(\xi_1,\xi_2,\xi_3)}
 b_{\ell,\epsilon}(\xi_1) c_{\ell,\epsilon'}(\xi_2) \d \xi_1 \d
 \xi_2\ \,
\end{equation}
as quadratic forms on $\gD_\ell\times\gD_\ell$.

We have that $\gD_\ell\subset
\cD(B^{(\alpha)}_{\ell,\epsilon,\epsilon'}(\xi_3))$ (resp.
$\gD_\ell\subset
\cD((B^{(\alpha)}_{\ell,\epsilon,\epsilon'}(\xi_3))^*)$ is a core
for $B^{(\alpha)}_{\ell,\epsilon,\epsilon'}(\xi_3)$ (resp. for
$(B^{(\alpha)}_{\ell,\epsilon,\epsilon'}(\xi_3))^*$). We still
denote by $B^{(\alpha)}_{\ell,\epsilon,\epsilon'}(\xi_3))$ and
$(B^{(\alpha)}_{\ell,\epsilon,\epsilon'}(\xi_3))^*)$ their
extensions to $\gF_L$.

It then follows that the operator $H_I$ with domain $\gD$ is
symmetric and can be written in the following form
\begin{equation}\nonumber
\begin{split}
 & H_I = \sum_{\alpha=1,2} \sum_{\ell=1}^3
 \sum_{\epsilon\neq\epsilon'}
 (H^{(\alpha)}_{I,\ell,\epsilon,\epsilon'} +
  (H^{(\alpha)}_{I,\ell,\epsilon,\epsilon'})^*) \\
     & \! = \! \sum_{\alpha=1,2} \sum_{\ell=1}^3
 \sum_{\epsilon\neq\epsilon'}\int
  B^{(\alpha)}_{\ell,\epsilon,\epsilon'}(\xi_3)
  \otimes a^*_\epsilon(\xi_3) \d \xi_3 + \!
  \sum_{\alpha=1,2} \sum_{\ell=1}^3
  \sum_{\epsilon\neq\epsilon'}\int
  (B^{(\alpha)}_{\ell,\epsilon,\epsilon'}(\xi_3))^*
  \otimes a_\epsilon(\xi_3) \d \xi_3\, .
\end{split}
\end{equation}
Let $N_\ell$ denote the operator number of massive leptons $\ell$
in $\gF_\ell$, i.e.,
\begin{equation}
  N_\ell = \sum_\epsilon \int b_{\ell,\epsilon}^*(\xi_1)
  b_{\ell,\epsilon}(\xi_1) \d \xi_1\ .
\end{equation}
The operator $N_\ell$ is a positive self-adjoint operator in
$\gF_\ell$. We still denote by $N_\ell$ its extension to $\gF_L$.
The set $\gD_L$ is a core for $N_\ell$.

We then have
\begin{proposition}\label{proposition:2.4}
For a.e. $\xi_3\in\Sigma_2$,
$\cD(B^{(\alpha)}_{\ell,\epsilon,\epsilon'}(\xi_3))$,
$\cD((B^{(\alpha)}_{\ell,\epsilon,\epsilon'}(\xi_3))^*)
\supset\cD(N_\ell^\frac12)$, and for
$\Phi\in\cD(N_\ell^\frac12)\subset\gF_L$ we have
\begin{equation}\label{eq:2.53}
  \| B^{(\alpha)}_{\ell,\epsilon, \epsilon'}(\xi_3) \Phi\|_{\gF_L}
  \leq \|
  G^{(\alpha)}_{\ell,\epsilon, \epsilon'}(.,.,\xi_3)
  \|_{L^2(\Sigma_1\times\Sigma_1)}
  \|N_\ell^\frac12 \Phi\|_{\gF_L}\ ,
\end{equation}
\begin{equation}\label{eq:2.54}
  \| (B^{(\alpha)}_{\ell,\epsilon, \epsilon'}(\xi_3))^* \Phi\|_{\gF_L}
  \leq \|
  G^{(\alpha)}_{\ell,\epsilon, \epsilon'}(.,.,\xi_3)
  \|_{L^2(\Sigma_1\times\Sigma_1)}
  \|N_\ell^\frac12 \Phi\|_{\gF_L}\ .
\end{equation}
\end{proposition}
%
%
%

\begin{proof}
The estimates \eqref{eq:2.53} and \eqref{eq:2.54} are examples of
$N_\tau$ estimates (see \cite{GlimmJaffe1977}). We give a proof
for sake of completeness. We only consider $B^{(1)}_{1,+,-}$. The
other cases are quite similar.

Let $\Phi=(\Phi^{(Q)})_Q$ and $\Psi = (\Psi^{(Q')})_{Q'}$ be two
vectors in $\gD_L$, where we use the notations
$Q=\ind_{\ell=1,2,3}$, and $Q'=(q_\ell', \bar q_\ell', r_\ell',
\bar r_\ell')_{\ell=1,2,3}$. We have
\begin{equation}
\begin{split}
  & (\Psi^{(Q')}, B^{(1)}_{1,+,-}(\xi_3) \Phi^{(Q)})_{\gF_L} =
  - \delta_{q'_1\, q_1-1} \delta_{\bar q_1'\,\bar q_1}
  \delta_{r'_1\,r_1} \delta_{\bar r_1'\,
  \bar r_1-1}
  \prod_{\ell = 2}^3 \delta_{q_\ell' q_\ell}
  \delta_{\bar q_\ell' \bar q_\ell}
  \delta_{r_\ell' r_\ell}
  \delta_{\bar r_\ell' \bar r_\ell} \\
  & \int_{\Sigma_1\times\Sigma_1} (\Psi^{(\tilde Q)},
  b_{1,+}(\xi_1) c_{1,-}(\xi_2) \Phi^{(Q)} )_{\gF_L}
  \overline{G^{(1)}_{1,+,-}(\xi_1, \xi_2, \xi_3)} \d \xi_1
  \d \xi_2\ .
\end{split}
\end{equation}
Here $\tilde Q = (q_1-1, \bar q_1, r_1, \bar r_1-1, q_2, \bar q_2,
r_2, \bar r_2, q_3, \bar q_3, r_3, \bar r_3)$.

For each $Q$,
\begin{equation}\label{eq:2.56}
 B^{(1)}_{1,+,-}(\xi_3) \Phi^{(Q)} \in \gF_1^{(q_1-1, \bar q_1,
 r_1, \bar r_1-1)}\otimes \gF_2^{(q_2, \bar q_2, r_2, \bar r_2)}
 \otimes \gF_3^{(q_3, \bar q_3, r_3, \bar r_3)}.
\end{equation}
By the Fubini theorem we have
\begin{equation}\nonumber
\begin{split}
  & \left| ( \Psi^{(\tilde Q)}, B^{(1)}_{1,+,-}(\xi_3)
  \Psi^{(Q)})_{\gF_L} \right| \\
  & = \left|\int_{\Sigma_1} \left(\int_{\Sigma_1}
  G^{(1)}_{1,+,-}(\xi_1, \xi_2, \xi_3) c_{1,-}^*(\xi_2)
  \Psi^{(\tilde Q)} \d \xi_2, b_{1,+}(\xi_1) \Phi^{(Q)}\right)_{\gF_L}
  \d\xi_1\right|\ .
\end{split}
\end{equation}
By \eqref{eq:2.32}, and the Cauchy-Schwarz inequality we get
\begin{equation}\nonumber
\begin{split}
  & \left| ( \Psi^{(\tilde Q)}, B^{(1)}_{1,+,-}(\xi_3)
  \Psi^{(Q)})_{\gF_L} \right|^2 \\
  & \leq  \left(\int_{\Sigma_1} \|b_{1,+}(\xi_1)\Phi^{(Q)}\|
  \left(\int_{\Sigma_1}
  | G^{(1)}_{1,+,-}(\xi_1, \xi_2, \xi_3)|^2 \d\xi_2
  \right)^\frac12 \d\xi_1\right)^2 \|\Psi^{(\tilde Q)}\|^2\ .
\end{split}
\end{equation}
By the definition of $b_{1,+}(\xi_1) \Phi^{(Q)}$ and the
Cauchy-Schwarz inequality we get
\begin{equation}\nonumber
\begin{split}
 & | (\Psi^{(\tilde Q)}, B^{(1)}_{1,+,-}(\xi_3) \Phi^{(Q)}
 )_{\gF_L}|^2\\
 & \leq
 q_1 \left(\int_{\Sigma_1}\int_{\Sigma_1} | G^{(1)}_{1,+,-}
 (\xi_1, \xi_2, \xi_3)|^2 \d\xi_1 \d\xi_2\right)
 \| \Psi^{(\tilde Q)}\|^2_{\gF_L} \|\Phi^{(Q)}\|^2_{\gF_L} \\
 & = \left(\int_{\Sigma_1}\int_{\Sigma_1} | G^{(1)}_{1,+,-}
 (\xi_1, \xi_2, \xi_3)|^2 \d\xi_1 \d\xi_2\right)
 \| \Psi^{(\tilde Q)}\|^2_{\gF_L} \|N_1^\frac12
 \Phi^{(Q)}\|^2_{\gF_L}\ .
\end{split}
\end{equation}
By \eqref{eq:2.56} we have
\begin{equation}\nonumber
\begin{split}
 | (\Psi, B^{(1)}_{1,+,-}(\xi_3) \Phi^{(Q)} )_{\gF_L}|^2
 \leq \|
 \Psi\|_{\gF_L}^2 \| N_1^\frac12 \Phi^{(Q)}\|^2_{\gF_L}
 \int_{\Sigma_1\times\Sigma_1} |
 G^{(1)}_{1,+,-}(\xi_1,\xi_2,\xi_3) |^2 \d\xi_1 \d\xi_2\ ,
\end{split}
\end{equation}
for every $\Psi\in\gD_L$. Therefore we get
\begin{equation}\nonumber
 \| B^{(1)}_{1,+,-}(\xi_3) \Phi^{(Q)}\|^2_{\gF_L}
 \leq \left( \int_{\Sigma_1\times\Sigma_1} |G^{(1)}_{1,+,-} (\xi_1,
 \xi_2, \xi_3)|^2 \d\xi_1 \d\xi_2\right) \|N_1^\frac12
 \Phi^{(Q)}\|^2_{\gF_L}\ ,
\end{equation}
and by \eqref{eq:2.56} we finally obtain
\begin{equation}\nonumber
 \| B^{(1)}_{1,+,-}(\xi_3) \Phi\|^2_{\gF_L}
 \leq \left( \int_{\Sigma_1\times\Sigma_1} |G^{(1)}_{1,+,-} (\xi_1,
 \xi_2, \xi_3)|^2 \d\xi_1 \d\xi_2\right) \|N_1^\frac12
 \Phi\|^2_{\gF_L}\ ,
\end{equation}
for every $\Phi\in\gD$.

Since $\gD_L$ is a core for $N_1^\frac12$ and $B^{(1)}_{1,+,-}$
with domain $\gD_L$ is closable,
$\cD(B^{(1)}_{1,+,-}(\xi_3))\supset \cD(N_1^\frac12)$, and
\eqref{eq:2.53} is satisfied for every $\Phi\in\cD(N_1^\frac12)$.
\end{proof}

Let
\begin{equation}\nonumber
  H^{(3)}_{0,\epsilon} = \int w^{(3)}(\xi_3)
  a_\epsilon^*(\xi_3) a_\epsilon(\xi_3) \d\xi_3\ .
\end{equation}
Then $H^{(3)}_{0,\epsilon}$ is a self-adjoint operator in $\gF_W$,
and $\gD_W$ is a core for $H^{(3)}_{0,\epsilon}$.

We get
\begin{proposition}\label{proposition:2.5}
\begin{equation}\label{eq:2.63}
\begin{split}
 & \| \int (B^{(\alpha)}_{\ell,\epsilon,\epsilon'}(\xi_3))^*\otimes
 a_\epsilon(\xi_3) \d\xi_3 \Psi\|^2 \\
 & \leq
 (\int_{\Sigma_1\times\Sigma_1\times\Sigma_2}
 \frac{|G^{(\alpha)}_{\ell,\epsilon,\epsilon'}(\xi_1,\xi_2,\xi_3)|^2}{w^{(3)}(\xi_3)}
 \d \xi_1 \d \xi_2 \d \xi_3)\
 \|(N_\ell+1)^\frac12 \otimes (H_{0,\epsilon}^{(3)})^\frac12
 \Psi\|^2
\end{split}
\end{equation}
and
\begin{equation}\label{eq:2.64}
\begin{split}
 & \| \int B^{(\alpha)}_{\ell,\epsilon,\epsilon'}(\xi_3)\otimes
 a^*_\epsilon(\xi_3) \d\xi_3 \Psi\|^2 \\
 & \leq
 (\int_{\Sigma_1\times\Sigma_1\times\Sigma_2}
 \frac{|G^{(\alpha)}_{\ell,\epsilon,\epsilon'}(\xi_1,\xi_2,\xi_3)|^2}{w^{(3)}(\xi_3)}
 \d \xi_1 \d \xi_2 \d \xi_3)\
 \|(N_\ell+1)^\frac12 \otimes (H_{0,\epsilon}^{(3)})^\frac12
 \Psi\|^2 \\
 & +
 (\int_{\Sigma_1\times\Sigma_1\times\Sigma_2}
 |G^{(\alpha)}_{\ell,\epsilon,\epsilon'}(\xi_1,\xi_2,\xi_3)|^2
 \d \xi_1 \d \xi_2 \d \xi_3)\
 (\eta \|(N_\ell+1)^\frac12 \otimes \1\ \Psi\|^2
 + \frac{1}{4\eta} \|\Psi\|^2) \ ,
\end{split}
\end{equation}
for every $\Psi\in\cD(H_0)$ and every $\eta>0$.
\end{proposition}
%
%
%
%
%
%
\begin{proof}
Suppose that $\Psi \in \cD(N_\ell^\frac12) \hat\otimes
\cD((H_{0,\epsilon}^{(3)})^\frac12)$. Let
\begin{equation}\nonumber
 \Psi_\epsilon(\xi_3) = w^{(3)} (\xi_3)^\frac12 (
 (N_\ell+1)^\frac12 \otimes a_\epsilon(\xi_3)) \Phi\ .
\end{equation}
We have
\begin{equation}\nonumber
 \int_{\Sigma_2} \|\Psi_\epsilon(\xi_3) \|^2 \d\xi_3
 = \| (N_\ell+1)^\frac12 \otimes (H_{0,\epsilon}^{(3)})^\frac12
 \Psi\|^2 \ .
\end{equation}
We get
\begin{equation}\nonumber
\begin{split}
 & \int
 (B^{(\alpha)}_{\ell,\epsilon,\epsilon'}(\xi_3))^*\otimes
 a_\epsilon(\xi_3) \d \xi_3 \Psi \\
 & = \int_{\Sigma_2} \frac{1}{(w^{(3)}(\xi_3))^{\frac12}}
 ((B^{(\alpha)}_{\ell,\epsilon,\epsilon'}(\xi_3))^*
 (N_\ell+1)^{-\frac12} \otimes\1)
 \Psi_\epsilon(\xi_3) \d\xi_3\ .
\end{split}
\end{equation}
Therefore
\begin{equation}\label{eq:2.68}
\begin{split}
 & \| \int (B^{(\alpha)}_{\ell,\epsilon,\epsilon'} (\xi_3))^*
 \otimes a_\epsilon(\xi_3)\, \Psi \d\xi_3\|^2_\gF \\
 & \leq (\int_{\Sigma_2} \frac{1}{w^{(3)}(\xi_3)}
 \| (B^{(\alpha)}_{\ell,\epsilon,\epsilon'} (\xi_3))^*
 (N_\ell+1)^{-\frac12}\|_{\gF_L} \|\Psi_\epsilon(\xi_3)\|_{\gF}
 \d\xi_3)^2 \\
 & \leq (\int_{\Sigma_1\times\Sigma_1\times\Sigma_2}
 \frac{ | G^{(\alpha)}_{\ell,\epsilon,\epsilon'}(\xi_2, \xi_2, \xi3)|^2}
 {w^{(3)}(\xi_3)} \d\xi_1 \d\xi_2 \d\xi_3 ) \| (N_\ell+1)^\frac12
 \otimes (H^{(3)}_{0,\epsilon})^\frac12 \Psi\|_\gF^2\ ,
\end{split}
\end{equation}
as it follows from Proposition~\ref{proposition:2.4}.

We now have
\begin{equation}\nonumber
\begin{split}
 & \| \int B^{(\alpha)}_{\ell,\epsilon,\epsilon'}(\xi_3) \otimes
 a_\epsilon^*(\xi_3) \, \Psi \d \xi_3 \|^2_\gF \\
 & = \int (B^{(\alpha)}_{\ell,\epsilon,\epsilon'}(\xi_3) \otimes
 a_\epsilon(\xi_3') \Psi,\
 B^{(\alpha)}_{\ell,\epsilon,\epsilon'}(\xi_3')
 \otimes a_\epsilon(\xi_3) \,\Psi) \d\xi_3 \d\xi_3'
 + \int \| (B^{(\alpha)}_{\ell,\epsilon,\epsilon'} \otimes
 \1)\Psi\|^2 \d \xi_3\ ,
\end{split}
\end{equation}
and
\begin{equation}\label{eq:2.70}
\begin{split}
  & \int_{\Sigma_2\times\Sigma_2} (B^{(\alpha)}_{\ell,\epsilon,\epsilon'}(\xi_3) \otimes
  a_\epsilon(\xi_3') \Psi,
  B^{(\alpha)}_{\ell,\epsilon,\epsilon'}(\xi_3')\otimes
  a_\epsilon(\xi_3)\Psi ) \d\xi_3 \d\xi_3' \\
  & = \int_{\Sigma_2\times\Sigma_2}
  \frac{1}{ w^{(3)}(\xi_3)^\frac12 w^{(3)}(\xi_3')^\frac12}
  \Big( (B^{(\alpha)}_{\ell,\epsilon,\epsilon'}(\xi_3)
  (N_\ell +1)^{-\frac12} \otimes \1 ) \Psi_\epsilon(\xi_3'), \\
  & (B^{(\alpha)}_{\ell,\epsilon,\epsilon'}(\xi_3')
  (N_\ell +1)^{-\frac12} \otimes \1 ) \Psi_\epsilon(\xi_3)
  \Big) \d\xi_3 \d\xi_3' \\
  & \leq (\int_{\Sigma_2}
  \frac{1}{ w^{(3)}(\xi_3)^\frac12}
  \| B^{(\alpha)}_{\ell,\epsilon,\epsilon'}(\xi_3)
  (N_\ell + 1)^{-\frac12} \|_{\gF_L}
  \| \Psi_\epsilon(\xi_3)\| \d\xi_3)^2 \\
  & \leq (\int_{\Sigma_1\times\Sigma_1\times\Sigma_2}
  \frac{|G^{(\alpha)}(\xi_1,\xi_2,\xi_3)|^2}{ w^{(3)}(\xi_3)}
  \d\xi_1 \d\xi_2 \d\xi_3)
  \| (N_\ell +1)^\frac12 \otimes (H_{0,\epsilon}^{(3)})^\frac12
  \Psi\|^2 \ .
\end{split}
\end{equation}
Furthermore
\begin{equation}\label{eq:2.71}
\begin{split}
  & \int_{\Sigma_2} \| B^{(\alpha)}_{\ell,\epsilon,\epsilon'}(\xi_3)
  \otimes \1) \Psi\|^2 \d\xi_3 \\
  & = \int_{\Sigma_2} \| B^{(\alpha)}_{\ell,\epsilon,\epsilon'}(\xi_3)
  (N_\ell+1)^{-\frac12}\otimes \1)
  ((N_\ell+1)^\frac12\otimes\1) \Psi\|^2 \d\xi_3 \\
  & \leq \left(\int_{\Sigma_1\times\Sigma_1\times\Sigma_2}
  |G^{(\alpha)}_{\ell,\epsilon,\epsilon'}(\xi_1,\xi_2,\xi_3)|^2
  \d\xi_1\d\xi_2\d\xi_3 \right)\, (\eta \|(N_\ell +1)\Psi\|^2 +
  \frac{1}{4\eta} \|\Psi\|^2)\ ,
\end{split}
\end{equation}
for every $\eta>0$.

By \eqref{eq:2.68}, \eqref{eq:2.70}, and \eqref{eq:2.71}, we
finally get \eqref{eq:2.63} and \eqref{eq:2.64} for every $\Psi\in
\cD(N_\ell^\frac12)\hat\otimes\cD(H_{0,\epsilon}^{(3)})$. The set
$\cD(N_\ell^\frac12)\hat\otimes\cD(H_{0,\epsilon}^{(3)})$ is a
core for $N_\ell^\frac12\otimes H_{0,\epsilon}^{(3)}$ and
$\cD(H_0)\subset \cD(N_\ell^\frac12\otimes H_{0,\epsilon}^{(3)})$.
It then follows that \eqref{eq:2.63} and \eqref{eq:2.64} are
verified for every $\Psi\in\cD(H_0)$.
\end{proof}

We now prove that $H$ is a self-adjoint operator in $\gF$ for $g$
sufficiently small.
\begin{theorem}
Let $g_1>0$ be such that
\begin{equation}\nonumber
 \frac{3 g_1^2}{m_W} (\frac{1}{m_1^2} +1) \sum_{\alpha=1,2}
 \sum_{\ell=1}^3 \sum_{\epsilon\neq\epsilon'} \|
 G^{(\alpha)}_{\ell,\epsilon,\epsilon'}
 \|^2_{L^2(\Sigma_1\times\Sigma_1\times\Sigma_2)} <1\ .
\end{equation}
Then for every $g$ satisfying $g\leq g_1$, $H$ is a self-adjoint
operator in $\gF$ with domain $\cD(H)=\cD(H_0)$, and $\gD$ is a
core for $H$.
\end{theorem}
%
%
%
\begin{proof}
Let $\Psi$ be in $\gD$. We have
\begin{equation}\label{eq:2.73}
\begin{split}
 \| H_I \Psi \|^2 \leq &
 12 \sum_{\alpha=1,2} \sum_{\ell=1}^3 \sum_{\epsilon\neq\epsilon'}
 \Big\{ \left\| \int (B^{(\alpha)}_{\ell,\epsilon,\epsilon'}(\xi_3))^*
 \otimes a_\epsilon(\xi_3)\, \Psi\d\xi_3
 \right\|^2 \\
 & + \left\| \int (B^{(\alpha)}_{\ell,\epsilon,\epsilon'}(\xi_3))
 \otimes a^*_\epsilon(\xi_3) \, \Psi \d\xi_3 \right\|^2 \Big\} \ .
\end{split}
\end{equation}
Note that
\begin{equation}\nonumber
\| H_{0,\epsilon}^{(3)}\Psi\| \leq \|H_0^{(3)} \Psi\| \leq \|H_0
\Psi\|\ ,
\end{equation}
and
\begin{equation}\nonumber
 \| N_\ell \Psi\| \leq \frac{1}{m_\ell} \|H_{0,\ell}\Psi\|
 \leq \frac{1}{m_1} \|H_{0,\ell}\Psi\|
 \leq \frac{1}{m_1} \|H_{0}\Psi\|\ ,
\end{equation}
where
\begin{equation}\label{def:eq:2.76}
 H_{0,\ell} = \sum_\epsilon \int w_\ell^{(1)}(\xi_1)
 b_{\ell,\epsilon}^*(\xi_1) b_{\ell,\epsilon}(\xi_1) \d \xi_1
 + \sum_\epsilon \int w_\ell^{(2)}(\xi_2)
 c_{\ell,\epsilon}^*(\xi_2) c_{\ell,\epsilon}(\xi_2) \d \xi_2\ .
\end{equation}
We further note that
\begin{equation}\label{eq:2.77}
 \| (N_\ell + 1)^\frac12 \otimes (H^{(3)}_{0,\epsilon})^\frac12
 \Psi\|^2 \leq
 \frac12 (\frac{1}{m_1^2} + 1) \| H_0\Psi\|^2 + \frac{\beta}{2
 m_1^2} \| H_0\Psi\|^2 + (\frac12 + \frac{1}{8\beta}) \|\Psi\|^2
 ,
\end{equation}
for $\beta>0$, and
\begin{equation}\label{eq:2.78}
 \eta \| ((N_\ell+1) \otimes \1) \Psi\|^2 + \frac{1}{4\eta}
 \|\Psi\|^2
 \leq \frac{\eta}{m_1^2} \|H_0\Psi\|^2 + \frac{\eta\beta}{m_1^2}
 \|H_0\Psi\|^2 + \eta(1+\frac{1}{4\beta}) \|\Psi\|^2 +
 \frac{1}{4\eta} \|\Psi\|^2 .
\end{equation}
Combining \eqref{eq:2.73} with \eqref{eq:2.63}, \eqref{eq:2.64},
\eqref{eq:2.77} and \eqref{eq:2.78} we get for $\eta>0$, $\beta>0$
\begin{equation}\label{eq:2.79}
\begin{split}
 & \|H_I \Psi\|^2 \leq
 6 (\sum_{\alpha=1,2} \sum_{\ell=1}^3 \sum_{\epsilon\neq\epsilon'}
 \|
 G^{(\alpha)}_{\ell,\epsilon,\epsilon'}\|^2) \\
 &\Big(\frac{1}{2m_W}
 (\frac{1}{m_1^2} +1) \|H_0 \Psi\|^2
 + \frac{\beta}{2 m_W m_1^2} \|H_0\Psi\|^2
 + \frac{1}{2 m_W} (1+\frac{1}{4\beta})\|\Psi\|^2\Big) \\
 & + 12(\sum_{\alpha=1,2} \sum_{\ell=1}^3 \sum_{\epsilon\neq\epsilon'}
 \|
 G^{(\alpha)}_{\ell,\epsilon,\epsilon'}\|^2)
 (\frac{\eta}{m_1^2} (1+\beta) \|H_0\Psi\|^2 + (\eta
 (1+\frac{1}{4\beta}) +\frac{1}{4\eta}) \|\Psi\|^2) ,
\end{split}
\end{equation}
by noting
\begin{equation}\label{eq:2.80}
 \int_{\Sigma_1\times\Sigma_1\times\Sigma_2}
 \frac{ | G_{\ell,\epsilon,\epsilon'}(\xi_1,\xi_2,\xi_3)|^2}
 {w^{(3)}(\xi_3)} \d\xi_1\d\xi_2\d\xi_3 \leq \frac{1}{m_W}
 \| G^{(\alpha)}_{\ell,\epsilon,\epsilon'}\|^2 .
\end{equation}
By \eqref{eq:2.79} the theorem follows from the Kato-Rellich
theorem.
\end{proof}

\section{Main results}\label{section3}
\setcounter{equation}{0} In the sequel, we shall make the
following additional assumptions on the kernels
$G^{(\alpha)}_{\ell,\epsilon,\epsilon'}$.

\begin{hypothesis}\label{hypothesis:3.1}$\ $\\
\noindent$(i)$ For $\alpha=1,2,\ \ell=1,2,3,\
\epsilon,\epsilon'=\pm$,
\begin{equation}\nonumber
 \int_{\Sigma_1\times\Sigma_1\times\Sigma_2}
 \frac{|
 G^{(\alpha)}_{\ell,\epsilon,\epsilon'}(\xi_1,\xi_2,\xi_3)|^2}{|p_2|^2}
 \d\xi_1 \d\xi_2 \d\xi_3 <\infty,\quad
\end{equation}

\noindent$(ii)$ There exists $C>0$ such that for $\alpha=1,2,\
\ell=1,2,3,\ \epsilon,\epsilon'=\pm$,
\begin{equation}\nonumber
 \left(\int_{\Sigma_1\times\{|p_2|\leq \sigma\}\times\Sigma_2}
 |G^{(\alpha)}_{\ell,\epsilon,\epsilon'}(\xi_1,\xi_2,\xi_3)|^2
 \d\xi_1 \d\xi_2 \d\xi_3\right)^\frac12 \leq C\sigma^2.
\end{equation}

\noindent$(iii)$ For $\alpha=1,2,\ \ell=1,2,3,\
\epsilon,\epsilon'=\pm $, and $i,j=1,2,3$
\begin{equation}\nonumber
 (iii.a)\quad \int_{\Sigma_1\times\Sigma_1\times\Sigma_2}
 \left| [ (p_2\cdot\nabla_{p_2})
 G^{(\alpha)}_{\ell,\epsilon,\epsilon'}](\xi_1,\xi_2,\xi_3)\right|^2
 \d\xi_1 \d\xi_2 \d\xi_3 < \infty\ ,
\end{equation}
and
\begin{equation}\nonumber
 (iii.b)\quad \int_{\Sigma_1\times\Sigma_1\times\Sigma_2}
 p_{2,i}^2\, p_{2,j}^2
 \left|
 \frac{\partial^2 G^{(\alpha)}_{\ell,\epsilon,\epsilon'}}
 {\partial p_{2,i} \partial p_{2,j}}(\xi_1,\xi_2,\xi_3)\right|^2
 \d\xi_1 \d\xi_2 \d\xi_3 < \infty\ .
\end{equation}

\noindent$(iv)$ There exists $\Lambda>m_1$ such, that for
$\alpha=1,2$, $\ell=1,2,3$, $\epsilon, \epsilon' = \pm$,
\begin{equation}\nonumber
 G^{(\alpha)}_{\ell,\epsilon,\epsilon'}
 (\xi_1,\xi_2,\xi_3) = 0 \quad\mbox{if}\quad |p_2| \geq \Lambda\ .
\end{equation}

\end{hypothesis}

\begin{remark}
Hypothesis~\ref{hypothesis:3.1} (ii) is nothing but an infrared
regularization of the kernels
$G^{(\alpha)}_{\ell,\epsilon,\epsilon'}$. In order to satisfy this
hypothesis it is, for example, sufficient to suppose
\begin{equation}\nonumber
 G^{(\alpha)}_{\ell,\epsilon,\epsilon'} (\xi_1,\xi_2,\xi_3)
 = |p_2|^\frac12
 \tilde{G}^{(\alpha)}_{\ell,\epsilon,\epsilon'}(\xi_1,\xi_2,\xi_3)\
 ,
\end{equation}
where $\tilde{G}^{(\alpha)}_{\ell,\epsilon,\epsilon'}$ is a smooth
function of $(p_1,p_2,p_3)$ in the Schwartz space.

The Hypothesis~\ref{hypothesis:3.1} (iv), which is a sharp
ultraviolet cutoff, is actually not necessary, and can be removed
at the expense of some additional technicalities in
Appendix~\ref{appendix}. However, in order to simplify the proof
of Proposition~\ref{proposition:3.5}, we shall leave it.
\end{remark}

Our first result is devoted to the existence of a ground state for
$H$ together with the location of the spectrum of $H$ and of its
absolutely continuous spectrum when $g$ is sufficiently small.

\begin{theorem}\label{thm:3.3}
Suppose that the kernels $G^{(\alpha)}_{\ell,\epsilon,\epsilon'}$
satisfy Hypothesis~\ref{hypothesis:2.1} and
Hypothesis~\ref{hypothesis:3.1} (i). Then there exists $0<g_2\leq
g_1$ such that $H$ has a unique ground state for $g\leq g_2$.
Moreover
\begin{equation}\nonumber
   \sigma(H) = \sigma_{\rm{ac}}(H) = [\inf \sigma(H), \infty)\ ,
\end{equation}
with $\inf\sigma(H) \leq 0$.
\end{theorem}

According to Theorem~\ref{thm:3.3} the ground state energy
$E=\inf\sigma(H)$ is a simple eigenvalue of $H$ and our main
results are concerned with a careful study of the spectrum of $H$
above the ground state energy. The spectral theory developed in
this work is based on the conjugated operator method as described
in \cite{Mourre1981}, \cite{Amreinetal1996} and
\cite{Sahbani1997}. Our choice of the conjugate operator denoted
by $A$ is the second quantized dilation generator for the
neutrinos.

Let $a$ denote the following operator in $L^2(\Sigma_1)$
\begin{equation}\nonumber
 a = \frac12 (p_2\cdot i \nabla_{p_2} + i \nabla_{p_2}\cdot p_2)\ .
\end{equation}
The operator $a$ is essentially self-adjoint on $C_0^\infty(\R^3,
\C^2)$. Its second quantized version $\d\Gamma(a)$ is a
self-adjoint operator in $\gF_a(L^2(\Sigma_1))$. From the
definition \eqref{eq:2.4} of the space $\gF_\ell$, the following
operator in $\gF_\ell$
\begin{equation}\nonumber
 A_\ell = \1\otimes\1\otimes\d\Gamma(a)\otimes\1
 + \1\otimes\1\otimes\1\otimes\d\Gamma(a)
\end{equation}
is essentially self-adjoint on $\gD_L$.

Let now $A$ be the following operator in $\gF_L$
\begin{equation}\nonumber
  A = A_1\otimes\1_2\otimes\1_3 + \1_1\otimes A_2\otimes\1_3
  + \1_1 \otimes\1_2 \otimes A_3 \ .
\end{equation}
Then $A$ is essentially self-adjoint on $\gD_L$.

We shall denote again by $A$ its extension to $\gF$. Thus $A$ is
essentially self-adjoint on $\gD$ and we still denote by $A$ its
closure.

We also set
\begin{equation}\nonumber
 \langle A \rangle = (1+A^2)^\frac12\ .
\end{equation}

We then have
\begin{theorem}\label{thm:3.4}
Suppose that the kernels $G_{\ell,\epsilon,\epsilon'}^{(\alpha)}$
satisfy Hypothesis~\ref{hypothesis:2.1} and \ref{hypothesis:3.1}.
For any $\delta>0$ satisfying $0<\delta<m_1$ there exists
$0<g_\delta\leq g_2$ such that, for $0< g \leq g_\delta$,

 \noindent $(i)$ The spectrum of $H$ in $(\inf\sigma(H),\, m_1-\delta]$ is
 purely absolutely continuous.

 \noindent $(ii)$ Limiting absorption principle.

 For every $s>1/2$ and $\vp$, $\psi$ in $\gF$, the limits
 \begin{equation}\nonumber
   \lim_{\varepsilon\rightarrow 0}
   (\vp,\ \la A \ra ^{-s} (H-\lambda \pm i\varepsilon)
   \la A \ra ^{-s} \psi)
 \end{equation}
exist uniformly for $\lambda$ in any compact subset of
$(\inf\sigma(H),\, m_1-\delta]$.

\noindent $(iii)$ Pointwise decay in time.

Suppose $s\in (\frac12,1)$ and $f\in C_0^\infty(\R)$ with
$\mathrm{supp} f \subset (\inf\sigma(H),\, m_1-\delta)$. Then
\begin{equation}\nonumber
  \| \la A \ra^{-s} \mathrm{e}^{-itH} f(H) \la A \ra ^{-s}\| =
  \mathcal{O}({t^{\frac12-s}})\ ,
\end{equation}
as $t\rightarrow\infty$.
\end{theorem}

The proof of Theorem~\ref{thm:3.4} is based on a positive
commutator estimate, called the Mourre estimate and on a
regularity property of $H$ with respect to $A$ (see
\cite{Mourre1981}, \cite{Amreinetal1996} and \cite{Sahbani1997}).
According to \cite{Frohlichetal2008}, the main ingredient of the
proof are auxiliary operators associated with infrared cutoff
Hamiltonians with respect to the momenta of the neutrinos that we
now introduce.

Let $\chi_0(.)$, $\chi_\infty(.)\in C^\infty(\R, [0,1])$ with
$\chi_0=1$ on $(-\infty, 1]$, $\chi_\infty =1$ on $[2, \infty)$
and $\chi_0{}^2 + \chi_\infty{}^2 =1$.

For $\sigma>0$ we set
\begin{equation}\label{def:chitilde2}
\begin{split}
 &  \chi_\sigma(p) = \chi_0(|p|/\sigma)\ , \\
 &  \chi^\sigma(p) = \chi_\infty(|p|/\sigma)\ ,\\
 &  \tilde\chi^\sigma(p) = 1 - \chi_\sigma(p)\ ,
\end{split}
\end{equation}
where $p\in\R^3$.

The operator $H_{I, \sigma}$ is the interaction given by
\eqref{eq:2.35}, \eqref{eq:2.36} and \eqref{eq:2.37} and
associated with the kernels $\tilde{\chi}^\sigma(p_2)
G^{(\alpha)}_{\ell,\epsilon,\epsilon'}(\xi_1,\xi_2,\xi_3)$. We
then set
\begin{equation}\nonumber
  H_\sigma := H_0 + g H_{I,\sigma}\ .
\end{equation}

Let
\begin{eqnarray*}
 & & \Sigma_{1,\sigma} = \Sigma_1 \cap \{(p_2, s_2);\
 |p_2|<\sigma\}\ , \\
 & & \Sigma_1^\sig = \Sigma_1 \cap \{ (p_2,s_2);\ |p_2|\geq
 \sigma\} \\
 & & \gF_{\ell,2,\sigma} = \gF_a(L^2(\Sigma_{1,\sigma})) \otimes
 \gF_a(L^2(\Sigma_{1,\sigma})) \ ,\\
 & & \gF_{\ell,2}^\sig = \gF_a(L^2(\Sigma_1^\sig)) \otimes
 \gF_a(L^2(\Sigma_1^\sig)) \ ,\\
 & & \gF_{\ell,2} = \gF_{\ell,2,\sigma} \otimes \gF_{\ell,2}^\sig\
 ,\\
 & & \gF_{\ell,1} = \bigotimes^2 \gF_a(L^2(\Sigma_1))\ .
\end{eqnarray*}
The space $\gF_{\ell,1}$ is the Fock space for the massive leptons
$\ell$ and $\gF_{\ell,2}$ is the Fock space for the neutrinos and
antineutrinos $\ell$.

Set
\begin{eqnarray*}
 & & \gF_\ell^\sig   = \gF_{\ell,1} \otimes \gF_{\ell,2}^\sig\
 ,\\
 & & \gF_{\ell,\sigma} = \gF_{\ell,2,\sigma}\
 .
\end{eqnarray*}
We have
\begin{equation*}
 \gF_\ell \simeq \gF_\ell^\sig \otimes \gF_{\ell,\sigma}\ .
\end{equation*}
Set
\begin{equation*}
\begin{split}
 &\gF_L^\sig  = \bigotimes_{\ell=1}^3 \gF_\ell^\sig \ ,\\
 &\gF_{L,\sigma} = \bigotimes_{\ell=1}^3 \gF_{\ell,\sigma}\ .
\end{split}
\end{equation*}
We have
\begin{equation*}
 \gF_L \simeq \gF_L^\sig \otimes \gF_{L,\sigma}\ .
\end{equation*}
Set
\begin{equation*}
\begin{split}
 & \gF^\sig = \gF_L^\sig \otimes\gF_W\ , \\
\end{split}
\end{equation*}
We have
\begin{equation*}
 \gF \simeq \gF_{L,\sigma} \otimes \gF^\sig \ .
\end{equation*}
Set
\begin{equation*}
\begin{split}
 & H_0^{(1)} =
 \sum_{\ell=1}^3\sum_{\epsilon = \pm} \int w_\ell^{(1)}(\xi_1)\,
 b^*_{\ell,\epsilon}(\xi_1) b_{\ell,\epsilon} (\xi_1)
 \d\xi_1\ , \\
 & H_0^{(2)} =
 \sum_{\ell=1}^3\sum_{\epsilon = \pm} \int w_\ell^{(2)}(\xi_2)\,
 c^*_{\ell,\epsilon}(\xi_2) c_{\ell,\epsilon} (\xi_2)
 \d\xi_2\ , \\
 & H_0^{(3)} =
 \sum_{\epsilon = \pm} \int w^{(3)}(\xi_3)
 a^*_{\epsilon}(\xi_3) a_{\epsilon} (\xi_3)
 \d\xi_3\ , \\
\end{split}
\end{equation*}
and
\begin{equation*}
\begin{split}
 & H_0^{(2)\sig} =
 \sum_{\ell=1}^3\sum_{\epsilon = \pm} \int_{|p_2|>\sigma} w_\ell^{(2)}(\xi_2)\,
 c^*_{\ell,\epsilon}(\xi_2) c_{\ell,\epsilon} (\xi_2)
 \d\xi_2\ , \\
 & H_{0,\sigma}^{(2)} =
 \sum_{\ell=1}^3\sum_{\epsilon = \pm} \int_{|p_2|\leq\sigma} w_\ell^{(2)}(\xi_2)\,
 c^*_{\ell,\epsilon}(\xi_2) c_{\ell,\epsilon} (\xi_2)
 \d\xi_2\ .
\end{split}
\end{equation*}
We have on $\gF^\sig\otimes\gF_\sigma$
\begin{equation*}
  H_0^{(2)} = H_0^{(2)\sigma} \otimes \1_\sigma +
  \1^\sig\otimes H_{0,\sigma}^{(2)}\ .
\end{equation*}
Here, $\1^{\sigma}$ (resp. $\1_\sigma$) is the identity operator
on $\gF^{\sigma}$ (resp. $\gF_\sigma$).

Define
\begin{equation}\label{def:h0n}
 H^{\sigma} = H_\sigma|_{\gF^{\,\sigma}}\quad\mbox{and}\quad
 H_0^{\,\sigma} = H_0|_{\gF^\sigma}\ .
\end{equation}

We get
\begin{equation*}
  H^{\sigma} = H_0^{(1)} + H_0^{(2)\,\sigma} + H_0^{(3)}+
  g H_{I,\sigma}\quad\mbox{on}\ \gF^{\,\sigma}\ ,
\end{equation*}
and
\begin{equation*}
  H_\sigma = H^{\sigma}\otimes\1_\sigma +
  \1^{\,\sigma}\otimes H_{0,\sigma}^{(2)}
  \quad\mbox{on}\ \gF^{\,\sigma}\otimes\gF_\sigma\ .
\end{equation*}
In order to implement the conjugate operator theory we have to
show that $H^{\,\sigma}$ has a gap in its spectrum above its
ground state.

We now set, for $\beta>0$ and $\eta>0$,
\begin{equation}\label{eq:3.33}
 C_{\beta\,\eta} =
 \left( \frac{3}{m_W} (1 + \frac{1}{m_1{}^2})
 + \frac{3\beta}{m_W m_1{}^2} +
 \frac{12\,\eta}{m_1{}^2}(1+\beta)\right)^\frac12\ ,
\end{equation}
and
\begin{equation}\label{eq:3.34}
 B_{\beta\,\eta} =
 \left( \frac{3}{m_W} (1 + \frac{1}{4\beta})
 + 12 (\,\eta(1+ \frac{1}{4\beta}) + \frac{1}{4\eta}\,)
 \right)^\frac12\ .
\end{equation}
Let
\begin{equation}\label{eq:3.35}
 G = \left( G^{(\alpha)}_{\ell,\epsilon,\epsilon'}(.,.,.)
 \right)_{\alpha=1,2;
 \ell=1,2,3;\epsilon,\epsilon'=\pm,\epsilon\neq\epsilon'}
\end{equation}
and set
\begin{equation}\label{eq:3.36}
 K(G) = \left( \sum_{\alpha=1,2} \sum_{\ell=1}^3
 \sum_{\epsilon\neq\epsilon'} \|
 G^{(\alpha)}_{\ell,\epsilon,\epsilon'}
 \|^2_{L^2(\Sigma_1\times\Sigma_1\times\Sigma_2)}\right)^\frac12 \
 .
\end{equation}
Let
\begin{equation}\label{eq:3.38}
 \tilde{C}_{\beta\eta} = C_{\beta\eta}
 \left(1 + \frac{g_1 K(G) C_{\beta\eta}}{1 - g_1 K(G)
 C_{\beta\eta}}\right) ,
\end{equation}
\begin{equation}\label{eq:3.39}
 \tilde{B}_{\beta\eta} =
 \Big(\, 1 + \frac{g_1\, K(G) C_{\beta\eta}}{1-g_1\, K(G)\, C_{\beta\eta}}
 (\, 2 + \frac{g_1 K(G) B_{\beta\eta} C_{\beta\eta}}{ 1 - g_1 K(G)
 C_{\beta\eta}}\,)\, \Big)B_{\beta\eta}\ .
\end{equation}



Let
\begin{equation}\nonumber
 \tilde{K}(G) = \left( \sum_{\alpha=1,2} \sum_{\ell=1}^3
 \sum_{\epsilon\neq\epsilon'}
 \int_{\Sigma_1\times\Sigma_1\times\Sigma_2}
 \frac{|G^{(\alpha)}_{\ell,\epsilon,\epsilon'}(\xi_1,\xi_2,\xi_3)|^2}
 {|p_2|^2} \d\xi_1 \d\xi_2 \d\xi_3\right)^\frac12\ .
\end{equation}
Let $\delta\in\R$ be such that
\begin{equation}\nonumber
 0 < \delta < m_1\ .
\end{equation}

We set
\begin{equation}\label{def:Dtilde}
   \tilde D = \, \sup (\frac{4\Lambda\gamma}{2m_1 - \delta},\, 1)\,
   \tilde K(G)\,
   (\, 2m_1\, \tilde C_{\beta\eta} + \tilde B_{\beta\eta}\, )\ ,
\end{equation}
where $\Lambda>m_1$ has been introduced in
Hypothesis~\ref{hypothesis:3.1}(iv).

Let us define the sequence $(\sigma_n)_{n\geq 0}$ by
\begin{equation}\nonumber
\begin{split}
 & \sigma_0 = \Lambda\ ,\\
 & \sigma_1 = m_1 -\frac{\delta}{2}\ , \\
 & \sigma_2 = m_1 -\delta = \gamma\sigma_1\ ,\\
 & \sigma_{n+1} = \gamma\sigma_n,\ n\geq 1\ ,
\end{split}
\end{equation}
where $\gamma = 1 - \delta/(2m_1 -\delta)$.

Let $g_\delta^{(1)}$ be such that
\begin{equation}\nonumber
 0 < g^{(1)}_\delta < \inf
 (1, g_1, \frac{\gamma-\gamma^2}{3 \tilde D})\ .
\end{equation}
For $0 < g \leq g_\delta^{(1)}$ we have
\begin{equation}\nonumber
 0 < \gamma < (1-\frac{3 g \tilde D}{\gamma})\ ,
\end{equation}
and
\begin{equation}\label{eq:prop-sigman}
 0 < \sigma_{n+1} < (1 - \frac{3 g \tilde D}{\gamma})\sigma_n,\quad n\geq
 1\ .
\end{equation}
Set
\begin{equation}\nonumber
\begin{split}
 & H^n = H^{\sigma_n};\quad H_0^n = H_0^{\sigma_n},\quad n\geq 0\, \\
 & E^n = \inf\sigma(H^n)\, ,\quad n\geq 0\ .
\end{split}
\end{equation}
We then get
\begin{proposition}\label{proposition:3.5}
Suppose that the kernels $G^{(\alpha)}_{\ell,\epsilon,\epsilon'}$
satisfy Hypothesis~\ref{hypothesis:2.1},
Hypothesis~\ref{hypothesis:3.1}(i) and \ref{hypothesis:3.1}(iv).
Then there exists $0<\tilde g_\delta \leq g^{(1)}_\delta$ such
that, for $g\leq \tilde g_\delta$ and $n\geq 1$, $E^n$ is a simple
eigenvalue of $H^n$ and $H^n$ does not have spectrum in $(\,E^n,\,
E^n+ (1-\frac{3g\tilde D}{\gamma})\sigma_n\, )$.
\end{proposition}
The proof of Proposition~\ref{proposition:3.5} is given in
Appendix~A.

We now introduce the positive commutator estimates and the
regularity property of $H$ with respect to $A$ in order to prove
Theorem~\ref{thm:3.4}

The operator $A$ has to be split into two pieces depending on
$\sigma$.

Let

\begin{eqnarray*}
 & & \eta_\sigma(p_2) = \chi_{2\sigma} (p_2)\ , \\
 & & \eta^{\sigma}(p_2) = \chi^{2\sigma} (p_2)\ , \\
 & & a_\sigma = \eta_\sigma(p_2)\, a \, \eta_\sigma(p_2)\ ,\\
 & & a^\sigma = \eta ^\sigma (p_2)\, a \, \eta^\sigma(p_2)\ .
\end{eqnarray*}
Since $\eta_\sigma^2 + (\eta^\sigma)^2 = 1$, and $[\eta_\sigma,\,
[\eta_\sigma,\, a]\,] = 0 = [\eta^\sigma,\, [\eta^\sigma,\, a]\,
]$, we obtain (see \cite{Frohlichetal2008})
\begin{equation}\nonumber
  a = a^\sigma + a_\sigma\ .
\end{equation}

Note that we also have
\begin{eqnarray*}
 & & a_\sigma = \frac12\left( \eta_\sigma(p_2)^2 p_2\cdot
 i\nabla p_2
 + i\nabla p_2.\eta_\sigma(p_2)^2 p_2 \right)\ ,\\
 & & a^\sigma = \frac12\left( \eta^\sigma(p_2)^2 p_2\cdot
 i\nabla p_2
 + i\nabla p_2.\eta^\sigma(p_2)^2 p_2 \right)\ .
\end{eqnarray*}

The operators $a$, $a_\sigma$ and $a^\sigma$ are essentially
self-adjoint on $C_0^\infty(\R^3,\, \C^2)$ (see
\cite[Proposition~4.2.3]{Amreinetal1996}). We still denote by $a$,
$a_\sigma$ and $a^\sigma$ their closures. If $\tilde a$ denotes
any of the operator $a$, $a_\sigma$ and $a^\sigma$, we have
 $$
\cD(\tilde a) = \{\ u\in L^2(\Sigma_1);\ \tilde a u \in
L^2(\Sigma_1)\ \}\ .
 $$

The operators $\d\Gamma(a)$, $\d\Gamma(a^\sigma)$,
$\d\Gamma(a_\sigma)$ are self-adjoint operators in
$\gF_a(L^2(\Sigma_1))$ and we have
\begin{equation}\nonumber
 \d\Gamma(a) = \d\Gamma (a^\sigma) + \d\Gamma(a_\sigma)\ .
\end{equation}

By \eqref{eq:2.4}, the following operators in $\gF_\ell$, denoted
by $A_\ell^{\,\sigma}$ and $A_{\sigma\ell}$ respectively,
\begin{equation}\nonumber
 A_\ell^{\,\sigma} = \1\otimes\1\otimes\d\Gamma(a^\sigma)\otimes\1
 + \1\otimes\1\otimes\1\otimes\d\Gamma(a^\sigma)\ ,
\end{equation}
\begin{equation}\nonumber
 A_{\sigma\ell} = \1\otimes\1\otimes\d\Gamma(a_\sigma)\otimes\1
 + \1\otimes\1\otimes\1\otimes\d\Gamma(a_\sigma)\ ,
\end{equation}
are essentially self-adjoint on $\gD_\ell$.

Let $A^\sigma$ and $A_\sigma$ be the following two operators in
$\gF_L$,
\begin{equation}\nonumber
 A^\sigma =  A_1^{\,\sigma}\otimes\1_2\otimes\1_3
           + \1_1\otimes A_2^{\,\sigma}\otimes\1_3
           + \1_1\otimes\1_2\otimes A_3^{\,\sigma}\ ,
\end{equation}
\begin{equation}\nonumber
 A_\sigma =  A_{\sigma 1}\otimes\1_2\otimes\1_3
           + \1_1\otimes A_{\sigma 2}\otimes\1_3
           + \1_1\otimes\1_2\otimes A_{\sigma 3}.
\end{equation}
The operators $A^\sigma$ and $A_\sigma$ are essentially
self-adjoint on $\gD_L$. Still denoting by $A^\sigma$ and
$A_\sigma$ their extensions to $\gF$, $A^\sigma$ and $A_\sigma$
are essentially self-adjoint on $\gD$ and we still denote by
$A^\sigma$ and $A_\sigma$ their closures.

We have
\begin{equation}\nonumber
  A = A^\sigma + A_\sigma\ .
\end{equation}

The operators $a$, $a^\sigma$ and $a_\sigma$ are associated to the
following $C^\infty$-vector fields in $\R^3$ respectively,
\begin{equation}\label{eq:3.63}
\begin{split}
 & v(p_2) = p_2\ , \\
 & v^\sigma(p_2) = \eta^\sigma(p_2)^2 p_2\ ,\\
 & v_\sigma(p_2) = \eta_\sigma(p_2)^2 p_2\ .
\end{split}
\end{equation}
Let $\cV(p)$ be any of these vector fields. We have
\begin{equation}\nonumber
 | \cV(p)| \leq \Gamma\, |p|\ ,
\end{equation}
for some $\Gamma>0$ and we also have
\begin{equation}\label{eq:3.65}
 \cV(p) = \tilde v(|p|) p\ ,
\end{equation}
where the $\tilde{v}$'s are defined by \eqref{eq:3.63} and
\eqref{eq:3.65}, and fulfill $|p|^\alpha \frac{\d^\alpha}{\d
|p|^\alpha} \tilde v(|p|)$ bounded for $\alpha=0,1,2$.

Let $\psi_t(.):\, \R^3\rightarrow\R^3$ be the corresponding flow
generated by $\cV$:
\begin{equation}\nonumber
\begin{split}
 & \frac{\d}{\d t} \psi_t(p) = \cV(\psi_t(p))\ ,\\
 & \psi_0(p) = p\ .
\end{split}
\end{equation}
$\psi_t(p)$ is a $C^\infty$-flow and we have
\begin{equation}\label{eq:3.69}
  \mathrm{e}^{-\Gamma |t|}\, |p| \leq |\psi_t(p)| \leq
  \mathrm{e}^{\Gamma |t|}\, |p|\ .
\end{equation}
$\psi_t(p)$ induces a one-parameter group of unitary operators
$U(t)$ in $L^2(\Sigma_1)\simeq L^2(\R^3,\, \C^2)$ defined by
\begin{equation}\nonumber
 (U(t) f)(p) = f(\psi_t(p)) (\det \nabla\psi_t(p))^\frac12
\end{equation}
Let $\phi_t(.)$, $\phi^{\,\sigma}_t(.)$ and $\phi_{\sigma t}(.)$
be the flows associated with the vector fields $v(.)$,
$v^\sigma(.)$ and $v_\sigma(.)$ respectively.

Let $U(t)$, $U^{\sigma}(t)$ and $U_\sigma(t)$ be the corresponding
one-parameter groups of unitary operators in $L^2(\Sigma_1)$. The
operators $a$, $a^\sigma$, and $a_\sigma$ are the generators of
$U(t)$, $U^{\sigma}(t)$ and $U_\sigma(t)$ respectively, i.e.,
\begin{equation}\nonumber
\begin{split}
 & U(t) = \mathrm{e}^{-i a t}\ , \\
 & U^\sigma(t) = \mathrm{e}^{-i a^\sigma t} \ ,  \\
 & U_\sigma(t) = \mathrm{e}^{-i a_\sigma t} \ .
 \end{split}
\end{equation}

Let
\begin{equation}\nonumber
 w^{(2)}(\xi_2) = (w^{(2)}_\ell(\xi_2))_{\ell=1,2,3}
\end{equation}
and
\begin{equation}\nonumber
 \d\Gamma (w^{(2)}) = \sum_{\ell=1}^3 \sum_\epsilon
 \int w^{(2)}_\ell (\xi_2) c_{\ell,\epsilon}^*(\xi_2) c_{\ell
 \epsilon}(\xi_2) \d\xi_2\ .
\end{equation}
Let $V(t)$ be any of the one-parameter groups $U(t)$,
$U^\sigma(t)$ and $U_\sigma(t)$. We set
\begin{equation}\nonumber
 V(t) w^{(2)} V(t)^* = (V(t) w_\ell^{(2)} V(t)^*)_{\ell=1,2,3}\ ,
\end{equation}
and we have
\begin{equation}\nonumber
 V(t) w^{(2)} V(t)^* = w^{(2)} (\psi_t)\ .
\end{equation}
Here $\psi_t$ is the flow associated to $V(t)$.

This yields, for any $\vp\in\gD$, (see
\cite[Lemma~2.8]{DerezinskiGerard1999})
\begin{equation}\label{eq:3.76}
\begin{split}
  \mathrm{e}^{-i A t} H_0 \mathrm{e}^{i A t} \vp - H_0 \vp
  & = (\d\Gamma (\mathrm{e}^{-i a t} w^{(2)} \mathrm{e}^{i a t})
  - \d\Gamma( w^{(2)})) \vp \\
  & = (\d\Gamma ( w^{(2)} \circ \phi_t - w^{(2)})) \vp\ ,
\end{split}
\end{equation}
\begin{equation}\label{eq:3.77}
\begin{split}
  \mathrm{e}^{-i A^\sigma t} H_0 \mathrm{e}^{i A^\sigma t} \vp - H_0 \vp
  & = (\d\Gamma (\mathrm{e}^{-i a^\sigma t} w^{(2)} \mathrm{e}^{i a^\sigma t})
  - \d\Gamma( w^{(2)})) \vp \\
  & = (\d\Gamma ( w^{(2)} \circ \phi_t^{\,\sigma} - w^{(2)})) \vp\ ,
\end{split}
\end{equation}
\begin{equation}\label{eq:3.78}
\begin{split}
  \mathrm{e}^{-i A_\sigma t} H_0 \mathrm{e}^{i A_\sigma t} \vp - H_0 \vp
  & = (\d\Gamma (\mathrm{e}^{-i a_\sigma t} w^{(2)} \mathrm{e}^{i a_\sigma t})
  - \d\Gamma( w^{(2)})) \vp \\
  & = (\d\Gamma ( w^{(2)} \circ \phi_{\sigma t} - w^{(2)})) \vp\ .
\end{split}
\end{equation}

\begin{proposition}\label{proposition:3.6}
Suppose that the kernels $G^{(\alpha)}_{\ell,\epsilon,\epsilon'}$
satisfy Hypothesis~\ref{hypothesis:2.1}.

For every $t\in\R$ we have, for $g\leq g_1$,
\begin{equation}\nonumber
\begin{split}
 (i) &\quad \mathrm{e}^{it A} \cD(H_0) = \mathrm{e}^{itA} \cD(H)
 \subset \cD(H_0) = \cD(H)\ , \\
 (ii) & \quad \mathrm{e}^{it A^\sigma} \cD(H_0) =
 \mathrm{e}^{itA^\sigma} \cD(H)
 \subset \cD(H_0) = \cD(H)\ , \\
 (iii) & \quad \mathrm{e}^{it A_\sigma} \cD(H_0) =
 \mathrm{e}^{itA_\sigma} \cD(H)
 \subset \cD(H_0) = \cD(H)\ .
\end{split}
\end{equation}
\end{proposition}
%
%
%
\begin{proof}
We only prove $i)$, since $ii)$ and $iii)$ can be proved
similarly. By \eqref{eq:3.76} we have, for $\vp\in\gD$,
\begin{equation}\label{eq:3.84}
 \mathrm{e}^{- i t A} H_0 \mathrm{e}^{i t A} \varphi =
 (H_0^{(1)} + H_0^{(3)} + \d\Gamma (w^{(2)}\circ \phi_t)) \vp\ .
\end{equation}
It follows from \eqref{eq:3.69} and \eqref{eq:3.84} that
\begin{equation}\nonumber
 \| H_0 \mathrm{e}^{i t A} \vp\| \leq \mathrm{e}^{\Gamma |t|}
 \| H_0 \vp\|\ .
\end{equation}
This yields $i)$ because $\gD$ is a core for $H_0$. Moreover we
get
\begin{equation}\nonumber
 \| H_0 \mathrm{e}^{i t A} (H_0 + 1)^{-1} \| \leq
 \mathrm{e}^{\Gamma |t|}\ .
\end{equation}
In view of $\gD(H_0) = \gD(H)$, the operators $H_0(H+i)^{-1}$ and
$H(H_0+i)^{-1}$ are bounded and there exists a constant $C>0$ such
that
\begin{equation}\nonumber
  \| H \mathrm{e}^{it A} (H+i)^{-1} \| \leq C \mathrm{e}^{\Gamma
  |t|}\ .
\end{equation}
Similarly, we also get
\begin{equation}\nonumber
\begin{split}
  & \| H_0 \mathrm{e}^{it A^\sigma} (H_0 + 1)^{-1} \|
  \leq \mathrm{e}^{\Gamma
  |t|}\ , \\
  & \| H_0 \mathrm{e}^{it A_\sigma} (H_0 + 1)^{-1} \| \leq \mathrm{e}^{\Gamma
  |t|}\ , \\
  & \| H \mathrm{e}^{it A^\sigma} (H+i)^{-1} \| \leq C \mathrm{e}^{\Gamma
  |t|}\ , \\
  & \| H \mathrm{e}^{it A_\sigma} (H+i)^{-1} \| \leq C \mathrm{e}^{\Gamma
  |t|}\ .
\end{split}
\end{equation}
\end{proof}

Let $H_I(G)$ be the interaction associated with the kernels $G = (
G_{\ell,\epsilon,\epsilon'}^{(\alpha)} )_{\alpha=1,2;\
\ell=1,2,3;\ \epsilon\neq\epsilon' = \pm}$, where the kernels
$G_{\ell,\epsilon,\epsilon'}^{(\alpha)} )$ satisfy
Hypothesis~\ref{hypothesis:2.1}

We set
\begin{equation}\nonumber
 V(t) G = (V(t)
 G_{\ell,\epsilon,\epsilon'}^{(\alpha)})_{\alpha=1,2;\
 \ell=1,2,3;\ \epsilon\neq\epsilon '=\pm}
\end{equation}
We have for $\vp\in\gD$ (see
\cite[Lemma~2.7]{DerezinskiGerard1999}),
\begin{equation}\label{eq:3.81}
\begin{split}
 & \mathrm{e}^{- i A t} H_I(G) \mathrm{e}^{ i A t} \vp
 = H_I( \mathrm{e}^{- i a t} G) \vp\ , \\
 & \mathrm{e}^{- i A^\sigma t} H_I(G) \mathrm{e}^{ i A^\sigma t} \vp
 = H_I( \mathrm{e}^{- i a^\sigma t} G) \vp\ , \\
 & \mathrm{e}^{- i A_\sigma t} H_I(G) \mathrm{e}^{ i A_\sigma t} \vp
 = H_I( \mathrm{e}^{- i a_\sigma t} G) \vp\ .
\end{split}
\end{equation}

According to \cite{Amreinetal1996} and \cite{Sahbani1997}, in
order to prove Theorem~\ref{thm:3.4} we must prove that $H$ is
locally of class $C^2(A^\sigma)$, $C^2(A_\sigma)$ and $C^2(A)$ in
$(-\infty, m_1-\frac{\delta}{2})$ and that $A$ and $A_\sigma$ are
locally strictly conjugate to $H$ in $(E, m_1-\frac{\delta}{2})$.

Recall that $H$ is locally of class $C^2(A)$ in $(-\infty,
m_1-\frac{\delta}{2})$ if, for any $\vp\in C_0^\infty ((-\infty,
m_1-\frac{\delta}{2}))$, $\vp(H)$ is of class $C^2(A)$, i.e.,
$t\rightarrow \mathrm{e}^{-i A t} \vp(H) \mathrm{e}^{i t A} \psi$
is twice continuously differentiable for all $\vp\in
C_0^\infty((-\infty, m_1-\frac{\delta}{2})$ and all $\psi\in\gF$.

Thus, one of our main results is the following one
\begin{theorem}\label{thm:3.7} Suppose that the kernels
$G_{\ell,\epsilon,\epsilon'}^{(\alpha)}$ satisfy
Hypothesis~\ref{hypothesis:2.1} and \ref{hypothesis:3.1}(i)-(iii).
\begin{itemize}
 \item[(a)] $H$ is locally of class $C^2(A)$,
 $C^2(A^\sigma)$ and $C^2(A_\sigma)$ in $(-\infty,
 m_1-{\delta}/{2})$.
 \item[(b)] $H^\sigma$ is locally of class
 $C^2(A^\sigma)$ in $(-\infty, m_1-{\delta}/{2})$.
\end{itemize}
\end{theorem}
%
%
%

It follows from Theorem~\ref{thm:3.7} that $[H,\, iA]$, $[H,\,
iA_\sigma]$, $[H, \, iA^\sigma]$ and $[H^\sigma,\, iA^\sigma]$ are
defined as sesquilinear forms on $\cup_K E_K(H) \gF$, where the
union is taken over all the compact subsets $K$ of $(-\infty,
m_1-\delta/2)$.

Furthermore, by Proposition~\ref{proposition:3.6},
Theorem~\ref{thm:3.7} and \cite[Lemma~29]{Frohlichetal2008}, we
get for all $\vp\in C_0^\infty((E, m_1-\delta/2))$ and all
$\psi\in\gF$,
\begin{equation}\label{eq:3.92}
\begin{split}
 & \vp(H)\,  [H,\, iA]\,  \vp(H)\, \psi =
 \lim_{t\rightarrow 0} \vp(H)\,
 \big[H,\, \frac{\mathrm{e}^{i t A }-1}{t}\big]\,
 \vp(H)\,  \psi \ , \\
 & \vp(H)\,  [H,\, iA_\sigma]\,  \vp(H)\, \psi =
 \lim_{t\rightarrow 0} \vp(H)\,
 \big[H,\, \frac{\mathrm{e}^{i t A_\sigma }-1}{t}\big]\,
 \vp(H)\,  \psi \ , \\
 & \vp(H)\, [H,\, iA^\sigma]\, \vp(H)\, \psi =
 \lim_{t\rightarrow 0} \vp(H)\,
 \big[H,\, \frac{\mathrm{e}^{i t A^\sigma }-1}{t}\big]\,
 \vp(H)\,  \psi \ , \\
 & \vp(H^\sigma)\,  [H^\sigma,\, i A^\sigma]\,  \vp(H^\sigma)
 \, \psi = \lim_{t\rightarrow 0} \vp(H^\sigma)\,
 \big[H^\sigma,\, \frac{\mathrm{e}^{i t A^\sigma }-1}{t}\big]\,
 \vp(H^\sigma)\,  \psi \ . \\
\end{split}
\end{equation}

The following proposition allows us to compute $[H,\, iA]$, $[H,\,
iA^\sigma]$, $[H,\, iA_\sigma]$ and $[H^\sigma,\, iA^\sigma]$ as
sesquilinear forms. By Hypothesis~\ref{hypothesis:2.1} and
\ref{hypothesis:3.1}~(iii.a), the kernels $G_{\ell, \epsilon,
\epsilon'}^{(\alpha)} (\xi_1, ., \xi_3)$ belong to the domains of
$a$, $a^\sigma$, and $a_\sigma$.

\begin{proposition}\label{proposition:3.8}$\ $
Suppose that the kernels $G^{(\alpha)}_{\ell,\epsilon,\epsilon'}$
satisfy Hypothesis~\ref{hypothesis:2.1} and
\ref{hypothesis:3.1}~(iii.a). Then
\begin{itemize}
 \item[(a)] For all $\psi\in\cD(H)$ we have
   \begin{itemize}
     \item[$(i)$] $\lim_{t\rightarrow0} \big[H, \frac{\mathrm{e}^{it A}
     -1}{t}\big] \psi = \big(\,\d\Gamma (w^{(2)}) + g H_I(-ia
     G)\,\big)\psi$,
     \item[$(ii)$] $\lim_{t\rightarrow0} \big[H, \frac{\mathrm{e}^{it A^\sigma}
     -1}{t}\big] \psi = \big(\,\d\Gamma ( (\eta^{\sigma})^2w^{(2)})
     + g H_I(-i a^\sigma G)\,\big)\psi$,
     \item[$(iii)$] $\lim_{t\rightarrow0} \big[H, \frac{\mathrm{e}^{it A_\sigma}
     -1}{t}\big] \psi = \big(\,\d\Gamma ( (\eta_{\sigma})^2w^{(2)})
     + g H_I(-i a_\sigma G)\,\big)\psi$,
     \item[$(iv)$] $\lim_{t\rightarrow0} \big[H^\sigma,
     \frac{\mathrm{e}^{it A^\sigma}
     -1}{t}\big] \psi = \big(\,\d\Gamma ( (\eta^{\sigma})^2w^{(2)})
     + g H_I(-i a^\sigma (\tilde\chi^\sigma(p_2)G))\,\big)\psi$.
   \end{itemize}
 \item[(b)]
   \begin{itemize}
     \item[$(i)$] $\sup_{0<|t|\leq 1} \big\| \big[ H,
     \frac{\mathrm{e}^{it A}-1}{t} \big] (H+i)^{-1} \big\| <
     \infty$,
     \item[$(ii)$] $\sup_{0<|t|\leq 1} \big\| \big[ H,
     \frac{\mathrm{e}^{it A^\sigma}-1}{t} \big] (H+i)^{-1} \big\| <
     \infty$,
     \item[$(iii)$] $\sup_{0<|t|\leq 1} \big\| \big[ H,
     \frac{\mathrm{e}^{it A_\sigma}-1}{t} \big] (H+i)^{-1} \big\| <
     \infty$,
     \item[$(iv)$] $\sup_{0<|t|\leq 1} \big\| \big[ H^\sigma,
     \frac{\mathrm{e}^{it A^\sigma}-1}{t} \big] (H+i)^{-1} \big\| <
     \infty$.
   \end{itemize}
\end{itemize}
\end{proposition}
%
%
%
\begin{proof}
Part $\mathrm{(b)}$ follows from part $\mathrm{(a)}$ by the
uniform boundedness principle. For part $\mathrm{(a)}$, we only
prove $\mathrm{(a)}$(i), since other statements can be proved
similarly.

By \eqref{eq:3.69}, we obtain
\begin{equation}\nonumber
 \frac{1}{|t|} \big| w_\ell^{(2)} (\phi_t(p_2)) -
 w_\ell^{(2)}(p_2) \big| \leq \frac{1}{|t|}
 \big( \mathrm{e}^{\Gamma\, |t|} -1\big) w_\ell^{(2)}(p_2)\ ,
\end{equation}
for $\ell=1,2,3$.

By \eqref{eq:3.76}-\eqref{eq:3.78} and the Lebesgue's Theorem we
then get for all $\psi\in\cD(H_0)$
\begin{equation}\nonumber
\begin{split}
 & \lim_{t\rightarrow0} \big[ H_0, \frac{\mathrm{e}^{i t
 A}-1}{t}\big] \psi
 = \lim_{t\rightarrow0} \frac{1}{t} \big[ \mathrm{e}^{-i t A
 }H_0\mathrm{e}^{i t A} - H_0\big] \psi = \d\Gamma( w^{(2)}) \psi\
 , \\
 & \lim_{t\rightarrow0} \big[ H_0, \frac{\mathrm{e}^{i t
 A^\sigma}-1}{t}\big] \psi
 = \lim_{t\rightarrow0} \frac{1}{t} \big[ \mathrm{e}^{-i t A^\sigma
 }H_0\mathrm{e}^{i t A^\sigma} - H_0\big] \psi
 = \d\Gamma( (\eta^\sigma)^2 w^{(2)}) \psi\
 , \\
 & \lim_{t\rightarrow0} \big[ H_0, \frac{\mathrm{e}^{i t
 A_\sigma}-1}{t}\big] \psi
 = \lim_{t\rightarrow0} \frac{1}{t} \big[ \mathrm{e}^{-i t A_\sigma
 }H_0\mathrm{e}^{i t A_\sigma} - H_0\big] \psi =
 \d\Gamma( (\eta_\sigma)^2 w^{(2)}) \psi\
 .
\end{split}
\end{equation}
By \eqref{eq:3.81}, we obtain for all $\psi\in\cD(H)$,
\begin{equation}\nonumber
\begin{split}
 & \lim_{t\rightarrow 0} \big[ H_I(G), \frac{\mathrm{e}^{i t
 A}\!-1\!}{t}\big] \psi
 = \lim_{t\rightarrow0} \frac{1}{t} \big[ \mathrm{e}^{-i t A
 }H_I (G)\mathrm{e}^{i t A } - H_I (G) \big] \psi
 = H_I ( -i (a\, G))\psi, \\
 & \lim_{t\rightarrow 0} \big[ H_I(G), \frac{\mathrm{e}^{i t
 A^\sigma}\!-\!1}{t}\big] \psi
 = \lim_{t\rightarrow0} \frac{1}{t} \big[ \mathrm{e}^{-i t
 A^\sigma
 }H_I (G)\mathrm{e}^{i t A^\sigma } - H_I (G) \big] \psi
 = H_I ( -i (a^\sigma G))\psi, \\
 & \lim_{t\rightarrow 0} \big[ H_I(G), \frac{\mathrm{e}^{i t
 A_\sigma}\!-\!1}{t}\big] \psi
 = \lim_{t\rightarrow0} \frac{1}{t} \big[ \mathrm{e}^{-i t
 A_\sigma
 }H_I (G)\mathrm{e}^{i t A_\sigma } - H_I (G) \big] \psi
 = H_I ( -i (a_\sigma G))\psi, \\
 & \lim_{t\rightarrow 0} \big[ H_I(\tilde\chi^\sigma(p_2)G),
 \frac{\mathrm{e}^{i t A^\sigma }\!-\!1}{t}\big] \psi \\
 & = \lim_{t\rightarrow0} \frac{1}{t} \big[ \mathrm{e}^{-i t
 A^\sigma}H_I (\tilde\chi^\sigma(p_2) G)
 \mathrm{e}^{i t A^\sigma} - H_I (\tilde\chi^\sigma(p_2) G) \big] \psi
 = H_I ( -i (a^\sigma (\tilde\chi^\sigma(p_2) G))) \psi \ .
\end{split}
\end{equation}
This concludes the proof of Proposition~\ref{proposition:3.8}.
\end{proof}

Combining \eqref{eq:3.92} with Proposition~\ref{proposition:3.8},
we finally get for every $\varphi\in C_0^\infty((-\infty, m_1 -
\delta/2))$ and every $\psi\in\gF$
\begin{equation}\label{eq:3.104-added}
 \vp(H) \big[ H,\, iA\big] \vp(H) \psi
 = \vp(H) \big[ \d\Gamma(w^{(2)}) + gH_I (-i (a\, G))\big]
 \vp(H) \psi \ ,
\end{equation}
\begin{equation}\label{eq:3.105}
 \vp(H) \big[ H,\, iA^\sigma \big] \vp(H) \psi
 = \vp(H) \big[ \d\Gamma((\eta^\sigma)^2 w^{(2)}) + gH_I (-i (a^\sigma G))
 \big]
 \vp(H) \psi \ ,
\end{equation}
\begin{equation}\label{eq:3.106}
 \vp(H) \big[ H,\, iA_\sigma \big] \vp(H) \psi
 = \vp(H) \big[ \d\Gamma( (\eta_\sigma)^2 w^{(2)}) + gH_I (-i (a_\sigma G))
 \big]
 \vp(H) \psi \ ,
\end{equation}
and
\begin{equation}\label{eq:3.107}
 \vp(H^\sigma) \big[ H^\sigma,\, i A ^\sigma \big] \vp(H^\sigma) \psi
 = \vp(H^\sigma) \big[ \d\Gamma( (\eta^\sigma)^2 w^{(2)})
 + g H_I (-i (a^\sigma
 (\tilde\chi^\sigma G) ))\big]
 \vp(H^\sigma) \psi \ .
\end{equation}


We now introduce the Mourre inequality.

Let $N$ be the smallest integer such that
 $$
  N\gamma\geq 1 .
 $$

We have, for $g\leq g^{(1)}_\delta$,
\begin{equation}\label{eq:3.108}
\begin{split}
 & \gamma < \gamma + \frac{1}{N} (1 - \frac{3 g \tilde D}{\gamma}
 -\gamma) < 1 -\frac{3 g \tilde D}{\gamma}\ ,\\
 & \frac{\gamma}{N} \leq \gamma - \frac{1}{N} (1 - \frac{3 g \tilde
 D}{\gamma} - \gamma) < \gamma\ .
\end{split}
\end{equation}
Let
\begin{equation}\nonumber
 \epsilon_\gamma = \frac{1}{2N} ( 1 -\frac{ 3 g_\delta^{(1)}
 \tilde D}{\gamma} - \gamma)\ .
\end{equation}
We choose $f\in C_0^\infty(\R)$ such that $1\geq f \geq 0$ and
\begin{equation}\label{eq:3.109}
f(\lambda) = \left\{
\begin{array}{ll}
 1&  \mbox{ if } \lambda\in[(\gamma-\epsilon_\gamma)^2,
 \gamma+ \epsilon_\gamma]\ , \\
 0&  \mbox{ if } \lambda > \gamma+\frac{1}{N} (1 -
 \frac{3 g_\delta^{(1)}\tilde D}{\gamma} - \gamma) = \gamma + 2
 \epsilon_\gamma\ , \\
 0             &  \mbox{ if } \lambda < (\gamma -\frac{1}{N}
              (1 - \frac{ 3 g_\delta^{(1)} \tilde D}{\gamma}
              -\gamma))^2 = (\gamma - 2\epsilon_\gamma)^2\ .
\end{array}
\right.
\end{equation}
Note that $\gamma + 2\epsilon_\gamma < 1 - 3 g \tilde D/\gamma$
for $g\leq g_\delta^{(1)}$ and $\gamma - \epsilon_\gamma
>\gamma/N$.

We set, for $n\geq 1$,
\begin{equation}\nonumber
 f_n(\lambda) = f\left(\frac{\lambda}{\sigma_n}\right)\ .
\end{equation}

Let
\begin{equation}\nonumber
\begin{split}
& H_n = H_{\sigma_n}\ , \\
& E_n = \inf\sigma(H_n)\ , \\
& H_{0\, n}^{(2)} = H_{0\, \sigma_n}^{(2)}\ .
\end{split}
\end{equation}
Let $P^n$ denote the ground state projection of $H^n$. It follows
from proposition~\ref{proposition:3.5} that, for $n\geq 1$ and
$g\leq \tilde{g}_\delta \leq g_\delta^{(1)}$,
\begin{equation}\label{eq:3.112}
  f_n(H_n - E_n) = P^n\otimes f_n(H_{0,\, n}^{(2)})\ .
\end{equation}
Note that
\begin{equation}\label{eq:3.113}
 E_n = E^n = \inf\sigma(H^n)\ .
\end{equation}
Set
\begin{equation}\nonumber
\begin{split}
 & a^n = a^{\sigma_n} \ , \\
 & a_n = a_{\sigma_n} \ , \\
 & A^n = A^{\sigma_n} \ , \\
 & A_n = A_{\sigma_n} \ , \\
 & \gF^n = \gF^{\sigma_n}\ , \\
 & \gF_n = \gF_{\sigma_n}\ .
\end{split}
\end{equation}
We have
\begin{equation}\nonumber
\begin{split}
 & \gF \simeq \gF^n \otimes \gF_n\ , \\
 & A =  A^n + A_n \ .
\end{split}
\end{equation}
We further note that
\begin{equation}\label{eq:3.117}
 a^n \tilde\chi^{\sigma_n} ( p_2 ) = a^n\ .
\end{equation}
By \eqref{eq:3.105}, \eqref{eq:3.107} and \eqref{eq:3.117}, we
obtain
\begin{equation}\nonumber
  [H, iA^n] = [H^n, i A^n] \otimes \1\ ,
\end{equation}
as sesquilinear forms with respect to $\gF = \gF^n \otimes\gF_n$.

Furthermore, it follows from the virial Theorem (see
\cite[Proposition~3.2]{Sahbani1997} and
Proposition~\ref{proposition:6.1}) that
\begin{equation}\label{eq:3.119}
 P^n [H^n, iA^n]P^n =0\ .
\end{equation}
By \eqref{eq:3.112} and \eqref{eq:3.119} we then get, for $g\leq
\tilde{g}_\delta \leq g_\delta^{(1)}$,
\begin{equation}\nonumber
 f_n(H_n - E_n) [H, iA^n] f_n(H_n - E_n) =0\ .
\end{equation}

We then have
\begin{proposition}\label{proposition:3.9}
 Suppose that the kernels $G^{(\alpha)}_{\ell,\epsilon,\epsilon'}$
 satisfy Hypothesis~\ref{hypothesis:2.1} and \ref{hypothesis:3.1}.
 Then there exists $\tilde{C}_\delta>0$ and
 $\tilde{g}_\delta^{(1)}>0$ such that $\tilde g_\delta^{(1)} \leq
 \tilde{g}_\delta$ and
\begin{equation}\nonumber
 f_n(H_n - E_n) [H, i A_n] f_n(H_n - E_n) \geq
 \tilde C_\delta \frac{\gamma^2}{N^2} \sigma_n
 f_n(H_n -E_n)^2
\end{equation}
for $n\geq 1$ and $g\leq \tilde g_\delta^{(1)}$.
\end{proposition}
Let $E_\Delta(H-E)$ be the spectral projection for the operator
$H-E$ associated with the interval $\Delta$, and let
\begin{equation}\label{eq:3.27-bis}
 \Delta_n = [ (\gamma -\epsilon_\gamma)^2\sigma_n, \,
 (\gamma+\epsilon_\gamma)\sigma_n],\ n\geq 1\ .
\end{equation}
Note that
\begin{equation}\label{eq:3.27-ter}
 [\sigma_{n+2},\sigma_{n+1}] \subset \left(
 (\gamma-\epsilon_\gamma)^2\sigma_n,\,
 (\gamma+\epsilon_\gamma)\sigma_n\right),\ n\geq 1\ .
\end{equation}
\begin{theorem}\label{thm:3.10}
Suppose that the kernels $G_{\ell,\epsilon,\epsilon'}^{(\alpha)}$
satisfy Hypothesis~\ref{hypothesis:2.1} and \ref{hypothesis:3.1}.
Then there exists $C_\delta>0$ and $\tilde g_\delta^{(2)}>0$ such
that $\tilde g_\delta^{(2)}\leq\tilde g_\delta^{(1)}$ and
\begin{equation}\nonumber
 E_{\Delta_n} (H-E) [H,\, iA] E_{\Delta_n} (H-E)
 \geq C_\delta \frac{\gamma^2}{N^2} \sigma_n E_{\Delta_n} (H-E)\ ,
\end{equation}
for $n\geq 1$ and $g\leq \tilde g_\delta^{(2)}$.
\end{theorem}

\section{Existence of a ground state and location of the absolutely continuous spectrum}
\setcounter{equation}{0}

We now prove Theorem~\ref{thm:3.3}. The scheme of the proof is
quite well known (see \cite{Bachetal1999S}, \cite{Hiroshima2001}).
It follows from Proposition~\ref{proposition:3.5} that $H^n$ has
an unique ground state, denoted by $\phi^n$, in $\gF^n$,
\begin{equation}\nonumber
 H^n \phi^n = E^n \phi^n,\quad \phi^n\in\cD(H^n),\quad
 \|\phi^n\|=1,\quad n\geq 1\ .
\end{equation}
Therefore $H_n$ has an unique normalized ground state in $\gF$,
given by $\tilde\phi_n = \phi^n \otimes \Omega_n$, where
$\Omega_n$ is the vacuum state in $\gF_n$,
\begin{equation}\nonumber
 H_n \tilde\phi_n = E^n \tilde\phi_n,\quad \tilde\phi_n\in\cD(H_n),\quad
 \|\tilde\phi_n\|=1,\quad n\geq 1\ .
\end{equation}
Since $\| \tilde\phi_n \| =1$, there exists a subsequence
$(n_k)_{k\geq 1}$, converging to $\infty$ such that
$(\tilde\phi_{n_k})_{k\geq 1}$ converges weakly to a state
$\tilde\phi\in\gF$. We have to prove that $\tilde\phi\neq 0$. By
adapting the proof of Theorem~4.1 in \cite{Amouretal2007} (see
also \cite{Barbarouxetal2004}), the key point is to estimate $\|
c_{\ell,\epsilon}(\xi_2) \tilde\Phi_n\|_{\gF}$ in order to show
that
\begin{equation}\label{eq:4.0}
 \sum_{\ell=1}^3 \sum_\epsilon \int \|c_{\ell,\epsilon}(\xi_2)
\tilde\phi_n\|^2 \d\xi_2 = \mathcal{O}(g^2)\ ,
\end{equation}
uniformly with respect to $n$.

The estimate \eqref{eq:4.0} is a consequence of the so-called
``pull-through'' formula as it follows.

Let $H_{I_,n}$ denote the interaction $H_I$ associated with the
kernels $\1_{\{ |p_2|\geq \sigma_n\} }(p_2)
G^{(\alpha)}_{\ell,\epsilon,\epsilon'}$. We thus have
\begin{equation}\nonumber
\begin{split}
 & H_0 c_{\ell,\epsilon}(\xi_2) \tilde\phi_n =
 c_{\ell,\epsilon}(\xi_2) H_0\tilde\phi_n - w_\ell^{(2)}(\xi_2)
 c_{\ell,\epsilon}(\xi_2)\tilde\phi_n \\
 & g H_{I,n} c_{\ell,\epsilon}(\xi_2) \tilde\phi_n =
 c_{\ell,\epsilon}(\xi_2) g H_{I,n} \tilde\phi_n
 + g V_{\ell,\epsilon,\epsilon'}(\xi_2)\tilde\phi_n\ ,
\end{split}
\end{equation}
with
\begin{equation}\nonumber
\begin{split}
 V_{\ell,\epsilon,\epsilon'} (\xi_2)
 = & g \int G^{(1)}_{\ell,\epsilon'\epsilon}(\xi_2,\xi_2,\xi_3)
 b_{\ell,\epsilon'}^*(\xi_1) a_\epsilon(\xi_3) \d\xi_1\, \d\xi_3
 \\
 & + g \int  G^{(2)}_{\ell,\epsilon'\epsilon}(\xi_2,\xi_2,\xi_3)
 b_{\ell,\epsilon'}^* (\xi_1) a_\epsilon^*(\xi_3) \d\xi_1\,
 \d\xi_3\ .
\end{split}
\end{equation}
This yields
\begin{equation}\label{eq:4.8}
 \left(H_n - E_n + w_\ell^{(2)}(\xi_2)\right) c_{\ell,\epsilon}(\xi_2)
 \tilde\phi_n = V_{\ell,\epsilon,\epsilon'} (\xi_2)\tilde\phi_n\ .
\end{equation}
By adapting the proof of Propositions~\ref{proposition:2.4} and
\ref{proposition:2.5} we easily get
\begin{equation}\label{eq:4.9}
\begin{split}
 \| V_{\ell,\epsilon,\epsilon'} \psi\|_{\gF}
 & \leq \frac{g}{m_W{}^\frac12} \left( \sum_{\alpha=1,2}
 \| G^{(\alpha)}_{\ell,\epsilon,\epsilon'}
 (.,\xi_2,.)\|_{L^2(\Sigma_1\times\Sigma_2)}\right) \|
 H_0^\frac12\psi\| \\
 & + g\, \| G^{(2)}_{\ell,\epsilon,\epsilon'}
 (.,\xi_2,.)\|_{L^2(\Sigma_1\times\Sigma_2)} \|\psi\|\ ,
\end{split}
\end{equation}
where $\psi\in\cD(H_0)$.

Let us estimate $\| H_0\tilde\phi_n \|$. By \eqref{eq:2.79},
\eqref{eq:2.80}, \eqref{eq:3.33}, \eqref{eq:3.34} and
\eqref{eq:3.36}  we have
\begin{equation}\nonumber
 g \| H_{I, n} \tilde\phi_n\| \leq g K(G) (C_{\beta\eta} \|H_0
 \tilde\phi_n\| + B_{\beta\eta})
\end{equation}
and
\begin{equation}\nonumber
 \| H_0 \tilde\phi_n\| \leq |E_n| + g \| H_{I,n}\tilde\phi_n\|\ .
\end{equation}
Therefore
\begin{equation}\label{eq:4.12}
 \|H_0 \tilde\phi_n\| \leq \frac{|E_n|}{1 - g_1 K(G)
 C_{\beta\eta}} + \frac{g K(G) B_{\beta\eta}}{1 - g_1 K(G)
 C_{\beta\eta}}\ .
\end{equation}
By \eqref{eq:3.113}, \eqref{eq:A.18} and \eqref{eq:4.12}, there
exists $C>0$ such that
\begin{equation}\label{eq:4.13}
 \| H_0 \tilde\phi_n\| \leq C\ ,
\end{equation}
uniformly in $n$ and $g\leq g_1$.

By \eqref{eq:4.8}, \eqref{eq:4.9} and \eqref{eq:4.13} we get
\begin{equation}\nonumber
 \| c_{\ell,\epsilon} \tilde\phi_n\|
 \leq \frac{g}{|p_2|} \left( C^\frac12 \left(\sum_{\alpha=1}^2 \|
 G^{(\alpha)}_{\ell,\epsilon,\epsilon'}(.,\xi_2,.)\|_{L^2(\Sigma_1\times\Sigma_2)}\right)
 +
 \| G^{(2)}_{\ell,\epsilon,\epsilon'}(.,\xi_2,.)\|_{L^2(\Sigma_1\times\Sigma_2)}
 \right)
\end{equation}
By Hypothesis~\ref{hypothesis:3.1}(i), there exists a constant
$C(G)>0$ depending on the kernels
$G=(G^{(\alpha)}_{\ell,\epsilon,\epsilon'}
)_{\ell=1,2,3;\alpha=1,2;\epsilon\neq\epsilon'=\pm}$ and such that
\begin{equation}\nonumber
 \sum_{\ell=1}^3 \sum_\epsilon \int \| c_{\ell,\epsilon}(\xi_2)
 \tilde\phi_n\|^2 \d \xi_2 \leq C(G)^2 g^2\ .
\end{equation}
The existence of a ground state $\tilde\phi$ for $H$ follows by
choosing $g$ sufficiently small, i.e. $g\leq g_2$, as in
\cite{Amouretal2007} and \cite{Barbarouxetal2004}. By adapting the
method developed in \cite{Hiroshima2005} (see
\cite[Corollary~3.4]{Hiroshima2005}), one proves that the ground
state of $H$ is unique. We omit here the details.

Statements about $\sigma(H)$ are consequences of the existence of
a ground state and follows from the existence of asymptotic Fock
representations for the CAR associated with the
$c_{\ell,\epsilon}^\sharp(\xi_2)$'s. For $f\in L^2(\R^3,\,\C^2)$,
we define on $\cD(H_0)$ the operators
\begin{equation}\nonumber
 c_{\ell,\epsilon}^{\sharp\, t}(f) = \mathrm{e}^{i t H}
 \mathrm{e}^{- i t H_0} c_{\ell,\epsilon}^\sharp(f) \mathrm{e}^{i t
 H_0} \mathrm{e}^{_ i t H}\ .
\end{equation}
By mimicking the proof given in \cite{Hiroshima2001, Takaesu2008}
one proves, under the hypothesis of Theorem~\ref{thm:3.3} and for
$f\in C_0^\infty(\R^3\,\C^2)$, that the strong limits of
$c_{\ell,\epsilon}^{\sharp\, t}(f)$ when $t\rightarrow\pm\infty$
exist for $\psi\in\cD(H_0)$,
\begin{equation}\label{eq:4.17}
 \lim_{t\rightarrow\pm\infty} c_{\ell,\epsilon}^{\sharp\, t}(f)
 \psi := c_{\ell,\epsilon}^{\sharp\, \pm} (f) \psi\ .
\end{equation}
The operators $c_{\ell,\epsilon}^{\sharp\, \pm} (f)$ satisfy the
CAR and we have
\begin{equation}\label{eq:4.18}
 c_{\ell,\epsilon}^{\, \pm}(f)\tilde\phi = 0,\quad f\in
 C_0^\infty(\R^3\,\C^2)\, ,
\end{equation}
where $\tilde\phi$ is the ground state of $H$.

It then follows from \eqref{eq:4.17} and \eqref{eq:4.18} that the
absolutely continuous spectrum of $H$ equals to $[\inf\sigma(H),\,
\infty)$. We omit the details (see \cite{Hiroshima2001,
Takaesu2008}).

\section{Proof of the Mourre Inequality}
\setcounter{equation}{0}

We first prove Proposition~\ref{proposition:3.9}. In view of
Proposition~\ref{proposition:3.8}(a) (iii) and \eqref{eq:3.106},
we have, as sesquilinear forms,
\begin{equation}\label{eq:5.1}
 [H,\, iA_\sigma] = (1-g) \d\Gamma ( (\eta_\sigma)^2 w^{(2)}) +
 g (\d\Gamma ( (\eta_\sigma)^2 w^{(2)}) + g H_I(-i (a_\sigma G))\
 .
\end{equation}
Let $\gF_\ell^{(1)}$ (respectively $\gF_\ell^{(2)}$) be the Fock
space for the massive leptons $\ell$ (respectively the neutrinos
and antineutrinos $\ell$).

We have
\begin{equation}\nonumber
 \gF_\ell \simeq \gF_\ell^{(1)}\otimes\gF_\ell^{(2)}\ .
\end{equation}

Let
\begin{equation}\nonumber
 \gF^{(1)} = \gF_W \otimes (\otimes_{\ell=1}^3
 \, \gF_\ell^{(1)})\quad\mbox{and}\quad
 \gF^{(2)} = \otimes_{\ell=1}^3 \gF_\ell^{(2)}\ .
\end{equation}

We have
\begin{equation}\label{eq:5.4}
 \gF \simeq \gF^{(1)} \otimes \gF^{(2)}\ ,
\end{equation}
$\gF^{(1)}$ is the Fock space for the massive leptons and the
bosons $W^\pm$, and $\gF^{(2)}$ is the Fock space for the
neutrinos and antineutrinos.

We have, as sesquilinear forms and with respect to \eqref{eq:5.4},
\begin{equation}\label{eq:5.5}
\begin{split}
 & \d\Gamma ((\eta_\sigma)^2(p_2) w_\ell^{(2)})
 + H_I (-i (a_\sigma G)) \\
 & = \sum_{\ell=1}^3 \sum_\epsilon \int \eta_\sigma(p_2)^2 |p_2|
 c^*_{\ell,\epsilon}(\xi_2) c_{\ell,\epsilon}(\xi_2) \d\xi_2 \\
 & + \sum_{\ell=1}^3 \sum_{\epsilon\neq\epsilon'} \int |p_2|
 \left(\1_1\otimes\eta_\sigma(p_2) c^*_{\ell,\epsilon}(\xi_2)
 + \sum_{\alpha=1,2}
 \frac{\mathcal{M}^{(\alpha)\,*}_{\ell,\epsilon,\epsilon',\sigma}
 (\xi_2)}{|p_2|}
 \otimes\1_2\right) \\
 & \left( \1_1 \otimes \eta_\sigma(p_2) c_{\ell,\epsilon}(\xi_2)
 + \sum_{\alpha=1,2}
 \frac{\mathcal{M}^{(\alpha)}_{\ell,\epsilon,\epsilon',\sigma}
 (\xi_2)}{|p_2|}
 \otimes\1_2   \right) \d\xi_2 \\
 & - \sum_{\ell=1}^3 \sum_{\epsilon\neq\epsilon'} \int
 \left(\sum_{\alpha=1,2}
 \frac{\mathcal{M}^{(\alpha)\,*}_{\ell,\epsilon,\epsilon',\sigma}
 (\xi_2)}{|p_2|^\frac12}
 \otimes\1_2\right)
 \left(\sum_{\alpha=1,2}
 \frac{\mathcal{M}^{(\alpha)}_{\ell,\epsilon,\epsilon',\sigma}
 (\xi_2)}{|p_2|^\frac12}
 \otimes\1_2\right) \d\xi_2\ ,
\end{split}
\end{equation}
where
\begin{equation}\nonumber
 \mathcal{M}_{\ell,\epsilon,\epsilon',\sigma}^{(\alpha)}(\xi_2)
 = i \int \left( \sum_{\alpha=1,2} (a \, \eta_\sigma(p_2)
 G_{\ell,\epsilon,\epsilon'}^{(\alpha)}(\xi_2,\xi_2,\xi_3)) \right)
 b^*_{\ell,\epsilon'}(\xi_1) a_{\epsilon'}(\xi_3) \d\xi_1 \d\xi_3\
 ,
\end{equation}
and where $\1_j$ is the identity operator in $\gF^{(j)}$.

By mimicking the proofs of Proposition~\ref{proposition:2.4} and
\ref{proposition:2.5}, we get, for every $\psi\in\gD$,
\begin{equation}\nonumber
\begin{split}
  & \sum_{\ell=1}^3 \sum_{\epsilon\neq\epsilon'}
  \left(\psi,\, \int (\sum_{\alpha=1,2}
  \frac{\mathcal{M}_{\ell,\epsilon,\epsilon',\sigma}^{(\alpha)\,*}(\xi_2)}
  {|p_2|^\frac12} \otimes \1_2)
  (\sum_{\alpha=1,2}
  \frac{\mathcal{M}_{\ell,\epsilon,\epsilon',\sigma}^{(\alpha)}(\xi_2)}
  {|p_2|^\frac12} \otimes \1_2) \psi\, \d\xi_2 \right) \\
  & = \sum_{\ell=1}^3 \sum_{\epsilon\neq\epsilon'}
  \left\| \int
  (\sum_{\alpha=1,2}
  \frac{\mathcal{M}_{\ell,\epsilon,\epsilon',\sigma}^{\alpha}(\xi_2)}
  {|p_2|^\frac12} \otimes \1_2) \psi\, \d\xi_2 \right\|^2 \\
  & \leq \left( \int
  \frac{ | \sum_{\alpha=1,2} |(a\, \eta_\sigma(p_2)
  G_{\ell,\epsilon,\epsilon'}^{(\alpha)})(\xi_2,\xi_2,\xi_3)|^2}
  {w^{(3)}(\xi_3) |p_2|} \d\xi_1\d\xi_2\d\xi_3 \right) \,
  \| (H_0^{(3)})^\frac12\psi\|\ .
\end{split}
\end{equation}
Noting that $|(a\, \eta_\sigma)(p_2)|\leq C$ uniformly with
respect to $\sigma$, it follows from
hypothesis~\ref{hypothesis:2.1} and \ref{hypothesis:3.1} that
there exists a constant $C(G)>0$ such that
\begin{equation}\nonumber
 \int \frac{|\sum_{\alpha=1,2} (a\, \eta_\sigma(p_2)
 G^{(\alpha)}_{\ell,\epsilon,\epsilon'})(\xi_1,\xi_2,\xi_3)|^2}
 {w^{(3)}(\xi_3) |p_2|} \d\xi_1 \d\xi_2 \d\xi_3 \leq C(G) \sigma\
 .
\end{equation}
This yields
\begin{equation}\label{eq:5.9}
 - \int (\sum_{\alpha=1,2} \frac{\mathcal{M}^{(\alpha)\,
 *}_{\ell,\epsilon,\epsilon',\sigma}(\xi_2)}{|p_2|^\frac12}
 \otimes \1_2)
 (\sum_{\alpha=1,2} \frac{\mathcal{M}^{(\alpha)}
 _{\ell,\epsilon,\epsilon',\sigma}(\xi_2)}{|p_2|^\frac12}
 \otimes \1_2) \d\xi_2 \geq - C(G) \sigma\ .
\end{equation}

Combining \eqref{eq:5.1}, \eqref{eq:5.5} with \eqref{eq:5.9}, we
obtain
\begin{equation}\label{eq:5.10}
 [H,\, iA_n] \geq (1-g) \d\Gamma((\eta_{\sigma_n})^2 w_\ell^{(2)})
 - g C(G)\sigma_n\ .
\end{equation}

We have
\begin{equation}\label{eq:5.11}
 \d\Gamma( (\eta_{\sigma_n})^2 w_\ell^{(2)}) \geq H_{0\, n}^{(2)}\ .
\end{equation}

By \eqref{eq:3.108}, \eqref{eq:3.112} and \eqref{eq:5.11} we get
\begin{equation}\nonumber
\begin{split}
 f_n(H_n - E_n) \d\Gamma(\eta_{\sigma_n}{}^2 w_\ell^{(2)}) f_n(H_n-E_n)
 & \geq P_n\otimes f_n(H_{0\, n}^{(2)})
 \, H_{0\, n}^{(2)}\, f_n(H_{0\,
 n}^{(2)})\\
 & \geq \frac{\gamma^2}{N^2} \sigma_n f_n(H_n - E_n)^2\ ,
\end{split}
\end{equation}
for $g\leq g_\delta^{(1)}$.

This, together with \eqref{eq:5.10}, yields for $g\leq
g_\delta^{(1)}$
\begin{equation}\nonumber
\begin{split}
  & f_n(H_n - E_n) [H,\, iA_n] f_n(H_n-E_n) \\
  & \geq
  (1-g_\delta^{(1)})\frac{\gamma^2}{N^2} \sigma_n f_n(H_n-E_n)^2
  - g\, C(G)\, \sigma_n f_n(H_n - E_n)^2\ .
\end{split}
\end{equation}
Setting
 $$
 g_\delta^{(2)} = \inf (g_\delta^{(1)},\,
 \frac{1-g_\delta^{(1)}}{2\, C(G)} \frac{\gamma^2}{N^2})\ ,
 $$
we get
\begin{equation}\nonumber
 f_n( H_n - E_n) [H,\, iA_n] f_n(H_n - E_n) \geq
\frac{1 - g_\delta^{(1)}}{2} \frac{\gamma^2}{N^2}\, \sigma_n
f_n(H_n - E_n)^2\ ,
\end{equation}
for $g\leq g_\delta^{(2)}$.

Proposition~\ref{proposition:3.9} is proved by setting
$\tilde{g}_\delta^{(1)} = g_\delta^{(2)}$ and
$\tilde{C}_\delta=\frac{1 - g_\delta^{(1)}}{2}$.

The proof of Theorem~\ref{thm:3.10} is the consequence of the
following two lemmata.

\begin{lemma}\label{lemma:5.1}
Assume that the kernels $G^{(\alpha)}_{\ell,\epsilon,\epsilon'}$
satisfy Hypothesis~\ref{hypothesis:2.1} and
\ref{hypothesis:3.1}(ii). Then there exists a constant $D>0$ such
that
\begin{equation}\nonumber
 |E - E_n| \leq g\, D\, \sigma_n{}^2\ ,
\end{equation}
for $n\geq 1$ and $g\leq g^{(2)}$.
\end{lemma}
%
%
%
\begin{proof}
Let $\phi$ (respectively  $\tilde\phi_n$) be the unique normalized
ground state of $H$ (respectively $H_n$). We have
\begin{equation}\label{eq:5.15}
\begin{split}
 & E - E_n \leq (\tilde \phi_n, (H-H_n)\tilde \phi_n) \\
 & E_n - E \leq (\phi, (H_n - H) \phi )\ ,
\end{split}
\end{equation}
with
\begin{equation}\label{eq:5.16}
 H - H_n = g H_I (\chi_{\sigma_n}(p_2) G)\ .
\end{equation}
Combining \eqref{eq:2.79} and \eqref{eq:2.80} with
\eqref{eq:3.33}-\eqref{eq:3.36} and \eqref{eq:5.16}, we get
\begin{equation}\label{eq:5.17}
 \| (H-H_n)\tilde \phi_n\| \leq g\,
 K(\chi_{\sigma_n}(p_2) G) \, (C_{\beta\eta} \|H_0\tilde\phi_n\|
 + B_{\beta\eta})
\end{equation}
and
\begin{equation}\label{eq:5.18}
 \| (H-H_n) \phi \| \leq g\, K(\chi_{\sigma_n}(p_2) G) \, (C_{\beta\eta}
 \|H_0 \phi \|
 + B_{\beta\eta})
\end{equation}
It follows from Hypothesis~\ref{hypothesis:3.1}(ii),
\eqref{eq:4.13}, \eqref{eq:5.17} and \eqref{eq:5.18} that there
exists a constant $D>0$ such that
\begin{equation}\nonumber
 \max (\| (H-H_n) \tilde\phi_n\|,\, \|(H-H_n)\phi\|\leq g\, D\,
 \sigma_n{}^2\ ,
\end{equation}
for $n\geq 1$ and $g\leq g^{(2)}$.

By \eqref{eq:5.15}, this proves Lemma~\ref{lemma:5.1}.
\end{proof}

\begin{lemma}\label{lemma:5.2}
Suppose that the kernels $G_{\ell,\epsilon,\epsilon'}^{(\alpha)}$
satisfy Hypothesis~\ref{hypothesis:2.1} and
\ref{hypothesis:3.1}(ii). Then there exists a constant $C>0$ such
that
\begin{equation}
 \| f_n(H-E) - f_n(H_n - E_n)\| \leq g\, C\, \sigma_n\ ,
\end{equation}
for $n\geq 1$ and $g\leq g^{(2)}$.
\end{lemma}
%
%
%
\begin{proof} Let $\tilde f(.)$ be an almost analytic extension of
$f(.)$ given by \eqref{eq:3.109} satisfying
\begin{equation}\label{eq:5.20}
 \left| \partial_{\bar z} \tilde f(x+iy)\right| \leq C y^2\ .
\end{equation}
Note that $\tilde f(x+iy)\in C_0^\infty(\R^2)$. We thus have
\begin{equation}\label{eq:5.22}
 f(s) = \int \frac{\d\tilde f(z)}{ z -s },\quad
 \d\tilde f(z) = -\frac{1}{\pi} \frac{\partial \tilde f}{\partial
 \bar z}\, \d x\, \d y\ .
\end{equation}
Using the functional calculus based on this representation of
$f(s)$, we get
\begin{equation}\label{eq:5.23}
 f_n(H-E) - f_n (H_n - E_n)
 = \sigma_n \int \frac{1}{H - E - z\sigma_n}
 (H-H_n+E_n-E)
 \frac{1}{H_n - E_n - z\sigma_n} \d\tilde f(z)\ .
\end{equation}
Combining \eqref{eq:2.79} and \eqref{eq:2.80} with
\eqref{eq:3.33}-\eqref{eq:3.36} and
Hypothesis~\ref{hypothesis:3.1}(ii), we get, for every
$\psi\in\cD(H_0)$ and for $g\leq g^{(2)}$,
\begin{equation}\label{eq:5.24}
 g\| H_I (\chi_{\sigma_n} G) \psi \|
 \leq 2\, g\, C\, \sigma_n{}^2 K(G)\,
 (C_{\beta\eta} \|(H_0 +1) \psi\| +
 (C_{\beta\eta} + B_{\beta\eta}) \| \psi\|)\ .
\end{equation}
This yields
\begin{equation}\label{eq:5.25}
 g \|H_I(\chi_{\sigma_n}(p_2)G) (H_0+1)^{-1} \| \leq g \, C_1\,
 \sigma_n{}^2\ ,
\end{equation}
for some constant $C_1>0$ and for $g\leq g^{(2)}$.

By mimicking the proof of \eqref{eq:B.28} we show that there
exists a constant $C_2>0$ such that
\begin{equation}\label{eq:5.26}
 \| (H_0+1) (H_n - E_n - z\sigma_n)^{-1}\| \leq C_2 (1 +
 \frac{1}{|\Im z| \sigma_n})\ ,
\end{equation}
for $g\leq g^{(1)}$.

Combining Lemma~\ref{lemma:5.1} and \eqref{eq:5.23} with
\eqref{eq:5.24}-\eqref{eq:5.26} we obtain
\begin{equation}\nonumber
  \| f_n(H-E) - f_n(H_n -E_n)\| \leq g\, C\, \sigma_n
  \int \frac{|\frac{\partial \tilde f}{\partial \bar z} (x+i
  y)|}{y^2} \d x \d y\ ,
\end{equation}
for some constant $C>0$ and for $g\leq g^{(2)}$.

Using \eqref{eq:5.20} and $\tilde f(x+iy)\in C_0^\infty(\R^2)$ one
concludes the proof of Lemma~\ref{lemma:5.2}.
\end{proof}
%
%
%

We now prove Theorem~\ref{thm:3.10}.
\begin{proof}
It follows from Proposition~\ref{proposition:3.9} that
\begin{equation}\nonumber
\begin{split}
 & f_n(H_n - E_n) [H,\, iA] f_n(H_n - E_n)\\
 & = f_n(H_n - E_n) [H,\, iA_n] f_n(H_n - E_n)
 \geq \tilde C_\delta \frac{\gamma^2}{N^2} \sigma_n\,
 f_n(H_n - E_n)^2\ ,
\end{split}
\end{equation}
for $n\geq 1$ and $g\leq \tilde g_\delta^{(1)}$.

This yields
\begin{equation}\nonumber
\begin{split}
 & f_n(H-E) [H, iA_n] f_n(H-E)
 \geq \tilde C_\delta \frac{\gamma^2}{N^2} \sigma_n \, f_n(H-E)^2
 \\
 & - f_n(H-E) [H,\, iA] (f_n(H_n-E_n) - f_n(H-E)) \\
 & - (f_n(H_n - E_n) - f_n(H-E)) [H,\, iA] f_n(H_n - E_n) \\
 & + \tilde C_\delta \frac{\gamma^2}{N^2} \sigma_n (f_n(H_n - E_n) -
 f_n(H-E))^2 \\
 & + \tilde C_\delta \frac{\gamma^2}{N^2} \sigma_n
 f_n(H-E) (f_n(H_n -E_n) - f_n(H-E))\\
 & + \tilde C_\delta \frac{\gamma^2}{N^2} \sigma_n (f_n(H_n - E_n)
 -f_n(H-E)) f_n (H-E)\ .
\end{split}
\end{equation}
Combining Proposition~\ref{proposition:3.8} (i) and
\eqref{eq:5.22} with \eqref{eq:5.25} and \eqref{eq:5.26} we show
that $[H,\, iA] f_n(H_n -E_n)$ and $f_n(H-E) [H,\, iA]$ are
bounded operators uniformly with respect to $n$. This, together
with Lemma~\ref{lemma:5.2}, yields
\begin{equation}\label{eq:5.29}
 f_n (H-E) [H,\, iA] f_n (H-E)
 \geq \tilde C_\delta \frac{\gamma^2}{N^2} \sigma_n f_n(H-E)^2 -
 \tilde C \,g\, \sigma_n\ ,
\end{equation}
for some constant $\tilde C>0$ and for $g\leq\inf(g^{(2)},\,
\tilde g_\delta^{(1)})$.

Multiplying both sides of \eqref{eq:5.29} with $E_{\Delta_n}
(H-E)$ we then get
\begin{equation}\nonumber
 E_{\Delta_n} (H-E) [H,\, iA] E_{\Delta_n}(H-E) \geq \tilde
 C_\delta \frac{\gamma^2}{N^2} \sigma_n E_{\Delta_n} (H-E) -
 \tilde C\, g\, \sigma_n E_{\Delta_n}(H-E)\ .
\end{equation}
Setting
\begin{equation}\nonumber
 \tilde g_\delta^{(2)} < \inf \left( \frac{\tilde C_\delta}{\tilde
 C} \frac{\gamma^2}{N^2},\, g^{(2)},\, \tilde
 g_\delta^{(1)}\right)\ ,
\end{equation}
Theorem~\ref{thm:3.10} is proved with $C_\delta= \tilde C_\delta -
\tilde C \frac{N^2}{\gamma^2} \tilde g_\delta^{(2)}>0$.
\end{proof}

\section{Proof of Theorem~\ref{thm:3.7}}
\setcounter{equation}{0}

We set
\begin{equation}\nonumber
\begin{split}
 & A_t = \frac{\mathrm{e}^{i t A} - 1}{t}\ , \\
 & \adat \cdot = [A_t,\, .\, ] \ , \\
 & A_t^\sigma = \frac{\mathrm{e}^{i t A^\sigma -1}}{t} \ , \\
 & A_{\sigma\, t} = \frac{\mathrm{e}^{i t A_\sigma} -1}{t}\ .
\end{split}
\end{equation}
The fact that $H$ is of class $C^1(A)$, $C^1(A^\sigma)$ and
$C^1(A_\sigma)$ in $(-\infty,\, m_1 -\frac{\delta}{2})$ is the
consequence of the following proposition
\begin{proposition}\label{proposition:6.1}
Suppose that the kernels $G^{(\alpha)}_{\ell,\epsilon,\epsilon'}$
satisfy Hypothesis~\ref{hypothesis:2.1} and
\ref{hypothesis:3.1}(iii.a). For every $\vp\in
C_0^\infty((-\infty, m_1-\frac{\delta}{2}))$ and $g\leq g_1$, we
then have
\begin{equation}\nonumber
\begin{split}
 & \sup_{0<|t|\leq 1} \| [\vp(H),\, A_t]\|<\infty\, , \\
 & \sup_{0<|t|\leq 1} \| [\vp(H), A_t^\sigma]\| <\infty\, , \\
 & \sup_{0 < |t| \leq 1} \| [\vp(H),\, A_{\sigma\, t}]\|
 <\infty\, , \\
 & \sup_{0<|t|\leq 1} \| [\vp(H^\sigma),\, A_t^\sigma]\|<\infty\ .
\end{split}
\end{equation}
\end{proposition}
%
%
%
\begin{proof}
We use the representation
\begin{equation}\nonumber
 \vp(H) = \int \d\phi(z) (z - H)^{-1}\ ,
\end{equation}
where $\phi(z)$ is an almost analytic extension of $\vp$ with
\begin{equation}\nonumber
 |\partial_{\bar z} \phi(x+iy)| \leq C |y|^2\quad\mbox{and}
 \quad\d\phi(z) = -\frac{1}{\pi}\frac{\partial}{\partial \bar z}
 \phi(z)\d x\d y  \ .
\end{equation}
Note that $\phi(x+iy)\in C_0^\infty(\R^2)$.

We get
\begin{equation}\nonumber
 \adat\vp(H) = \int\d\phi(z) (z-H)^{-1}
 [A_t,\, H] (z-H)^{-1}\ .
\end{equation}
This yields
\begin{equation}\nonumber
\begin{split}
 & \| \adat\vp(H)\| \\
 & \leq \sup_{0< |t|\leq 1}
 \| [A_t,\, H](i - H)^{-1}\| \,
 \int |\d\phi(z)|\, \|(z-H)^{-1}\|\, \|(i-H)(z-H)^{-1}\|\ .
\end{split}
\end{equation}
It is easy to prove that
\begin{equation}\label{eq:6.13}
 \int |\d\phi(z)|\, \|(z - H)^{-1}\| \,
 \|(i-H)(z-H)^{-1}\| \leq C \int\frac{|\d\phi(z)|}{|\Im z|^2}
 <\infty\ .
\end{equation}

By Proposition~\ref{proposition:3.8}(b)$(i)$ and \eqref{eq:6.13}
we finally get, for $g\leq g_1$
\begin{equation}\nonumber
 \sup_{0 < |t| \leq 1} \|  \adat\, \vp(H)\| <\infty\ .
\end{equation}
In a similar way we obtain, for $g\leq g_1$
\begin{equation}\nonumber
\begin{split}
 \sup_{0 <   |t| \leq 1} \| [A_t^\sigma,\, \vp(H)]\| <\infty \, ,
 \\
 \sup_{0<|t| \leq 1} \| [A_{\sigma\, t},\, \vp(H)\| <\infty\, ,\\
 \sup_{0<|t|\leq 1} \| [A_t^\sigma,\, \vp(H^\sigma)]\| <\infty\ .
\end{split}
\end{equation}
\end{proof}

The proof of Theorem~\ref{thm:3.7} is the consequence of the
following proposition
\begin{proposition}\label{proposition:6.2}
Suppose that the kernels $G^{(\alpha)}_{\ell,\epsilon,\epsilon'}$
satisfy Hypothesis~\ref{hypothesis:2.1} and
\ref{hypothesis:3.1}~(i)-(iii). We then have, for $g\leq g_1$,
\begin{equation}\nonumber
\begin{split}
 & \sup_{0<|t| \leq 1} \| [A_t,\, [A_t,\, H]] (H+i)^{-1} \| <\infty
 \, , \\
 & \sup_{0<|t|\leq 1} \| [A_t^\sigma, [A_t^\sigma,\, H](H+i)^{-1}\|
 <\infty \, ,\\
 & \sup_{0<|t|\leq 1} \| [A_{\sigma\,t}, [A_{\sigma\,t},\, H](H+i)^{-1}\|
 <\infty \, , \\
 & \sup_{0<|t|\leq 1} \| [A_t^\sigma, [A_t^\sigma,\, H^\sigma](H^\sigma+i)^{-1}\|
 <\infty \, ,
\end{split}
\end{equation}
\end{proposition}
\begin{proof}
We have, for every $\psi\in\cD(H)$,
\begin{equation}\label{eq:6.17}
 [A_t,[A_t,H]]\psi =
 \frac{1}{t^2} \mathrm{e}^{2 i t A}
 (\mathrm{e}^{- 2 i t A} H \mathrm{e}^{ 2 i t A }
 - 2\mathrm{e}^{- i t A}H \mathrm{e}^{ i t A} + H)\psi\ .
\end{equation}
By \eqref{eq:3.76} we get
\begin{equation}\label{eq:6.18}
 [A_t,[A_t,H_0]]\psi =
 \frac{1}{t^2} \mathrm{e}^{2 i t A}
 (\d\Gamma( w^{(2)}\circ \phi_{2t} - 2 w^{(2)}\circ \phi_t
 + w^{(2)}))\psi\ ,
\end{equation}
where, for $\ell=1,2,3$,
\begin{equation}\label{eq:6.19}
 (w_\ell^{(2)} \circ \phi_{2 t})(p_2) -
 2 (w_\ell^{(2)}\circ\phi_t)(p_2) + w_\ell^{(2)}(p_2)
 = | \phi_{2 t}(p_2)| - 2 |\phi_t(p_2)| + |p_2|\ .
\end{equation}
We further note that
\begin{equation}\label{eq:6.20}
 \frac{1}{t^2} \big|\, |\phi_{2t}(p_2)| - 2 |\phi_t(p_2)| + |p_2|\,
 \big|
 \leq \sup_{|s|\leq 2|t|} \left| \frac{\partial^2}{\partial s^2}
 |\phi_s(p_2)|\, \right|\ ,
\end{equation}
and
\begin{equation}\label{eq:6.21}
 \frac{\partial^2}{\partial s^2} |\phi_s(p_2)|
 = |\phi_s(p_2)| \leq \mathrm{e}^{\Gamma |s|} |p_2|\ .
\end{equation}
Combining \eqref{eq:6.18} with \eqref{eq:6.19}-\eqref{eq:6.21} we
get
\begin{equation}\nonumber
 \| [ A_t,\, [A_t,\, H_0]](H_0+1)^{-1}\| \leq \mathrm{e}^{2\Gamma
 |t|}\ ,
\end{equation}
and
\begin{equation}\nonumber
 \sup_{0 < |t| \leq 1} \| [A_t,\, [A_t,\, H_0]](H_0+1)^{-1}\|
 \leq \mathrm{e}^{2\Gamma}\ .
\end{equation}
In a similar way we obtain
\begin{equation}\nonumber
 \sup_{0 <  | t |\leq 1}
 \| [A_t^\sigma,\, [A_t^\sigma, H_0]](H_0+1)^{-1}\| \leq
 C \mathrm{e}^{2\Gamma}\ ,
\end{equation}
\begin{equation}\nonumber
 \sup_{0 <  | t |\leq 1}
 \| [A_{\sigma\, t},\, [A_{\sigma\, t}, H_0]](H_0+1)^{-1}\| \leq
 C \mathrm{e}^{2\Gamma}\ .
\end{equation}
Here $C$ is a positive constant.

Let us now prove that
\begin{equation}\nonumber
 \sup_{0<|t|\leq 1} \| [A_t,\, [A_t,\, H_I(G)]](H+i)^{-1}\|
 <\infty
\end{equation}
By \eqref{eq:3.81} and \eqref{eq:6.17} we get, for every $\psi\in
\cD(H)$,
\begin{equation}\label{eq:6.26}
\begin{split}
  & [A_t,\, [A_t,\, H_I(G)]]\psi \\
  & =
   \sum_{\alpha=1,2}\sum_{\ell=1,2,3}
  \sum_{\epsilon\neq\epsilon'} \frac{\mathrm{e}^{2 i t A}}{t^2}
  \Big( \mathrm{e}^{- 2 i t A} H_I
  (G^{(\alpha)}_{\ell,\epsilon,\epsilon'}) \mathrm{e}^{2 i t A}
  - 2 \mathrm{e}^{- i t A} H_I
  (G^{(\alpha)}_{\ell,\epsilon,\epsilon'})\mathrm{e}^{i t A} \\
  & + H_I(G^{(\alpha)}_{\ell,\epsilon,\epsilon'})\Big)\psi \\
  & =
   \sum_{\alpha=1,2}\sum_{\ell=1,2,3}
  \sum_{\epsilon\neq\epsilon'} \frac{\mathrm{e}^{2 i t A}}{t^2}
  \Big( H_I(G^{(\alpha)}_{\ell,\epsilon,\epsilon';2t})
  - 2 H_I(G^{(\alpha)}_{\ell,\epsilon,\epsilon';t})
  + H_I(G^{(\alpha)}_{\ell,\epsilon,\epsilon';0})\Big)\psi\ ,
\end{split}
\end{equation}
where
\begin{equation}\nonumber
\begin{split}
 G^{(\alpha)}_{\ell,\epsilon,\epsilon';t}(\xi_1,\xi_2,\xi_3)
 & = (D\phi_t(p_2))^\frac12 G^{(\alpha)}_{\ell,\epsilon,\epsilon'}
 (\xi_1;\, \phi_t(p_2),s_2;\, \xi_3) \\
 & = (e^{-i t a}
 G^{(\alpha)}_{\ell,\epsilon,\epsilon'})(\xi_1,\xi_2,\xi_3)\ .
\end{split}
\end{equation}
Combining \eqref{eq:2.79} and \eqref{eq:2.80} with
\eqref{eq:3.33}-\eqref{eq:3.36} and \eqref{eq:6.26} we get
\begin{equation}\label{eq:6.28}
 \| [ A_t,\, [A_t,\, H_I(G)]]\psi\| \leq
 g\, K(G_t)
 (C_{\beta\eta} \|(H_0+I)\psi\|
 + (C_{\beta\eta} + B_{\beta\eta})\|\psi\|)\ .
\end{equation}
Here $K(G_t)>0$ and
\begin{equation}\label{eq:6.29}
 K(G_t)^2 =
 \sum_{\alpha=1,2}\sum_{\ell=1,2,3}\sum_{\epsilon\neq\epsilon'}
 \frac{1}{t^2}
 \| G^{(\alpha)}_{\ell,\epsilon,\epsilon';2t} - 2
 G^{(\alpha)}_{\ell,\epsilon,\epsilon';t}
 +
 G^{(\alpha)}_{\ell,\epsilon,\epsilon'}\|^2_{L^2(\Sigma_1\times\Sigma_1
 \times\Sigma_2)}\ .
\end{equation}
We further note that, for $0 < |t| \leq 1$,
\begin{equation}\label{eq:6.30}
 K(G_t) \leq
 \sup_{0<|s|\leq 2}
 \Big( \sum_{\alpha=1,2}\sum_{\ell=1,2,3}\sum_{\epsilon\neq\epsilon'}
 \left\| \frac{\partial^2}{\partial s^2}
 G^{(\alpha)}_{\ell,\epsilon,\epsilon';s}\right\|^2
 _{L^2(\Sigma_1\times\Sigma_1\times\Sigma_2)}\Big)^\frac12\ .
\end{equation}
We get
\begin{equation}\label{eq:6.31}
\begin{split}
 & \left(\frac{\partial}{\partial t}
 G^{(\alpha)}_{\ell,\epsilon,\epsilon';t}\right)\\
 & = \frac32 (\mathrm{e}^{- i t a}
 G^{(\alpha)}_{\ell,\epsilon,\epsilon'})
 + (\mathrm{e}^{- i t a} (p_2\cdot\nabla_{p_2}
 G^{(\alpha)}_{\ell,\epsilon,\epsilon'}))\, ,
 \end{split}
\end{equation}
and
\begin{equation}\label{eq:6.32}
\begin{split}
 & \left(\frac{\partial^2}{\partial t^2}
 G^{(\alpha)}_{\ell,\epsilon,\epsilon';t}\right)\\
 & = \frac94 (\mathrm{e}^{- i t a}
 G^{(\alpha)}_{\ell,\epsilon,\epsilon'})
 +\frac72 (\mathrm{e}^{-i t a} (p_2\cdot\nabla_{p_2}
 G^{(\alpha)}_{\ell,\epsilon,\epsilon'}))
 +\! \sum_{i,j=1,2,3}\mathrm{e}^{- i t a}\big(p_{2,i} p_{2,j}
 \partial^2_{p_{2,i}p_{2,j}}
 G^{(\alpha)}_{\ell,\epsilon,\epsilon'}\big) .
\end{split}
\end{equation}
Recall that $\mathrm{e}^{- i t a}$ is an one parameter group of
unitary operators in $L^2(\Sigma_1\times\Sigma_1\times\Sigma_2)$.

Combining Hypothesis~\ref{hypothesis:3.1}(iii.a) and (iii.b), with
\eqref{eq:6.28}-\eqref{eq:6.32} we finally get
\begin{equation}\nonumber
 \sup_{0<|t|\leq 1} \| [\,A_t,\, [A_t,\,
 H_I(G)]\,](H_0+1)^{-1}\|<\infty\ .
\end{equation}
In view of $\cD(H) = \cD(H_0)$ the operators $H_0(H+i)^{-1}$ and
$H(H_0 - 1)^{-1}$ are bounded and we obtain
\begin{equation}\nonumber
 \sup_{0 < |t|\leq 1} \| [\,A_t,\, [A_t,\, H_0]\,] (H+i)^{-1}\|<\infty\
 ,
\end{equation}
\begin{equation}\label{eq:6.35}
 \sup_{0 < |t|\leq 1} \| [\,A_t,\, [A_t,\, H_I(G)]\,] (H+i)^{-1}\|<\infty\
 .
\end{equation}
This yields
\begin{equation}\label{eq:6.36}
 \sup_{0<|t|\leq 1} \| [\, A_t,\, [A_t,\,
 H]\,](H+i)^{-1}\|<\infty\ ,
\end{equation}
for $g\leq g_1$.

Let $V(p_2)$ denote any of the two $C^\infty$-vector fields
$v^\sigma(p_2)$ and $v_\sigma(p_2)$ and let $\tilde a$ denote the
corresponding $a^\sigma$ and $a_\sigma$ operators. We get
\begin{equation}\nonumber
\begin{split}
 & \left( \frac{\partial^2}{\partial t^2} (\mathrm{e}^{- i \tilde a
 t} G^{(\alpha)}_{\ell,\epsilon,\epsilon'})\right)
 (\xi_1,\xi_2,\xi_3) \\
 & =
 \frac14 \left( \mathrm{e}^{- i \tilde a t}(
 (\mathrm{div} V(p_2))^2
 G^{(\alpha)}_{\ell,\epsilon,\epsilon'}) \right)
 (\xi_1,\xi_2,\xi_3) \\
 & +
 \frac12 \left( \mathrm{e}^{- i \tilde a t}(
 (\mathrm{div} V(p_2)) V(p_2)\cdot\nabla_{p_2}
 G^{(\alpha)}_{\ell,\epsilon,\epsilon'}) \right)
 (\xi_1,\xi_2,\xi_3) \\
 & +
 \frac12 \left(
 \mathrm{e}^{- i \tilde a t}(\sum_{i,j=1}^3
 (V_i(p_2) (\partial^2_{p_{2,i} p_{2,j}} V_j(p_2)))
 G^{(\alpha)}_{\ell,\epsilon,\epsilon'}) \right)
 (\xi_1,\xi_2,\xi_3)\\
 & +
 \frac12 \left(
 \mathrm{e}^{- i \tilde a t}(\sum_{i,j=1}^3
 V_i(p_2) \frac{\partial V_j}{\partial p_{2,i}}(p_2)
 \frac{\partial}{\partial p_{2,j}}
 G^{(\alpha)}_{\ell,\epsilon,\epsilon'})\right)
 (\xi_1,\xi_2,\xi_3)\\
 & + \frac12 \left(\mathrm{e}^{- i \tilde a t}
 (\sum_{i,j=1}^3 V_i(p_2) V_j(p_2) \frac{\partial^2}
 {\partial p_{2,i} \partial p_{2,j}}
 G^{(\alpha)}_{\ell,\epsilon,\epsilon'})\right)(\xi_1,\xi_2,\xi_3)\
 .
\end{split}
\end{equation}
Combining the properties of the $C^\infty$ fields $v^\sigma(p_2)$
and $v_\sigma(p_2)$ together with Hypothesis~\ref{hypothesis:2.1}
and \ref{hypothesis:3.1} we get, from \eqref{eq:6.35} and by
mimicking the proof of \eqref{eq:6.36},
\begin{equation}\label{eq:6.38}
 \sup_{0 <|t|\leq 1} \| \, [\, A_t^\sigma,\,
 [A_t^\sigma, \, H]\,] (H+i)^{-1} \| <\infty\ ,
\end{equation}
\begin{equation}\nonumber
 \sup_{0 <|t|\leq 1} \| \, [\, A_{\sigma\,t},\,
 [A_{\sigma\, t},\, H]\,] (H+i)^{-1} \| <\infty\ ,
\end{equation}
for $g\leq g_1$.

Similarly, by mimicking the proof of \eqref{eq:6.38}, we easily
get, for $g\leq g_1$,
\begin{equation}\nonumber
 \sup_{0 <|t|\leq 1} \| \, [\, A_{t}^\sigma,\,
 [A_{t}^\sigma,\, H^\sigma]\,] (H^\sigma+i)^{-1} \| <\infty\ .
\end{equation}
This concludes the proof of Proposition~\ref{proposition:6.2}
\end{proof}


We now prove Theorem~\ref{thm:3.7}.

\smallskip

\noindent\textit{Proof of Theorem~\ref{thm:3.7}}. In view of
\cite[Lemma~6.2.3]{Amreinetal1996} (see also
\cite[Proposition~28]{Frohlichetal2008}), the proof of
Theorem~\ref{thm:3.7} will follow from
Proposition~\ref{proposition:6.1} and the following estimates
\begin{equation}\label{eq:6.41}
 \sup_{0< |t| \leq 1} \| \, [\, A_t,\, [A_t,\, \vp(H)]\, ]\,
 \|<\infty\ ,
\end{equation}
\begin{equation}\label{eq:6.42}
 \sup_{0< |t| \leq 1} \| \, [\, A_t^\sigma,\, [A_t^\sigma,\, \vp(H)]\, ]\,
 \|<\infty\ ,
\end{equation}
\begin{equation}\label{eq:6.43}
 \sup_{0< |t| \leq 1} \| \, [\, A_{\sigma\,t},\, [A_{\sigma\,t},\, \vp(H)]\, ]\,
 \|<\infty\ ,
\end{equation}
\begin{equation}\label{eq:6.44}
 \sup_{0< |t| \leq 1} \| \, [\, A_t^\sigma,\, [A_t^\sigma,\, \vp(H^\sigma)]\, ]\,
 \|<\infty\ ,
\end{equation}
for every $\vp\in C_0^\infty((-\infty, m_1-\delta/2))$ and for
$g\leq g_1$.

Let us prove \eqref{eq:6.41}. The inequalities
\eqref{eq:6.42}-\eqref{eq:6.44} can be proved similarly.

To this end, let $\phi$ be an almost analytic extension of $\vp$
satisfying
\begin{equation}\nonumber
 | \partial_{\bar z} \phi(x+iy)| \leq C |y|^3\ ,
\end{equation}
and
\begin{equation}\nonumber
 \vp(H) = \int (z-H)^{-1} \d\phi(z) \ ,\quad
 \d\phi(z) = -\frac{1}{\pi} \frac{\partial}{\partial \bar z}
 \phi(z) \d x \d y\ .
\end{equation}
It follows that
\begin{equation}\nonumber
\begin{split}
  & [A_t\, [A_t,\, \vp(H)]\, ]
  = \int  \Big( (z-H)^{-1}
 [A_t\, [A_t,\, H]\, ] (z-H)^{-1}\\
 & + 2 (z-H)^{-1} [A_t,\, H] (z-H)^{-1} [A_t,\, H]
 (z-H)^{-1}\Big) \d\phi(z)
\end{split}
\end{equation}

We note that
\begin{equation}\label{eq:6.56}
 \| (H+i) (H-z)^{-1}\| \leq \frac{C}{| \Im z |},\quad
 \mbox{for }z\in\mathrm{supp}\phi\ .
\end{equation}

We also have
\begin{equation}\label{eq:6.55}
\begin{split}
 & \sup_{0 < |t| \leq 1} \| \int  (z-H)^{-1}
 [A_t\, [A_t,\, H]\, ] (z-H)^{-1} \d\phi(z) \|
 \\
 & \leq \sup_{ 0<|t| \leq 1}
 \int
 \| [A_t\, [A_t,\, H]\, ](H+i)^{-1}\| \,
 \| (H+i)  (z-H)^{-1}\| \frac{ |\d\phi(z)|}{|\Im z|} \\
 & \leq C \sup_{0<|t|\leq 1} \|\, \left[A_t,\, [A_t,H]\,\right]
 (H+i)^{-1}\,\|
 \int\frac{|\d\phi(z)|}{|\Im z|^2}
 \ .
\end{split}
\end{equation}

Therefore, combining Proposition~\ref{proposition:3.8} (b)(i) and
\eqref{eq:6.56} we obtain
\begin{equation}\label{eq:6.57}
\begin{split}
 & \sup_{0 < |t| \leq 1}
 \| \int \d\phi(z) (H-z)^{-1}  [A_t,\, H]
 (H-z)^{-1} [A_t,\, H](H-z)^{-1} \| \\
 & = \sup_{ 0<|t|\leq 1}
 \| \int
 (H-z)^{-1}  [A_t,\, H]
 (H+i)^{-1} (H+i)(H-z)^{-1} \\
 & \quad\quad  [A_t,\, H]
 (H+i)^{-1} (H+i)(H-z)^{-1} \| \d\phi(z) \\
 & \leq C \left( \int \frac{|\d\phi(z)|}{|y| ^3}\right)
 \sup_{ 0 <|t|\leq 1} \| \, [A_t,\, H](H+i)^{-1}\|^2
 <\infty\ .
\end{split}
\end{equation}
Inequality \eqref{eq:6.57} together with \eqref{eq:6.55} yields
\eqref{eq:6.41}, and $H$ is locally of class $C^2(A)$ on
$(-\infty,\, m_1-\delta/2)$ for $g\leq g_1$.

In a similar way it follows from
Proposition~\ref{proposition:3.8}(b),
Proposition~\ref{proposition:6.1} and
Proposition~\ref{proposition:6.2} that $H$ is locally of class
$C^2(A^\sigma)$ and $C^2(A_\sigma)$ in $(-\infty, m_1-\delta/2)$
and that $H^\sigma$ is locally of class $C^2(A^\sigma)$ in
$(-\infty, m_1-\delta/2)$, for $g\leq g_1$. This ends the proof of
Theorem~\ref{thm:3.7}. \qed


\section{Proof of Theorem~\ref{thm:3.4}}
\setcounter{equation}{0}

By \eqref{eq:3.27-ter}, $\cup_{n\geq1} \left(
(\gamma-\epsilon_\gamma)^2\sigma_n,\,
(\gamma+\epsilon_\gamma)\sigma_n)\right)$ is a covering by open
sets of any compact subset of $(E,\, m_1-\delta]$ and of the
interval $(E,\, m_1-\delta]$ itself. Theorem~\ref{thm:3.4}~(i) and
(ii) follow from Theorems~0.1 and 0.2 in \cite{Sahbani1997} and
Theorems~\ref{thm:3.7} and \ref{thm:3.10} above with $g_\delta =
\tilde g_\delta^{(2)}$, where $\tilde g_\delta^{(2)}$ is given in
Theorem~\ref{thm:3.10}. Theorem~\ref{thm:3.4}~(iii) follows from
Theorem~25 in \cite{Mourre1981}.


\begin{appendix}
\section{}\label{appendix}
\setcounter{equation}{0}

In this appendix, we will prove Proposition~\ref{proposition:3.5}.
We apply the method developed in \cite{Bachetal2006} because every
infrared cutoff Hamiltonian that one considers has a ground state
energy which is a simple eigenvalue.

Let, for $n\geq 0$,
\begin{equation}\nonumber
\begin{split}
 & \gF^{\sigma_n} = \gF^n \, , \\
 & \Snn = \Sigma_1 \cap \{p_2;\
 \sigma_{n+1}\leq |p_2| < \sigma_n\} \ , \\
 & \Fldnn = \gF_a (L^2(\Snn))\otimes \gF_a(L^2(\Snn)) \, , \\
 & \Fnn = \otimes_{\ell=1}^3 \, \Fldnn  .
\end{split}
\end{equation}


We have
\begin{equation}\nonumber
  \gF^{n+1} \simeq \gF^n \otimes \Fnn\ .
\end{equation}
Let $\Omega^n$ (respectively $\Omega_n^{n+1}$) be the vacuum state
in $\gF^n$ (respectively in $\Fnn$). We now set
\begin{equation}\nonumber
 H_{0\, n}^{\ \, n+1} = H_0^{(1)} + H_0^{(3)} + \sum_{\ell=1}^3
 \sum_{\epsilon=\pm} \int_{\sigma_{n+1}\leq |p_2| < \sigma_n}
 \!\!\! w_\ell^{(2)}(\xi_2) c_{\ell,\epsilon}^*
 (\xi_2) c_{\ell,\epsilon}(\xi_2) \d \xi_2\ .
\end{equation}
The operator $H_{0\,n}^{\ \, n+1}$ is a self-adjoint operator in
$\Fnn$.

Let us denote by $H_I^n$ and $H_{I\, n}^{\ \, n+1}$ the
interaction $H_I$ given by \eqref{eq:2.35}-\eqref{eq:2.37} but
associated with the following kernels
\begin{equation}\nonumber
 \tilde\chi^{\sigma_n}(p_2) G^{(\alpha)}_{\ell,\epsilon,\epsilon'}
 (\xi_1, \xi_2,\xi_3)\ ,
\end{equation}
and
\begin{equation}\nonumber
 (\tilde\chi^{\sigma_{n+1}}(p_2) - \tilde\chi^{\sigma_n}(p_2))
 G^{(\alpha)}_{\ell,\epsilon,\epsilon'}(\xi_1,\xi_2,\xi_3)\ ,
\end{equation}
respectively, where $\tilde\chi^{\sigma_{n+1}}$ is defined by
\eqref{def:chitilde2}.

Let for $n\geq 0$,
\begin{equation}\nonumber
\begin{split}
 & H_+^n = H^n - E^n\ ,\\
 & \tilde H_+^{n} = H_+^n \otimes \1_n^{n+1} + \1_n \otimes H_{0\,
 n}^{\ \, n+1}\ .
\end{split}
\end{equation}
The operators $H_+^n$ and $\tilde H_+^n$ are self-adjoint
operators in $\gF^n$ and $\gF^{n+1}$ respectively. Here $\1^n$ and
$\1_n^{n+1}$ are the identity operators in $\gF^n$ and $\Fnn$
respectively.

Combining \eqref{eq:2.79} and \eqref{eq:2.80} with
\eqref{eq:3.33}-\eqref{eq:3.36} we obtain for $n\geq 0$,
\begin{equation}\label{eq:A.14}
 g \| H_I^n \psi\| \leq g K(G) ( C_{\beta\eta} \| H_0\psi\| +
 B_{\beta\eta} \|\psi\|)\ ,
\end{equation}
for every $\psi\in\cD(H_0^n) \subset \gF^n$.

It follows from \cite[\S V, Theorem~4.11]{Kato1966} that
\begin{equation}\nonumber
 H^n \geq -\frac{g K(G) B_{\beta\eta}}{1 - g_1 K(G) C_{\beta\eta}}
 \geq -\frac{ g_1 K(G) B_{\beta\eta}}{1 - g_1 K(G) C_{\beta\eta}}\
 ,
\end{equation}
and
\begin{equation}\nonumber
  E^n \geq -\frac{ g K(G) B_{\beta\eta}}{1 - g_1 K(G)
  C_{\beta\eta}}\ .
\end{equation}

We have
\begin{equation}\label{eq:A.17}
 (\Omega^n,\ H^n\Omega^n )=0\ .
\end{equation}
Therefore
\begin{equation}\nonumber
  E^n \leq 0\ ,
\end{equation}
and
\begin{equation}\label{eq:A.18}
 |E^n| \leq
 \frac{ g K(G) B_{\beta\eta}}{1 - g_1 K(G) C_{\beta\eta}}\ .
\end{equation}

Let
\begin{equation}\label{eq:A.19}
  K_n^{n+1} (G) = K(\1_{\sigma_{n+1} \leq |p_2| \leq 2 \sigma_n}\,
  G)\ .
\end{equation}

Combining \eqref{eq:2.79} and \eqref{eq:2.80} with
\eqref{eq:3.33}, \eqref{eq:3.34} and \eqref{eq:A.19} we obtain for
$n\geq 0$
\begin{equation}\label{eq:A.20}
 g \| H_{I\, n}^{\ \, n+1} \psi \| \leq g\, K_n^{n+1}(G)\,
 (C_{\beta\eta} \| H_0^{n+1} \psi\|
 + B_{\beta\eta} \| \psi \|) \ ,
\end{equation}
for $\psi\in \cD(H_0^{n+1}) \subset \gF^{n+1}$, where we remind
that $H_0^{n+1} = H_0|_{\gF^{\sigma_{n+1}}}$ as defined in
\eqref{def:h0n}.

We have for every $\psi\in\cD(H_0^{n+1})$,
\begin{equation}\label{eq:A.21}
 H_0^{n+1} \psi = \tilde H_+^n \psi + E^n \psi - g(H_I^n \otimes
 \1_n^{n+1})\psi\ ,
\end{equation}
and by \eqref{eq:A.14}
\begin{equation}\label{eq:A.22}
 g \| (H_I^n \otimes \1_n^{n+1}) \psi \| \leq g\, K(G)\,
 (C_{\beta\eta} \| H_0^{n+1} \psi\| +
 B_{\beta\eta} \|\psi\| )\ .
\end{equation}

In view of \eqref{eq:A.18} and \eqref{eq:A.21} it follows from
\eqref{eq:A.22} that
\begin{equation}\label{eq:A.23}
\begin{split}
  & g \| (H_I^n\otimes \1_n^{n+1})\psi\| \\
  & \leq
  \frac{g\, K(G)\, C_{\beta\eta}}{1 - g_1\, K(G)\, C_{\beta\eta}}
  \| \tilde H_+^n \psi\|
  + \frac{g\, K(G)\, B_{\beta\eta}}{1 - g_1\, K(G)\,
  C_{\beta\eta}} \big( 1 +
  \frac{g\, K(G)\, B_{\beta\eta}}{1 - g_1\, K(G)\, C_{\beta\eta}} \big)
  \| \psi\|\ .
\end{split}
\end{equation}
By \eqref{eq:3.38}, \eqref{eq:3.39}, \eqref{eq:A.20},
\eqref{eq:A.21}, \eqref{eq:A.23} we finally get
\begin{equation}\label{eq:A.24}
 g \| H_{I\, n}^{\ \, n+1} \psi\| \leq g K_n^{n+1} (G)
 (\tilde C_{\beta\eta} \| \tilde H_+^n \psi\|
 + \tilde B_{\beta\eta} \|\psi\|)\ .
\end{equation}

For $n\geq 0$, a straightforward computation yields
\begin{equation}\label{eq:A.25}
  K_n^{n+1} (G) \leq \sigma_n \tilde{K}(G)
  \leq  \sup (\frac{4\Lambda\gamma}{2m_1-\delta},\
  1)\, \tilde K(G) \frac{\sigma_{n+1}}{\gamma}\ .
\end{equation}

Recall that for $n\geq 0$,
\begin{equation}\label{eq:A.26}
 \sigma_{n+1} < m_1\ .
\end{equation}

By \eqref{eq:A.24}, \eqref{eq:A.25} and \eqref{eq:A.26}, we get,
for $\psi\in\cD(H_0)$,
\begin{equation}\nonumber
 g\, \|H_{I\, n}^{\ \, n+1} \psi\| \leq g\,
 K_n^{n+1}(G)\, \big(\, \tilde C_{\beta\eta} \| (\tilde H_+^n +
 \sigma_{n+1})\psi\| + (\tilde C_{\beta\eta}\, m_1 + \tilde
 B_{\beta\eta}) \|\psi\|\, \big)\ ,
\end{equation}
and for $\phi\in\gF$,
\begin{equation}\label{eq:B.28}
\begin{split}
  g\| H_{I\, n}^{\ \, n+1} (\tilde H_+^n +
 \sigma_{n+1})^{-1}\phi\|
 & \leq g\, K_n^{n+1}(G)\, \big(\, \tilde C_{\beta\eta} +
 \frac{ m_1 \tilde C_{\beta\eta} + \tilde
 B_{\beta\eta}}{\sigma_{n+1}}\, \big) \|\phi\| \\
 & \leq \frac{g}{\gamma}\,
 \sup(\frac{4\Lambda\gamma}{2m_1-\delta},\, 1)\,
 \tilde K(G) (2m_1 \tilde C_{\beta\eta} + \tilde B_{\beta\eta})
 \|\phi\|\ .
\end{split}
\end{equation}

Thus, by \eqref{eq:B.28}, the operator $H_{I\, n}^{\ \, n+1}
(\tilde H_+^n + \sigma_{n+1})^{-1}$ is bounded and
\begin{equation}\nonumber
 g \| H_{I\, n}^{\ \, n+1} (\tilde
H_+^n + \sigma_{n+1})^{-1} \| \leq g\frac{\tilde D}{\gamma}\ ,
\end{equation}
where $\tilde D$ is given by (see \eqref{def:Dtilde}
\begin{equation}\nonumber
 \tilde D = \, \sup (\frac{4\Lambda\gamma}{2m_1-\delta},\, 1)\,
 \tilde K(G)\,
 (2m_1 \tilde C_{\beta\eta} + \tilde B_{\beta\eta}) .
\end{equation}

This yields, for $\psi\in\cD(\tilde H_+^n)$,
\begin{equation}\nonumber
 g \| H_{I\, n}^{\ \, n+1} \psi\| \leq g \frac{\tilde D}{\gamma}
  \| (\tilde H_+^n + \sigma_{n+1})\psi\| \ .
\end{equation}
Hence it follows from \cite[\S V, Theorems~4.11 and
4.12]{Kato1966} that
\begin{equation}\label{eq:B.22}
  g | (H_{I\, n}^{\ \, n+1}\psi,\, \psi)| \leq g\frac{\tilde
  D}{\gamma} (\, (\tilde H_+^n + \sigma_{n+1})\psi,\,\psi\,)\ .
\end{equation}
Let $g_\delta^{(2)}>0$ be such that
\begin{equation}\nonumber
 g_\delta^{(2)} \frac{\tilde D}{\gamma} <1\quad\mbox{and}\quad
 g_\delta^{(2)} \leq g_\delta^{(1)}\ .
\end{equation}
By \eqref{eq:B.22} we get, for $g\leq g_\delta^{(2)}$,
\begin{equation}\label{eq:B.23}
 H^{n+1} = \tilde H_+^n + E^n + g H_{I\, n}^{\ \, n+1}
 \geq E^n - \frac{g\, \tilde D}{\gamma}\, \sigma_{n+1}
 + (1- \frac{g\, \tilde D}{\gamma}) \tilde H_+^n\ .
\end{equation}
Because $(1- g \tilde D / \gamma) \tilde H_+^n \geq 0$ we get from
\eqref{eq:B.23}
\begin{equation}\label{eq:B.24}
 E^{n+1} \geq E^n - \frac{g\, \tilde D}{\gamma}\, \sigma_{n+1},\ n\geq
 0\ .
\end{equation}
Suppose that $\psi^n\in\gF^n$ satisfies $\|\psi^n\|=1$ and for
$\epsilon>0$,
\begin{equation}\label{eq:B.25}
 (\psi^n,\, H^n \psi^n) \leq E^n + \epsilon \ .
\end{equation}
Let
\begin{equation}\label{eq:B.26}
 \tilde\psi^{n+1} = \psi^n\otimes\Omega_n^{n+1} \in \gF^{n+1}\ .
\end{equation}
We obtain
\begin{equation}\label{eq:B.27}
 E^{n+1} \leq (\tilde\psi^{n+1},\, H^{n+1} \tilde\psi^{n+1})
 \leq E^n + \epsilon + g(\tilde\psi^{n+1},\ H_{I\, n}^{\ \, n+1}
 \,\tilde\psi^{n+1})
\end{equation}
By \eqref{eq:B.22}, \eqref{eq:B.25}, \eqref{eq:B.26} and
\eqref{eq:B.27} we get, for every $\epsilon>0$,
\begin{equation}\nonumber
 E^{n+1} \leq E^n + \epsilon(1 + \frac{g\,\tilde D}{\gamma})
 + \frac{g\, \tilde D}{\gamma} \, \sigma_{n+1}\ ,
\end{equation}
where $g\leq g_\delta^{(2)}$.

This yields
\begin{equation}\label{eq:A.37}
 E^{n+1} \leq E^n + \frac{g\, \tilde D}{\gamma}\, \sigma_{n+1}\ ,
\end{equation}
and by \eqref{eq:B.24}, we obtain
\begin{equation}\nonumber
 |E^n - E^{n+1}| \leq \frac{g\, \tilde D}{\gamma}\, \sigma_{n+1}\
 .
\end{equation}

For $n=0$, since $\sigma_0 = \Lambda$, remind that $H_0^{\, 0} =
H_0^{n=0} = H_0^{\sigma_0} = H_0 |_{\gF^\Lambda}$. Thus, the
ground state energy of $H_0^{\, 0}$ is $0$ and it is a simple
isolated eigenvalue of $H_0^{\, 0}$ with $\Omega^0$, the vacuum in
$\gF^0$, as eigenvector. Moreover, since $\Lambda>m_1$,
\begin{equation}\nonumber
 \inf \left(\sigma(H_0^{\, 0}\right) \setminus \{0\}) = m_1\ ,
\end{equation}
thus $(0,m_1)$ belongs to the resolvent set of $H_0^{\, 0}$.

By Hypothesis~\ref{hypothesis:3.1}(iv) we have $H^0 = H_0^{\, 0}$.
Hence $E^0=\{0\}$ is a simple isolated eigenvalue of $H^0$ and
$H^0 = H_+^{\, 0}$. We finally get
\begin{equation}
 \inf\left( \sigma(H_+^{\, 0}) - \{0\}\right)
 = m_1 > m_1 - \frac{\delta}{2} =\sigma_1\ .
\end{equation}

We now prove Proposition ~\ref{proposition:3.5} by induction in
$n\in\N^*$. Suppose that $E^n$ is a simple isolated eigenvalue of
$H^n$ such that
\begin{equation}\nonumber
  \inf \left( \sigma(H_+^n)\setminus\{0\} \right)
  \geq (1 -\frac{3 g \tilde D}{\gamma})\sigma_n,\quad n\geq 1\ .
\end{equation}
Since \eqref{eq:prop-sigman} gives $\sigma_{n+1} < (1 - \frac{3 g
\tilde D}{\gamma})\sigma_n$ for $g\leq g_\delta^{(2)}$, $0$ is
also a simple isolated eigenvalue of $\tilde{H}_+^n$ such that
\begin{equation}\label{eq:A.47}
  \inf \left( \sigma(\tilde H_+^n)\setminus\{0\} \right)
  \geq \sigma_{n+1}\ .
\end{equation}
We must now prove that $E^{n+1}$ is a simple isolated eigenvalue
of $H^{n+1}$ such that
\begin{equation}\nonumber
 \inf \left( \sigma(H_+^{n+1}) \setminus\{0\}\right)
 \geq (1-\frac{3 g \tilde D}{\gamma})\sigma_{n+1}\ .
\end{equation}
Let
\begin{equation}\nonumber
 \lambda^{(n+1)} = \sup_{\psi\in\gF^{n+1};\, \psi\neq 0}
 \ \ \inf_{(\phi,\psi)=0;\, \phi\in\cD(H^{n+1});\, \|\phi\|=1}
 (\phi,\, H_+^{n+1}\phi)\ .
\end{equation}
By \eqref{eq:B.23} and \eqref{eq:A.37}, we obtain, in $\gF^{n+1}$
\begin{equation}\label{eq:A.50}
\begin{split}
 H_+^{n+1} & \geq E^n - E^{n+1}-\frac{g \tilde D}{\gamma}\sigma_{n+1}
 + (1- \frac{g \tilde D}{\gamma})\tilde H_+^{n}\\
 & \geq (1-\frac{g \tilde D}{\gamma}) \tilde H_+^n -
  \frac{2 g \tilde D}{\gamma} \sigma_{n+1}\ .
\end{split}
\end{equation}
By \eqref{eq:B.26}, $\tilde\psi^{n+1}$ is the unique ground state
of $\tilde H_+^n$ and by \eqref{eq:A.47} and \eqref{eq:A.50}, we
have, for $g\leq g_\delta^{(2)}$,
\begin{equation}\nonumber
\begin{split}
 \lambda^{(n+1)} & \geq \inf_{(\phi,\tilde\psi^{n+1})=0;\,
 \phi\in\cD(H^{n+1});\, \|\phi\|=1} (\phi, H_+^{n+1}\phi) \\
  & \geq (1-\frac{g\tilde D}{\gamma})\sigma_{n+1} - \frac{2 g
  \tilde D}{\gamma} \sigma_{n+1} = (1- \frac{3 g \tilde
  D}{\gamma})
  \sigma_{n+1} >0\ .
\end{split}
\end{equation}
This concludes the proof of Proposition~\ref{proposition:3.5} by
choosing $g_\delta = g_\delta^{(2)}$, if one proves that $H^1$
satisfies Proposition~\ref{proposition:3.5}. By noting that $0$ is
a simple isolated eigenvalue of $\tilde H_+^0$ such that
$\inf(\sigma(\tilde H_+^0)\setminus\{0\} ) =\sigma_1$, we prove
that $E^1$ is indeed an isolated simple eigenvalue of $H^1$ such
that $\inf(\sigma(H_+^1)\setminus\{0\}) \geq (1-\frac{3 g \tilde
D}{\gamma}) \sigma_1$ by mimicking the proof given above for
$H_+^{n+1}$. \par\qed
\end{appendix}



\end{document}